\documentclass[11pt]{article}
\usepackage[margin=1in]{geometry}

\usepackage{xcolor}
\usepackage[colorlinks=true,citecolor=blue,linkcolor=blue]{hyperref}

\hypersetup{
    colorlinks,
    linkcolor={red!50!black},
    citecolor={blue!50!black},
    urlcolor={blue!80!black}
}

\usepackage[]{amsmath,amssymb,amsfonts,latexsym,amsthm,enumerate,fullpage,xcolor,bbm}
\usepackage[nameinlink, noabbrev, capitalize]{cleveref}
\crefname{claim}{Claim}{Claims}

\usepackage{nicefrac}

\usepackage[utf8]{inputenc}
\usepackage{amsmath,amssymb}
\usepackage{graphics,graphicx,tikz,color,float}
\usepackage{mathtools}
\usepackage{amsfonts}
\usepackage{tikz}
\usepackage{comment}
\usepackage{framed}
\usepackage[section]{placeins}
\usepackage{algorithm}
\usepackage[noend]{algpseudocode}
\usepackage{xcolor}

\usepackage{thm-restate}
\newtheorem{theorem}{Theorem}

\newtheorem{lemma}{Lemma}

\newtheorem{definition}{Definition}
\newtheorem{remark}{Remark}

\newcommand{\bits}{\{0,1\}}
\newcommand{\thetatamp}{\theta^{\mathsf{tamp}}}
\newcommand{\thetasame}{\theta^{\mathsf{same}}}

\newcommand{\rhotamp}{\rho^{\mathsf{tamp}}}
\newcommand{\rhosame}{\rho^{\mathsf{same}}}

\newtheorem{fact}{Fact}

\newcommand{\beq}{\begin{equation}}
	\newcommand{\enq}{\end{equation}}
\newcommand{\bel}{\begin{lemma}}
	\newcommand{\enl}{\end{lemma}}
\newcommand{\bet}{\begin{theorem}}
	\newcommand{\ent}{\end{theorem}}

\newcommand{\tr}{\mathrm{Tr}}

\newcommand{\E}{\mathbb{E}}
\newcommand{\ketbra}[1]{|#1\rangle\langle#1|}

\newcommand{\eps}{\varepsilon}

\newcommand{\poly}{\textnormal{poly}}
\newcommand{\fid}{\mathsf{F}}

\newcommand*{\cC}{\mathcal{C}}
\newcommand*{\cA}{\mathcal{A}}

\newcommand*{\cH}{\mathcal{H}}

\newcommand*{\cD}{\mathcal{D}}

\newcommand*{\cO}{\mathcal{O}}

\newcommand*{\cX}{\mathcal{X}}

\newcommand*{\cZ}{\mathcal{Z}}
\newcommand*{\cE}{\mathcal{E}}

\newcommand{\cP}{\mathcal{P}}

\newcommand*{\IP}{\mathsf{IP}}

\newcommand{\supp}{\mathrm{supp}}
\newcommand{\suppress}[1]{}

\newcommand{\defeq}{\ensuremath{ \stackrel{\mathrm{def}}{=} }}
\newcommand{\F}{\mathbb{F}}

\newcommand {\br} [1] {\ensuremath{ \left( #1 \right) }}

\newcommand {\minusspace} {\: \! \!}

\newcommand {\fn} [2] {\ensuremath{ #1 \minusspace \br{ #2 } }}

\newcommand {\dmax} [2] {\fn{\mathrm{D}_{\max}}{#1 \middle\| #2}}

\newcommand {\imax}{\ensuremath{\mathrm{I}_{\max}}}

\newcommand {\hmin} [2] {\fn{ \mathrm{H }_{\min}}{#1 \middle | #2}}

\newcommand {\id} {\ensuremath{\mathbb{I}}}

\newcommand{\samp}{\mathsf{Samp}}

\usepackage{afterpage}

\newcommand{\auth}{\mathrm{Auth}}
\newcommand{\ver}{\mathrm{Ver}}

\newcommand{\epspriv}{\eps_{\mathsf{priv}}}
\newcommand{\epsnm}{\eps_{\mathsf{nm}}}
\newcommand{\epslk}{\eps_{\mathsf{leak}}}

\newcommand*{\sm}{\mathsf{same}}
\newcommand*{\qpas}{\mathsf{qpa\mhyphen state}}

\newcommand*{\nmext}{\mathsf{nmExt}}

\newcommand{\epr}{\mathsf{epr}}
\newcommand{\X}{\mathcal{X}}

\newcommand{\Leak}{\mathsf{Leak}}

\newcommand*{\cL}{\mathcal{L}}

\newcommand{\bra}[1]{\langle #1|}
\newcommand{\ket}[1]{|#1 \rangle}

\mathchardef\mhyphen="2D

\newcommand*{\enc}{\mathrm{Enc}}
\newcommand*{\qenc}{\mathrm{2NMC}}
\newcommand*{\qdec}{\mathrm{2NMD}}
\newcommand*{\lrenc}{\mathrm{lrShare}}
\newcommand*{\qshare}{\mathrm{qShare}}
\newcommand*{\qrec}{\mathrm{qRec}}

\newcommand*{\nmenc}{\mathrm{nmShare}}
\newcommand*{\dec}{\mathrm{Dec}}
\newcommand*{\lrdec}{\mathrm{lrRec}}
\newcommand*{\nmdec}{\mathrm{nmRec}}
\newcommand*{\cenc}{\mathrm{cEnc}}
\newcommand*{\cdec}{\mathrm{cDec}}
\newcommand*{\nmcenc}{\mathrm{\nmext}}
\newcommand*{\nmcdec}{\mathrm{\nmext}}

\newcommand{\nmshare}{\mathrm{nmShare}}
\newcommand{\share}{\mathrm{Share}}
\newcommand{\nmrec}{\mathrm{nmRec}}
\newcommand{\rec}{\mathrm{Rec}}

\newcommand*{\cSC}{\mathcal{SC}}

\makeatletter
\newcommand*{\rom}[1]{\expandafter\@slowromancap\romannumeral #1@}
\makeatother

\mathchardef\mhyphen="2D

\newfloat{Protocol}{tbhp}{lop}

\title{Split-State Non-Malleable Codes and Secret Sharing Schemes for 
\\ Quantum Messages\footnote{This work was presented as a contributed talk at QCRYPT 2023.}}
\author{
Naresh Goud Boddu\footnote{NTT Research, \texttt{naresh.boddu@ntt-research.com}.} \and
	Vipul Goyal\footnote{Carnegie Mellon University and NTT Research, \texttt{vipul@vipulgoyal.org}.} \and
 Rahul Jain\footnote{Centre for Quantum Technologies and Department of Computer Science, 
  National University of Singapore and MajuLab, UMI 3654, Singapore,  \texttt{rahul@comp.nus.edu.sg}.} \and 
 João Ribeiro\footnote{NOVA LINCS and NOVA School of Science and Technology, \texttt{joao.ribeiro@fct.unl.pt}.}
}
\date{}

\begin{document}
\maketitle

\begin{abstract}

Non-malleable codes are fundamental objects at the intersection of cryptography and coding theory. 
These codes provide security guarantees even in settings where error correction and detection are impossible, and have found applications to several other cryptographic tasks.
One of the strongest and most well-studied adversarial tampering models is $2$-split-state tampering.
Here, a codeword is split into two parts which are stored in physically distant servers, and the adversary can then independently tamper with each part using arbitrary functions.
This model can be naturally extended to the secret sharing setting with several parties by having the adversary independently tamper with each share.
Previous works on non-malleable coding and secret sharing in the split-state tampering model only considered the encoding of \emph{classical} messages.
Furthermore, until recent work by Aggarwal, Boddu, and Jain (IEEE Trans.\ Inf.\ Theory 2024 \& arXiv 2022), adversaries with quantum capabilities and \emph{shared entanglement} had not been considered, and it is a priori not clear whether previous schemes remain secure in this model.

In this work, we introduce the notions of split-state non-malleable codes and secret sharing schemes for quantum messages secure against quantum adversaries with shared entanglement. Then, we present explicit constructions of such schemes that achieve low-error non-malleability. 
More precisely, we construct efficiently encodable and decodable split-state non-malleable codes and secret sharing schemes for quantum messages preserving entanglement with external systems and achieving security against quantum adversaries having shared entanglement with codeword length $n$, any message length at most $n^{\Omega(1)}$, and error $\eps=2^{-{n^{\Omega(1)}}}$.
In the easier setting of \emph{average-case} non-malleability, we achieve efficient non-malleable coding with rate close to $1/11$.
    
\end{abstract}

\newpage

\tableofcontents

\newpage

\section{Introduction}
\label{sec:intro}

Non-Malleable Codes (NMCs), introduced in the work of Dziembowski, Pietrzak, and Wichs \cite{DPW10}, are now considered fundamental cryptographic primitives that provide security guarantees even in adversarial settings where error correction and detection are impossible. Informally, NMCs guarantee that an adversary cannot change the encoding of a message into that of a related message. They encode a classical message $M$ into a codeword $C$ in such a way that tampering $C$ into $f(C)$ using an allowed tampering function $f$ results in the decoder either outputting the original message $M$ or a message that is unrelated/independent of $M$.
Note that it is impossible to construct NMCs that protect against \emph{arbitrary} tampering functions. This is because an adversary could simply apply the decoder to $C$, recovering the message $M$, and then output the encoding of a related message $M+1$ as the tampered version of $C$. Therefore, previous work on non-malleable coding has focused on constructing NMCs for restricted, but still large and meaningful, classes of tampering functions.

One of the strongest and most studied tampering models is \emph{split-state tampering}, introduced by Liu and Lysyanskaya~\cite{LL12}. In the $2$-state version of this model (which is the hardest), we view a codeword $C$ as being composed of two parts, $E_1$ and $E_2$, and an adversary is allowed to \emph{independently} tamper with each part using an arbitrary tampering function. In other words, a split-state tampering adversary consists of a pair of arbitrary functions $(f,g)$, and a codeword $(E_1,E_2)$ is tampered to $(f(E_1),g(E_2))$ (see \cref{fig:splitstate121} for a diagram of this model). The split-state model is meaningful because we can imagine that the two parts of the codeword are stored in different physically isolated servers, making communication between tampering adversaries infeasible.

The notion of non-malleable secret sharing has also been widely studied as a strengthening of non-malleable codes. Secret sharing, dating back to the work of Blakley~\cite{Bla79} and Shamir~\cite{Sha79}, is a fundamental cryptographic primitive where a dealer encodes a secret into $p$ shares and distributes them among $p$ parties. Each secret sharing scheme has an associated \emph{monotone}\footnote{A set $\Gamma\subseteq 2^{[p]}$ is \emph{monotone} if $A\in\Gamma$ and $A\subseteq B$ imply that $B\in\Gamma$.} set $\Gamma\subseteq 2^{[p]}$, usually called an \emph{access structure}, whereby any set of parties $T\in\Gamma$, called \emph{authorized} sets, can reconstruct the secret from their shares, but any \emph{unauthorized} set of parties $T\not\in\Gamma$ gains essentially no information about the secret. One of the most natural and well-studied types of access structures are \emph{threshold} access structures, where a set of parties $T$ is authorized if and only if $|T|\geq t$ for some threshold $t$.

Non-Malleable Secret Sharing (NMSS), generalizing non-malleable coding, was introduced by Goyal and Kumar~\cite{GK16} and has received significant interest in the past few years in the classical setting. NMSS schemes additionally guarantee that an adversary who is allowed to tamper all the shares (according to some restricted tampering model) cannot make an authorized set of parties reconstruct a different but related secret. 
Non-malleable secret sharing is particularly well-studied in the context of the split-state tampering model described above, whereby an adversary can independently tamper with each share. Once again, the motivation is that shares are being held in physically distant devices, making communication between the several tampering adversaries infeasible.

The split-state and closely related tampering models for codes and secret sharing schemes have witnessed a flurry of work in the past decade \cite{LL12,DKO13,CG14a,CG14b,ADL17,CGL15,li15,Li17,GK16,GK18,ADNOP19,BS19,FV19,Li19,AO20,BFOSV20,BFV21,GSZ21,AKOOS22,CKOS22,Li23}, culminating in recent explicit constructions of classical split-state NMCs of rate $1/3$~\cite{AKOOS22} and constructions with a smaller constant rate but also smaller error~\cite{Li23} (with a known upper bound on the rate being $1/2$~\cite{CG14a}).
Split-state NMCs and NMSS schemes and related notions have also found applications in other cryptographic tasks, such as non-malleable commitments, secure message transmission, and non-malleable signatures~\cite{GPR16,GK16,GK18,ADNOP19,AP19,CKOS22}.

\paragraph{Split-state tampering and quantum computing.}

Given the rapid development of quantum technologies, it is natural to consider NMCs and NMSS schemes in the quantum setting and examine how adversaries with quantum capabilities affect previous assumptions made about tampering models. 
For instance, all known NMCs are tailored to classical messages, but it is equally important to design non-malleable coding schemes that allow us to encode quantum states as well. 
Moreover, the possibility of attackers with quantum capabilities challenges the \emph{independence assumption} made in the split-state models described earlier. Although servers holding different parts of the codeword (or different secret shares) may be physically isolated from each other, the tampering adversaries attacking each server may have pre-shared a large amount of entangled quantum states. Access to such shared entanglement can provide non-trivial advantages to these adversaries beyond what has been considered in the classical tampering models. 
For example, in the Clauser-Horne-Shimony-Holt (CHSH) game~\cite{CHSH69}, non-communicating parties can use local measurements on both halves of an EPR state to achieve a higher success probability than what is possible using fully classical strategies. 
Therefore, it is not clear whether any of the existing classical NMCs and NMSS schemes remain secure in the augmented split-state tampering models where the adversary is allowed to make use of arbitrary shared entanglement across multiple states (see \cref{fig:splitstate211} for a diagram of the split-state tampering model for two states). 
The only exception to this are the recent work of Aggarwal, Boddu, and Jain~\cite{ABJ22} and the concurrent work of Batra, Boddu, and Jain~\cite{BBJ23}.
The former constructs explicit NMCs for \emph{classical} messages that are secure in the $2$-split-state tampering model, while the latter focuses on explicit quantum-secure~\emph{non-malleable randomness encoders} (NMRE) with a higher rate in the same tampering model, and uses these objects to non-malleable codes in the $3$-split-state model\footnote{Meaning that the codeword is divided into three parts and the adversary tampers each part independently.} for quantum messages.

The shortcomings of existing split-state non-malleable coding schemes in the face of quantum messages and adversaries raise the following natural question:
\begin{quote}
\begin{center}
Can we design efficient $2$-split-state NMCs and split-state NMSS schemes for \emph{quantum messages} secure against quantum adversaries with shared entanglement?
\end{center}
\end{quote}
We resolve this question in the affirmative.

\begin{figure}
		\centering
  \resizebox{10cm}{5cm}{
	\begin{tikzpicture}
		
		\node at (1.8,1.5) {$M$};
		\node at (11,1.5) {$M'$};
		
		\draw (2,1.5) -- (3,1.5);
		\draw (3,-0.5) rectangle (4.5,3.5);
		\node at (3.8,1.5) {$\enc$};
		
		
		
		\node at (4.7,3) {$E_1$};
		\node at (8.1,3) {$E_1$};
		\draw (4.5,2.8) -- (6,2.8);
		\draw (7,2.8) -- (8.5,2.8);
		\draw (10,1.5) -- (10.7,1.5);
		
		\node at (4.7,0.4) {$E_2$};
		\node at (8.1,0) {$E_2$};
		\draw (4.5,0.2) -- (6,0.2);
		\draw (7,0.2) -- (8.5,0.2);
		
		\draw (6,2.2) rectangle (7,3.2);
		\node at (6.5,2.7) {$f$};
		\draw (6,-0.2) rectangle (7,0.8);
		\node at (6.5,0.3) {$g$};
		
		\node at (6.5,-0.5) {$\mathcal{A}=(f,g)$};
		\draw (5.2,-0.8) rectangle (7.8,3.8);

		
		
		\draw (8.5,-0.5) rectangle (10,3.2);
		\node at (9.2,1.5) {$\dec$};
		
	\end{tikzpicture} }
	\caption{Classical split-state tampering model.}\label{fig:splitstate121}
\end{figure}
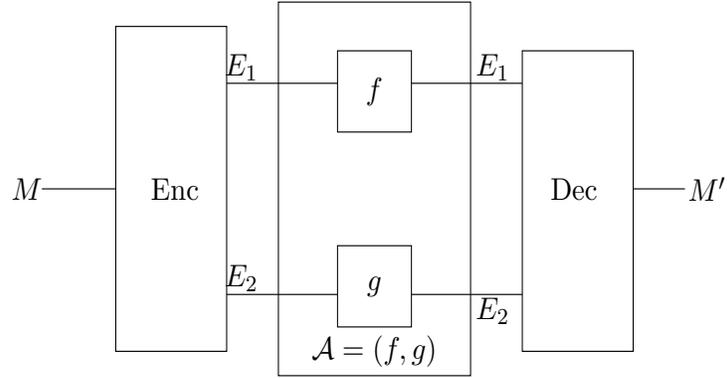

\begin{figure}
		\centering
   \resizebox{10cm}{5cm}{
		\begin{tikzpicture}
			

     \draw (2,4.3) -- (11,4.3);
     \node at (1.8,4.3) {$\hat{M}$};
      \node at (11.2,4.3) {$\hat{M}$};
			\node at (1.8,1.5) {$M$};
			\node at (11,1.5) {$M$};

\draw (1.8,2.8) ellipse (0.3cm and 1.8cm);
   
			
			\draw (2,1.5) -- (3,1.5);
			\draw (3,-0.5) rectangle (4.5,3.5);
			\node at (3.8,1.5) {$\enc$};
			
			
			
			\node at (4.7,3) {$E_1$};
			\node at (8.25,3) {$E_1$};
  
			\draw (4.5,2.8) -- (6,2.8);
			\draw (7,2.8) -- (8.6,2.8);
			\draw (10,1.5) -- (10.7,1.5);
			
			\node at (4.7,0.4) {$E_2$};
			\node at (8.25,0) {$E_2$};
			\draw (4.5,0.2) -- (6,0.2);
			\draw (7,0.2) -- (8.6,0.2);
			
			\draw (6,2) rectangle (7,3);
			\node at (6.5,2.5) {$U$};
			\draw (6,-0) rectangle (7,1);
			\node at (6.5,0.5) {$V$};
               \node at (6.5,1.5) {$\ket{\psi}_{W_1W_2}$};
			
			\node at (6.5,-0.4) {$\mathcal{A}=(U,V,\ket{\psi}_{})$};
			\draw (5,-0.8) rectangle (8,3.8);

			\draw (5.5,1.5) ellipse (0.3cm and 1.1cm);
			\node at (5.5,2) {$W_1$};
				\draw (5.7,2.2) -- (6,2.2);
			\node at (7.4,2) {$W_1$};
				\draw (7,2.2) -- (7.2,2.2);
				\node at (5.5,1) {$W_2$};
				\draw (5.7,0.7) -- (6,0.7);
			\node at (7.4,0.9) {$W_2$};
			\draw (7.0,0.7) -- (7.2,0.7);
			
			
			\draw (8.6,-0.5) rectangle (10,3.2);
			\node at (9.2,1.5) {$\dec$};
			
		\end{tikzpicture} }
		\caption{Split-state tampering model with shared entanglement. This shared entanglement is stored in registers $W_1$ and $W_2$.}\label{fig:splitstate211}
	\end{figure}
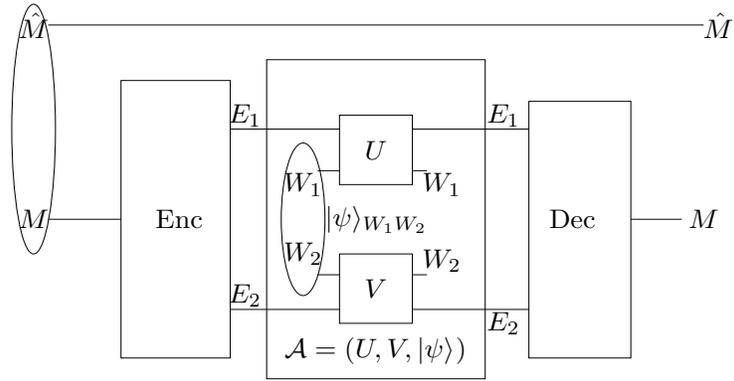

\subsection{Our contributions}

\subsubsection{Split-state non-malleable codes for quantum messages}\label{sec:nmcintro}

As our first contribution, we propose a definition of split-state non-malleability for quantum messages against adversaries with shared entanglement. Our definition is a natural extension of the one considered for classical messages in the literature~\cite{DPW10}. Here, we present it specifically for the $2$-split-state case, which is our main setting of interest.

Let $\sigma_M$ be an arbitrary state in a message register $M$, and $\sigma_{M\hat{M}}$ be its canonical purification.
We consider a ($2$-split-state) coding schemes given by an encoding Completely Positive Trace-Preserving (CPTP) map $\enc : \cL(\cH_M) \to \cL(\cH_{E_1} \otimes \cH_{E_2})$ and a decoding CPTP map $\dec : \cL(\cH_{E_1} \otimes \cH_{E_2}) \to \cL(\cH_M)$, where $\cL(\cH)$ is the space of all linear operators in the Hilbert space $\cH$.

The most basic property we require of this coding scheme $(\enc, \dec)$ is correctness (which includes preserving entanglement with external systems), i.e.,
\begin{equation*}
\dec(\enc(\sigma_{M\hat{M}})) = \sigma_{M\hat{M}},
\end{equation*}
where we use the shorthand $T$ to represent the CPTP map $T \otimes \id$ whenever the action of the identity operator $\id$ is clear from context.

Before we proceed to define split-state non-malleability, we describe the split-state adversarial tampering model in the quantum setting.
Let $\rho_{E_1 E_2} = \enc(\sigma_M)$ be the split-state encoding of message $\sigma_M$. A split-state tampering adversary $\mathcal{A}$ is specified by two tampering maps\footnote{Tampering maps are assumed to be unitary without any loss of generality. This is because, in the presence of unbounded arbitrary shared entanglement, tampering with unitary maps is equivalent to tampering with CPTP maps. More precisely, consider a tampering adversary that uses two CPTP maps $\Phi_1$ and $\Phi_2$ acting on registers $E_1 W_1$ and $E_2 W_2$, respectively. Then, the action of this adversary is equivalent to another adversary who tampers using Stinespring isometry extensions $U$ and $V$ of $\Phi_1$ and $\Phi_2$, respectively, which act on $E_1 W_1 A_1$ and $E_2 W_2 A_2$, respectively, where $A_1$ and $A_2$ are unentangled ancilla registers set to $\ket{0}$ without loss of generality and can be seen as part of the shared entanglement.} $U : \cL(\cH_{E_1}\otimes \cH_{W_1}) \to \cL(\cH_{E_1}\otimes\cH_{W_1})$ and $V : \cL(\cH_{E_2}\otimes \cH_{W_2}) \to \cL(\cH_{E_2}\otimes \cH_{W_2})$ along with a quantum state $\ket{\psi}_{W_1W_2}$ that captures the shared entanglement between the non-communicating tampering adversaries. Finally, the decoding procedure $\dec$ is applied to the tampered codeword. \cref{fig:splitstate211} presents a diagram of this tampering model.
Let
\begin{equation*}
\eta = \dec\left( ( U \otimes V) \left(\enc( \sigma_{M \hat{M}}) \otimes \ket{\psi}\bra{\psi} \right) ( U^\dagger \otimes V^\dagger) \right)
\end{equation*}
be the final state after applying the split-state tampering adversary $\mathcal{A}$ followed by the decoding procedure.

We are now ready to define split-state non-malleability of the coding scheme $(\enc,\dec)$.
\begin{definition}[Worst-case and average-case non-malleable codes for quantum messages]\label{def:qnmcodesfinaldef}  
We say that the coding scheme
$(\enc, \dec)$ is a \emph{(worst-case) $\eps$-non-malleable code for quantum messages} if for every split-state adversary $\cA=(U,V,\ket\psi_{W_1 W_2})$ and every quantum message $\sigma_M$ (with canonical purification $\sigma_{M\hat{M}}$)
    it holds that
    \begin{equation}\label{eq:splitnmc}
        \eta_{M\hat{M}} \approx_\eps p_{\cA} \sigma_{M \hat{M}}+(1-p_{\cA})\gamma^{\cA}_{M}\otimes \sigma_{\hat{M}},
    \end{equation}
where $p_{\cA}\in[0,1]$ and $\gamma^{\cA}_{M}$ depend only\footnote{By this, we mean that $p_{\cA}$ can be computed and the state $\gamma^{\cA}_{M}$ can be prepared without the knowledge of the input message $\sigma_{M\hat{M}}$.} on the split-state adversary $\cA$, and $\approx_\eps$ denotes that the two states are $\eps$-close in trace distance.

If \cref{eq:splitnmc} is only guaranteed to hold when $\sigma_{M}$ is the maximally mixed state, then we say that $(\enc,\dec)$ is an \emph{average-case $\eps$-non-malleable code for quantum messages}.
\end{definition}

\begin{remark}
    \em

    Intuitively, our definition of average-case non-malleability for quantum messages in \cref{def:qnmcodesfinaldef} is analogous to requiring that the \emph{average} non-malleability error of a given classical code is small when averaged over a uniformly random message. 
    Later, in \cref{thm:avgtoworst}, we show that every average-case non-malleable code for quantum messages is also a worst-case non-malleable code, though with a larger error.

\end{remark}

\cref{def:qnmcodesfinaldef} can be readily extended to encompass arbitrary classes of tampering adversaries. However, for the sake of readability, we do not provide the generalization here. In the context of split-state tampering, we present the following two results.

Our first result gives an explicit average-case $2$-split-state NMC for quantum messages with rate arbitrarily close to $1/11$.
\begin{theorem}[Average-case $2$-split-state NMC for quantum messages with constant rate]\label{thm:mainqnmc1}
For any fixed constant $\delta>0$ there exist an integer $n_0>0$ and $c\in(0,1)$ such that the following holds:
There exists a family of coding schemes $(\cC_n)_{n\in\mathbb{N}}$ where each $\cC_n$ has codeword length $n$ and message length $\lfloor \left(\frac{1}{11}-\delta\right)n\rfloor$ such that $C_n$ is average-case $\eps$-non-malleable for quantum messages with error $\eps = 2^{-n^{c}}$ for all integers $n\geq n_0$.
Furthermore, there exist encoding and decoding procedures for the family $(\cC_n)_{n\in\mathbb{N}}$ running in time $\poly(n)$.
\end{theorem}

Our second result, which builds on \cref{thm:mainqnmc1}, gives an explicit construction of a worst-case $2$-split-state NMC for quantum messages.
\begin{theorem}[Worst-case $2$-split-state NMC for quantum messages]\label{thm:mainqnmc}
    There exist constants $c\in(0,1)$ and $n_0\in\mathbb{N}$ such that the following holds:
    There exists a family of coding schemes $(\cC_n)_{n\in\mathbb{N}}$ where each $\cC_n$ has codeword length $n$ and message length $\lfloor n^{c}\rfloor$ such that $\cC_n$ is $\eps$-non-malleable for quantum messages with error $\eps=2^{-n^{c}}$ for all integers $n\geq n_0$.
    Furthermore, there exist encoding and decoding procedures for the family $(\cC_n)_{n\in\mathbb{N}}$ running in time $\poly(n)$.
\end{theorem}

In fact, we show something stronger: The explicit code from \cref{thm:mainqnmc} is actually a $2$-out-of-$2$ non-malleable secret sharing scheme for quantum messages with share size $n$, any message of length at most $n^{\Omega(1)}$, and error $\eps=2^{-n^{\Omega(1)}}$.
We refer the reader to \cref{sec:intronmss} for more details on non-malleable secret sharing.

\subsubsection{Split-state non-malleable secret sharing schemes for quantum messages}\label{sec:intronmss}

We begin by defining threshold non-malleable secret sharing schemes for quantum messages. This is an extension of our definition of $2$-split-state non-malleable codes in \cref{sec:nmcintro}, and it is analogous to the definition in the classical setting~\cite{GK16}. We focus on threshold schemes for simplicity, but the definition can be easily generalized to other access structures.

Let $\sigma_M$ be an arbitrary state in a message register $M$, and $\sigma_{M\hat{M}}$ be its canonical purification. An NMSS scheme for $p$ parties is composed of two CPTP maps $(\nmshare, \nmrec)$. $\nmshare$ is a sharing procedure $\nmshare : \cL(\cH_M) \to \cL(\cH_{S_1}\otimes \cH_{S_2} \otimes \cdots \otimes \cH_{S_p})$, where $S_i$ is the share register for the $i$-th party, and $\nmrec$ is a reconstruction procedure $\nmrec : \cL( \bigotimes_{i \in T}  \cH_{S_i}) \to \cL(\cH_{M})$, where $\cL(\cH)$ is the space of all linear operators in the Hilbert space $\cH$.
The reconstruction procedure $\nmrec$ acts on any authorized subset of shares $T$ to reconstruct the original message.

Before we proceed to define non-malleability, we describe the adversarial tampering model for secret sharing. 
This is a simple extension of the $2$-split-state tampering model described in \cref{sec:nmcintro}.
Fix an authorized subset $T$ of parties. 
A tampering adversary $\cA_T$ is specified by $|T|$ tampering maps\footnote{Tampering maps are assumed to be unitary without any loss of generality.}
$U_i:\cL(\cH_{S_i}\otimes \cH_{W_i})\to\cL(\cH_{S_i}\otimes\cH_{W_i})$ for $i \in  T$ along with a quantum state $\ket\psi_{W_1 W_2 \dots W_p}$ which captures the shared entanglement between non-communicating local tampering adversaries. We denote the overall tampered state by $\tau^{\cA_T}$, which may be written as
\begin{equation*}
    \tau^{\cA_T} = \left(\left(\bigotimes_{i\in T}U_i\right)\otimes\left(\bigotimes_{j\not\in T} \id\right)\right)\left( \nmshare(\sigma_{M \hat{M}}) \otimes \ketbra{\psi}_{W_1 W_2 \dots W_p} \right)\left(\left(\bigotimes_{i\in T}U^\dagger_i\right)\otimes\left(\bigotimes_{j\not\in T} \id\right)\right).
\end{equation*}
Let the final state of the tampering experiment be
\begin{equation*}
    \eta = \nmrec\left(\tau^{\cA_T}_{S_T}\right).
\end{equation*}
We are now ready to define split-state NMSS schemes, in a similar manner to \cref{def:qnmcodesfinaldef} for split-state NMCs.
\begin{definition}[Threshold non-malleable secret sharing scheme for quantum messages]\label{def:qnmss}
The coding scheme $(\nmshare,\nmrec)$ is said to be a \emph{$t$-out-of-$p$ $(\epspriv,\epsnm)$-non-malleable secret sharing scheme for quantum messages} if for any quantum message ${\sigma_M} \in \cD(\cH_M)$ (with canonical purification  $\sigma_{M\hat{M}}$) the following properties are satisfied:
\begin{itemize}
    \item \textbf{Correctness:} For any $T \subseteq [p]$ such that $ \vert T \vert \geq t$  it holds that
    \begin{equation*}
        \nmrec (\nmshare( \sigma_{M \hat{M}})_{S_{T}}) = \sigma_{M\hat{M}},
    \end{equation*}
    where we write $S_T=(S_i)_{i\in T}$.

    \item \textbf{Statistical privacy:} For any $T \subseteq [p]$ such that $ \vert T \vert \leq t-1$ it holds that
    \begin{equation*}
        \nmshare( \sigma)_{\hat{M}S_{T}}  \approx_{\epspriv} \sigma_{\hat{M}} \otimes  \zeta_{S_{T}},
    \end{equation*}
    where $\zeta_{S_{T}} $ is a fixed quantum state independent of $\sigma_{M \hat{M}}$.
    
    \item \textbf{Non-malleability:} For any $T \subseteq [p]$ such that $ \vert T \vert = t$, it holds that
    \begin{equation}\label{eq:qnmssdef}
        \eta_{M\hat{M}} \approx_{\epsnm} p_{\cA_T} \sigma_{M \hat{M}}+(1-p_{\cA_T})\gamma^{\cA_T}_{M}\otimes \sigma_{\hat{M}},
    \end{equation}
where $p_{\cA_T}\in[0,1]$ and $\gamma^{\cA_T}_{M}$ depend only\footnote{By this, we mean that $p_{\cA_T}$ can be computed and the state $\gamma^{\cA_T}_{M}$ can be prepared without the knowledge of the input message $\sigma_{M\hat{M}}$.} on the threshold tampering adversary $\cA_T$, and $\approx_\eps$ denotes that the two states are $\eps$-close in trace distance.

If \cref{eq:qnmssdef} is only guaranteed to hold when $\sigma_{M}$ is the maximally mixed state, then we say that $(\nmshare,\nmrec)$ is an \emph{average-case} $t$-out-of-$p$ $(\epspriv,\epsnm)$-non-malleable secret sharing scheme for quantum messages.
 \end{itemize}
\end{definition}

As our main result in this direction, which builds on \cref{thm:mainqnmc}, we construct efficient split-state NMSS schemes for quantum messages realizing more general threshold access structures with low privacy and non-malleability errors.

\begin{restatable}[Split-state threshold NMSS schemes for quantum messages]{theorem}{qnmss}\label{thm:mainqnmss}
    There exist constants $c,C>0$ and an integer $n_0\in\mathbb{N}$ such that the following holds for any number of parties $p$ and threshold $t\geq 3$ such that $t\leq p\leq 2t-1$ and for any $n\geq n_0$: 
    There exists a family of $t$-out-of-$p$ $(\epspriv=\eps,\epsnm=\eps)$-non-malleable secret sharing schemes for quantum messages with shares of size at most $(pn)^C$, message length $\lfloor n^c\rfloor$, and error $\eps=2^{-n^c}$.
    Furthermore, the sharing and reconstruction procedures of this scheme can be computed in time polynomial in $p$ and $n$.
\end{restatable}


Combined with our previously discussed $2$-out-of-$2$ non-malleable secret sharing scheme for quantum messages, \cref{thm:mainqnmss} covers all remaining threshold access structures for which secret sharing is possible in the quantum setting (i.e., those which do not violate no-cloning) except $2$-out-of-$3$.
We leave constructing a $2$-out-of-$3$ non-malleable secret sharing scheme for quantum messages as a very interesting open problem.

Finally, an analogous realization of the approach behind \cref{thm:mainqnmss} using the same techniques allows us to obtain \emph{classical} threshold non-malleable secret sharing schemes secure against quantum adversaries with shared entanglement. 
See \cref{sec:qsqnmss} for more details, including a formal result.

\subsection{Other related work}

In this section, we discuss relevant prior work on classical NMCs and on quantum non-malleability beyond what we covered above.

\paragraph{Classical NMCs.}

The first work on NMCs by Dziembowski, Pietrzak, and Wichs~\cite{DPW10} showed that, surprisingly, there exists a (possibly inefficient) NMC against any family of at most $2^{2^{\alpha n}}$ tampering functions, where $\alpha<1$ is an arbitrary constant, and $n$ is the message length. The best possible rate of (possibly inefficient) NMCs was studied by Cheraghchi and Guruswami~\cite{CG14a}, who showed that split-state NMCs can have a rate of at most $1/2$. To complement the above, \cite{DPW10} constructed efficient NMCs in the \emph{bitwise tampering model}, a strictly weaker model than split-state tampering, where each bit of the codeword is tampered independently. As discussed before, this spurred a deep line of work which recently culminated in explicit constructions of constant-rate NMCs in the split-state model~\cite{AO20,AKOOS22}.

Several works have also studied NMCs against computationally-bounded adversaries from various hardness assumptions and setups, such as a common reference string. Naturally, in such restricted settings, it is possible to achieve a better rate while allowing the adversary to perform many adaptive tamperings in a row. Computationally-restricted tampering models include polynomial-time algorithms~\cite{FMNV14}, split-state polynomial-time adversaries~\cite{AAGMPP16}, bounded-depth circuits~\cite{BDKM16,BDKM18,BFMV22}, low-degree polynomials~\cite{BCLMT20}, decision trees~\cite{BDKM18,BGW19}, and streaming space-bounded algorithms \cite{FHMV17,BDKM18}. Some other works have studied constructions of short NMCs based on conjectured properties of practical block ciphers~\cite{FKM18,BFRV22} or extractable hash functions~\cite{KLT16}.

We note that existing works in these settings only consider classical messages and adversaries.

\paragraph{Classical NMSS schemes.}

The notion of non-malleable secret sharing first appeared implicitly in the work of Aggarwal, Dziembowski, Kazana, and Obremski~\cite{ADKO15}, where it was shown that every classical $2$-split-state NMC is also a $2$-out-of-$2$ secret sharing scheme with statistical privacy. NMSS schemes were then studied explicitly and in much greater generality by Goyal and Kumar~\cite{GK16,GK18}, leading to a long line of research on classical split-state non-malleable secret sharing, as mentioned above.

Classical NMSS schemes have also been studied in tampering models beyond split-state tampering. For example, the original works of Goyal and Kumar~\cite{GK16,GK18} also consider a "joint" tampering model where, under certain restrictions, tampering adversaries may tamper with \emph{subsets} of multiple shares, instead of only a single share. Stronger joint tampering models have been considered in the computational~\cite{BFOSV20} and information-theoretic settings~\cite{GSZ21}.

In an orthogonal direction, other notions of non-malleability where the adversary learns the reconstruction of tampered secrets with respect to multiple authorized subsets of parties have also been studied~\cite{ADNOP19,BFOSV20}.

\paragraph{Other notions of quantum non-malleability.}

Some other notions of non-malleability for quantum messages have been studied in the context of \emph{keyed} coding schemes by Ambainis, Bouda, and Winter~\cite{ABW09} and Alagic and Majenz~\cite{AM17}. In the keyed setting, quantum authentication schemes~\cite{BW16} also provide non-malleability, since they allow one to detect tampering on the encoded quantum data.

We can view the setting studied in~\cite{ABW09,AM17} for keyed coding schemes in the split-state model as follows (see \cref{fig:splitstate21} for a diagram of the model): Let $\sigma_M$ be the quantum message, and $R$ be the key shared between the encoder and decoder. Let $\sigma_M$ be encoded to $\rho_Z= \enc_R(\sigma_M)$. Then, the adversary tampers $\rho_Z \to \tau_{Z}$, and the decoder outputs $\eta_{M}= \dec_R(\tau_{Z})$. Since the key $R$ is available at both the encoder and decoder, we can view $R$ as being the second part of the codeword, with the register $Z$ being the first part of the codeword. Observe that in this model, the adversary is only allowed to tamper with the first part of the codeword, while our split-state tampering model allows the adversary to simultaneously but independently tamper with both parts of the codeword.
Alagic and Majenz~\cite{AM17} used unitary $2$-designs in their protocol for encoding and decoding. However, we note that their security definition involving mutual information appears a priori different from our \cref{def:qnmcodesfinaldef}.

The question of whether one can even build split-state non-malleable coding schemes for quantum messages when there is no shared key and when the split-state adversary can tamper with both registers (with or without shared entanglement) remained open. In this work, we resolve this question in the affirmative.

\cref{table:nmcs} summarizes the main properties of known constructions of split-state NMCs and related constructions of keyed quantum schemes.

\paragraph{Concurrent work.} 
In concurrent and independent work, Bergamaschi~\cite{Ber23} introduces, among other things, a natural quantum analogue of the \emph{bitwise tampering} model~\cite{DPW10}, where each qubit of the encoding is tampered independently, and constructs high-rate codes in this setting satisfying a restricted form of keyless authentication called tamper detection which is somewhat stronger than non-malleability. We note that this result is incomparable to ours. First, we consider tampering adversaries with access to shared entanglement, while Bergamaschi~\cite{Ber23} only studies the weaker setting where the various tampering adversaries are unentangled. 
In fact, tamper-detection codes (where the receiver either recovers the original message or aborts if they detect the adversary) are impossible to construct against adversaries with access to shared entanglement.
This is because such adversaries can replace the quantum ciphertext with a fixed valid codeword.
Second, the split-state tampering adversaries we study are much more powerful than bitwise tampering adversaries. 

As already discussed above, in another concurrent work Batra, Boddu, and Jain~\cite{BBJ23} constructed a classical $2$-split-state non-malleable randomness encoder secure against quantum adversaries having shared entanglement with rate close to $1/2$.
They use it to build constant-rate quantum secure $3$-split-state NMCs for quantum messages.

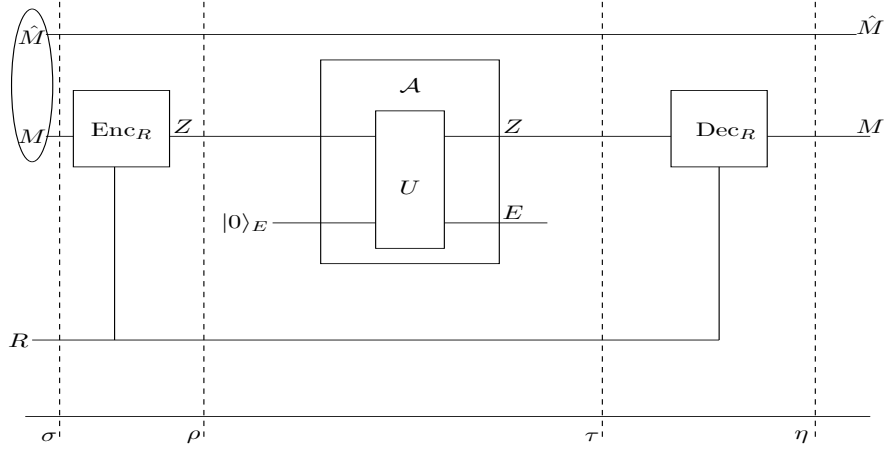
\begin{figure}
\centering
\resizebox{12cm}{6cm}{
\begin{tikzpicture}

\node at (1,6.5) {$\hat{M}$};
\node at (13.2,6.7) {$\hat{M}$};
\draw (1.2,6.5) -- (13,6.5);
\draw (1,5.5) ellipse (0.3cm and 1.5cm);
\draw (1.6,3.9) rectangle (3,5.4);
\draw (10.3,3.9) rectangle (11.7,5.4);
\node at (1,4.5) {$M$};

\node at (2.3,4.6) {$\enc_R$};
\node at (11.1,4.6) {$\dec_{R}$};


\draw (1.2,4.5) -- (1.6,4.5);
\draw (3,4.5) -- (6,4.5);
\draw (7,4.5) -- (10.3,4.5);
\draw (11.7,4.5) -- (13.2,4.5);
\node at (13.2,4.7) {$M$};
\node at (0.8,0.5) {$R$};
\draw (1,0.5) -- (11,0.5);


\draw  (2.2,0.5) -- (2.2,3.9);

\draw [dashed] (1.4,-1.4) -- (1.4,7.2);
\draw [dashed] (3.5,-1.4) -- (3.5,7.2);
\draw [dashed] (9.3,-1.4) -- (9.3,7.2);
\draw [dashed] (12.4,-1.4) -- (12.4,7.2);

\draw (0.9,-1) -- (13.2,-1);

\node at (1.25,-1.4) {$\sigma$};
\node at (3.35,-1.4) {$\rho$};
\node at (9.15,-1.4) {$\tau$};
\node at (12.2,-1.4) {$\eta$};

\node at (3.2,4.7) {$Z$};
\node at (8,4.7) {$Z$};
\node at (4.1,2.8) {$\ket{0}_E$};
\node at (8,3) {$E$};
\draw (4.5,2.8) -- (6,2.8);
\draw (7,2.8) -- (8.5,2.8);

\draw (11,0.5) -- (11,3.9);


\draw (6,2.3) rectangle (7,5);
\node at (6.5,3.5) {$U$};

\node at (6.5,5.5) {$\mathcal{A}$};
\draw (5.2,2) rectangle (7.8,6);





\end{tikzpicture}}

\caption{Quantum NMC with shared key.}\label{fig:splitstate21}
\end{figure}

\begin{table}
\centering
\begin{tabular}{||c c c c c||} 
 \hline
 \textbf{Work by} & \textbf{Rate} & \textbf{Messages} & \textbf{Adversary} & \textbf{Shared key}\\ [1ex] 
 \hline\hline 
\cite{CGL15} & $\frac{1}{\poly(n)}$  & classical  & classical  & No \\ [1ex] 
 \hline
 \cite{Li17} & $\Omega\left( \frac{1}{\log n}\right)$ & classical & classical & No\\ [1ex] 
 \hline
 \cite{Li19} & $ \Omega\left(\frac{ \log \log n}{\log n}\right)$ & classical & classical & No \\ [1ex] 
 \hline
 \cite{AO20} & $\Omega(1)$ & classical & classical & No \\[0.5ex] 
 \hline
  \cite{Li23} & $\Omega(1)$ & classical & classical & No \\[0.5ex] 
 \hline
 \cite{AKOOS22} & $1/3$ & classical & classical & No \\[0.5ex] 
 \hline
 \cite{ABJ22} & $\frac{1}{\poly(n)}$ & classical & \textbf{quantum} & No \\[0.5ex] 
 \hline
 \cite{BBJ23} & $\approx 1/5$ & {(average-case) classical} & \textbf{quantum} & No \\[0.5ex] 
 \hline
\cite{AM17} & $\Omega(1)$ & \textbf{quantum} & \textbf{quantum} & {\color{red} Yes} \\[0.5ex] 
\hline
 This work & $\approx 1/11$ & \textbf{(average-case) quantum} & \textbf{quantum}  & No\\ [0.5ex] 
 \hline
 This work & $\frac{1}{\poly(n)}$ & \textbf{quantum} & \textbf{quantum}  & No\\ [1ex] 
 \hline
\end{tabular}

\vspace{0.5cm}

\caption{Comparison between the best known explicit constructions of $2$-split-state NMCs. Here, $n$ denotes the codeword length.}
\label{table:nmcs}
\end{table}

\subsection{Technical overview}

\subsubsection{Split-state NMCs for quantum messages}\label{sec:techoverNMC}
We follow a high-level strategy similar to that of Aggarwal, Agrawal, Gupta, Maji, Pandey, and Prabhakaran~\cite{AAGMPP16}, who used authenticated encryption to transform any (augmented) split-state non-malleable code into a high-rate split-state non-malleable code resilient to computationally-bounded tampering in the classical setting. The main challenge we have to deal with, and which is not considered in~\cite{AAGMPP16}, is that our tampering adversaries have quantum capabilities and are allowed to a priori share arbitrary \emph{entangled quantum states}.

As a warm-up to our main construction, let us consider a simpler and natural approach that yields a $3$-split-state non-malleable code for quantum messages. First, we observe that if we have a shared classical secret key $R$ between the encoder and decoder, then we can detect whether \emph{arbitrary} tampering has occurred by encoding the quantum message $\sigma_{M}$ with a quantum authentication scheme $(\auth_R,\ver_R)$~\cite{BCGST02,BW16}. Therefore, a reasonable approach would be to first encode the message using this quantum authentication scheme, then encode the key $R$ using an existing split-state non-malleable code for classical messages, and lastly output the split-state encoding of $R$ as part of the final codeword.

Fortunately, since $R$ is a classical string, we can use the split-state non-malleable code $(\cenc,\cdec)$ of Aggarwal, Boddu, and Jain~\cite{ABJ22}, which protects against quantum adversaries!
Less fortunately, $\cenc(R)$ yields a codeword with two parts, call them $R_1$ and $R_2$, that cannot be stored together.

To elaborate on this approach, let's define a quantum authentication scheme as $(\auth_R, \ver_R)$, where $\auth_R$ is a quantum encoding procedure that takes a quantum state $\sigma_M$ and a classical secret key $R$ and outputs an authenticated state $\psi = \auth_R(\sigma_M)$. The verification procedure $\ver_R$ takes an authenticated state $\psi'$ and the secret key $R$ as input, and outputs either $1$ (accept) or $0$ (reject). The overall encoding procedure $\enc(\sigma_M)$ for our $3$-split-state non-malleable code is as follows:
\begin{enumerate}
    \item Sample a classical secret key $R$ uniformly at random from an appropriate keyspace.
 \item Compute the authenticated state $\psi = \auth_R(\sigma_M)$ using the quantum authentication scheme.
    \item Compute the classical split-state encoding $(R_1, R_2) = \cenc(R)$.
    \item Output the $3$-state codeword $(\psi, R_1, R_2)$.
\end{enumerate}

To decode a possibly tampered codeword $(\psi', R'_1, R'_2)$, the decoding procedure $\dec$ works as follows:
\begin{enumerate}
    \item 
Use the classical decoder $\cdec(R'_1, R'_2)$ to obtain a candidate key $R'$.
      \item Run the verification procedure $\ver_{R'}(\psi')$ to check the authenticity of the quantum state $\psi'$ using the candidate key $R'$.
    \item If the verification procedure outputs $1$ (accept), then output $\dec(\psi', R'_1, R'_2)$. Otherwise, output a special symbol to indicate tampering.

\end{enumerate}

This approach leverages the quantum authentication scheme to detect any tampering of the quantum message, and the classical split-state non-malleable code to protect the classical key $R$ against tampering. Intuitively, if we consider quantum adversaries \emph{without} shared entanglement, the non-malleability of the $3$-state coding scheme can be understood as follows: If the adversary attempts to tampering with $R_1$ and $R_2$, then the properties of the classical split-state non-malleable code $(\cenc, \cdec)$ guarantee that either the candidate key $R'$ remains unchanged ($R' = R$) or that it becomes independent of $R$. In the former case, when the decoder $\dec$ calls $\ver_R(\psi')$, the quantum authentication scheme will, with high probability, detect whether $\psi' \neq \psi$, indicating tampering, and $\dec$ will abort the decoding process. In the latter case, the decoder calls $\ver_{R'}(\psi')$ with $R'$ being independent of $R$ and, consequently, independent of $\psi'$ as well. Choosing a quantum authentication scheme, such as Clifford-based authentication, it is possible to show that $\ver_{R'}(\psi')$ will output a state independent of (and thus unrelated to) the message $\sigma$ with high probability.

Overall, this approach combines the strength of quantum authentication to detect tampering with the quantum message and the non-malleability property of the classical split-state code to protect the classical key $R$ against tampering. By leveraging these properties together, we obtain a $3$-split-state non-malleable code for quantum messages. Although the coding scheme above already gives some non-malleability guarantees for quantum messages, it features some major shortcomings. First, it requires dividing the codeword into three parts which must be stored separately.
Ideally, we would like to construct efficient coding schemes that only need to be divided into \emph{two} parts. Second, our argument above only works against quantum adversaries \emph{without} shared entanglement.
However, as we move to the more powerful setting with adversaries having shared entanglement, we need to develop more sophisticated techniques to ensure non-malleability.

\paragraph{Split-state non-malleable codes for average-case quantum messages.}
Perhaps the most natural approach to building a split-state non-malleable code secure against quantum adversaries with shared entanglement would be to take our $3$-state code above and merge two of the parts.
More precisely, we could attempt to analyze the code which outputs
\begin{equation*}
    (\psi, R_1)\quad \textrm{and}\quad R_2
\end{equation*}
as its two parts, where we recall that $\psi$ is the authenticated state via the secret key $R$ and $(R_1,R_2)=\cenc(R)$ is the classical split-state non-malleable encoding of $R$.

This matches the approach taken in~\cite{AAGMPP16}.
However, as already mentioned above, they did not have to handle quantum adversaries with shared entanglement.
Although this is also essentially the approach we successfully undertake, establishing non-malleability of this modified coding scheme is significantly more involved than the intuitive argument laid out above, and also requires some small further changes to the scheme. 
Indeed, one of the first difficulties arises due to the use of trap flags\footnote{Clifford-based quantum authentication schemes apply a random (secret) Clifford operator to the message plus several additional ``trap registers'' initialized to $\ket{0}$. Verifying whether tampering of the authenticated state occurred consists of checking whether the trap registers all return to the $\ket{0}$ state after applying the inverse Clifford operator. If this does not hold, then the verification procedure outputs the special symbol $\perp$, which we call the ``trap flag''.} in the authentication schemes, as we will discuss subsequently. 
We note that, unlike in an authentication scheme, the decoder need not output $\perp$ in a non-malleable code, since error detection is not required. 
We therefore get rid of the trap flags used in the authentication scheme specified in the 3-split argument above, and this is one of the first insights that allows our analysis to work.

We first describe a sub-optimal version of our approach in the setting of \emph{average-case non-malleability}, and then discuss how its rate can be optimized.
Per \cref{def:qnmcodesfinaldef}, this corresponds to the scenario where the message $\sigma_M$ is assumed to be a maximally mixed quantum state with canonical purification $\sigma_{M\hat{M}}$. To construct our split-state non-malleable code for quantum messages, we use random Clifford unitaries $\{(C_r, C^\dagger_r) \}_{r \leftarrow R}$ with underlying classical randomness $R$, along with the classical split-state coding scheme $(\cenc, \cdec)$ designed in \cite{ABJ22} as a quantum-secure non-malleable code for the classical string $R$. 

Inspired by the $3$-state approach above, we use $(\cenc, \cdec)$ to protect the key $R$ in a non-malleable manner, and then use the random Clifford $C_R$ to protect the quantum message $\sigma_M$. 
This yields the following encoding procedure $\enc(\sigma_M)$, where we use slightly different notation than the above to facilitate our analysis
(see also \cref{fig:splitstate2} for a diagram of our encoding and decoding procedures):
\begin{enumerate}
    \item Sample a classical secret key 
    $R$ uniformly at random (independent of $\sigma_{M\hat{M}}$) from an appropriate keyspace;
    
    \item Compute the  state $(\sigma_1)_Z = C_R(\sigma_M)C^\dagger_R$;
    
    \item Compute the split-state encoding $\cenc(R)=(\sigma_1)_{XY}$;
    
    \item Output the $2$-part codeword $((\sigma_1)_{ZX},(\sigma_1)_{Y})$.
\end{enumerate}
The message $\sigma_{M}=U_M$, along with its purification register $\hat{M}$, is thus encoded into $(\sigma_1)_{\hat{M}ZXY}$. 

As our first observation, note that the register $Z$, which holds the Clifford-protected message, may carry information about the classical key $R$, since $(\sigma_1)_Z = C_R(\sigma_M)C^\dagger_R$.
Fortunately, since $\sigma_M=U_M$, where $U_M$ is a maximally mixed state, we have 
\[ 
    (\sigma_1)_{ZXY}=(\sigma_1)_{Z} \otimes (\sigma_1)_{XY}.
\]
One may then expect that we can use the argument of~\cite{ABJ22} for classical messages to argue that the key $R$ remains secure even if the adversary sees the register $Z$ when tampering one of the parts of $\cenc(R)$. 
Unfortunately, the argument of~\cite{ABJ22} does not go through in this scenario.
Namely, for $(\cenc, \cdec)$ to protect the key $R$ in a non-malleable manner after adversarial tampering we need that
\[ 
    (\sigma_1)_{\hat{M}ZXY}=(\sigma_1)_{\hat{M}Z} \otimes (\sigma_1)_{XY}.
\]
However, it can be verified that the registers $\hat{M}Z$ are not independent of $XY$ in state $\sigma_1$. 
To circumvent this issue, we use the transpose method (see \cref{fact:transposetrick}) for state $\sigma_{M\hat{M}}$, and note that the application of the random Clifford gate $C_R$ and the adversarial operations commute in $\sigma_1$ (see \cref{fig:splitstate2,fig:splitstate3}). 
This allows us to delay the operation $C_R$ on register $\hat{M}$ (see \cref{fig:splitstate4}). 
Crucially, we could not have used the transpose method in the presence of flag registers as employed in quantum authentication. 

Carefully combining the above with arguments from~\cite{ABJ22}, we conclude that $(\cenc, \cdec)$ can protect the key $R$ after adversarial tampering on state $\theta_1$ (which corresponds to the state prior to applying the $C_R$ operation on register $\hat{M}$ in \cref{fig:splitstate4}), since 
\[
    (\theta_1)_{\hat{M}ZXY}=(\theta_1)_{\hat{M}Z} \otimes (\theta_1)_{XY}. 
\]
Let $\theta_2$ be the state obtained after adversarial tampering on $\theta_1$. 
By the properties of $(\cenc,\cdec)$, we are essentially guaranteed that either $R'=R$ or that $R'$ is independent of $R$ in $\theta_2$. 
However, this is not enough, and we observe that $(\cenc,\cdec)$ actually guarantees something even stronger!
More precisely, in state $\theta_2$ we  have either:
\begin{itemize}
    \item $R=R'$ and $R$ is independent of registers $Z\hat{M}$;
    \item $R$ is independent of $R'Z\hat{M}$.
\end{itemize}

In the case where $R$ is independent of $R'Z\hat{M}$, applying the delayed operation $C_R$ on register $\hat{M}$ decouples the registers $R'Z\otimes  \hat{M}$, and so we are done.
We are thus left to analyze the case where $R=R'$ and $R$ is independent of registers $Z\hat{M}$.
Here, we make use of the $2$-design properties of the Clifford scheme $(C_R, C^\dagger_R)$.
Roughly speaking, first suppose that the adversary applied  $\id_Z$ on register $Z$. Then, since the message was maximally mixed, we conclude that we get an EPR state as the outcome of the tampering experiment (after the decoding procedure is applied).
Now, suppose that the adversary applied some 
$P \ne \id_Z$ on register $Z$.
We make use of Clifford randomization and the twirl property to handle this scenario.
To elaborate, first note that the state
\begin{equation*}
    (\id \otimes P) (\theta_1)_{\hat{M}Z} (\id \otimes P^\dagger)
\end{equation*}
is in a subspace orthogonal to an EPR state in $\hat{M}Z$. Consider the state after applying the delayed operation $C_R$ on register $\hat{M}$ and $C^\dagger_R$ on register $Z$. 
Here, Clifford randomization and the twirl property ensure that any state orthogonal to an EPR state in registers $\hat{M}Z$ is exactly maximally mixed in a subspace (of dimension $4^{\vert \hat{M} \vert}-1$) orthogonal to this EPR state in registers $\hat{M}Z$ (see \cref{lem:equal101}). As a result, the outcome of the tampering experiment after the decoding procedure is applied is close in trace distance to $U_{\hat{M}} \otimes U_{M}$. 
Our proof overall crucially uses interesting properties of the Pauli and Clifford unitaries including Pauli twirl, Clifford twirl, and Clifford randomization.

\paragraph{Improving the rate from sub-constant to constant in the average-case setting.}

An important quantity associated with a (non-malleable code) is its \emph{rate} -- the ratio between the size of a message and the size of its corresponding encoding. Looking at the proposed approach we discussed above, we conclude that its rate is sub-constant, i.e., the size of the encoding is a superlinear function of the message size. The main reasons behind this are as follows: First, the classical key $R$ that we use to sample the Clifford operator $C_R$ is much longer than the message $\sigma_M$ (In fact $\vert R \vert = \cO(\vert M \vert^2)$). Second, the rate of the underlying non-malleable code for classical messages $(\cenc,\cdec)$ from~\cite{ABJ22} is already sub-constant.

We now explain how our approach above can be modified to yield a constant-rate average-case non-malleable code for quantum messages. First, instead of using the whole Clifford group, we can use a shorter random key $R$ to sample a Clifford operator from a smaller subgroup with special properties that suffice for our needs. 
A result of Cleve, Leung, Liu, and Wang~\cite{CLLW16} guarantees that this can be done efficiently with a classical key $R$ of length at most $5|M|$. Second, observe that we only care about obtaining a split-state non-malleable encoding of a \emph{uniformly random} classical key $R$. This means that we can replace the classical split-state non-malleable code $(\cenc,\cdec)$ from~\cite{ABJ22} by a \emph{Non-Malleable Randomness Encoder} (NMRE), an object originally introduced by Kanukurthi, Obbattu, and Sekar~\cite{KOS18} and whose known constructions enjoy much better rates than non-malleable codes. 
Batra, Boddu, and Jain~\cite{BBJ23} recently constructed a classical NMRE secure against quantum adversaries having shared entanglement with a rate close to $1/2$.\footnote{For the experienced reader, Batra, Boddu, and Jain~\cite{BBJ23} construct an explicit quantum-secure $2$-source non-malleable extractor $\nmext$ with a large output length. We can then sample classical bitstrings $X$ and $Y$ uniformly at random with appropriate lengths and set the classical key $R$ to be $R=\nmext(X,Y)$; this is the quantum-secure classical NMRE that we use in our optimized coding scheme.} 
In contrast, as mentioned above, the NMC for classical messages from~\cite{ABJ22} only has a sub-constant rate.

Using these two results in our previously described approach allows us to improve the rate of our average-case NMC for quantum messages from sub-constant to close to $1/11$.

\paragraph{From average-case to worst-case NMCs for quantum messages.}
In the discussion above we assumed that the quantum message was maximally mixed.
However, we would like to show that our code is a worst-case NMC, i.e., it is non-malleable for \emph{any} fixed quantum message.

Consider an arbitrary quantum message $\rho_M$ with canonical purification $\rho_{M\hat{M}}$.
Recall that $\sigma_{ \hat{M}M}$ is the canonical purification of $\sigma_M=U_M$. 
Using a quantum rejection sampling argument (see \cref{rejectionsampling}), we obtain a measurement acting on register $\hat{M}$ of $\sigma$ such that the state conditioned on ``success'' is exactly $\rho$. 
Moreover, this measurement on the register $\hat{M}$ commutes with $\enc$,  $\dec$, and the adversarial operations, and it succeeds with probability $2^{-|M|}$.
If the error of the underlying average-case NMC is $\eps$, then this measurement allows us to argue non-malleability of the same NMC for any message $\rho_{M}$ (i.e., worst-case non-malleability) with larger error $\eps' = 2^{|M|}\cdot \eps$, and we can easily handle this blow-up in the error at the expense of dropping the rate of the code from constant to sub-constant.
Intuitively, this happens because of the following: Suppose that there is a fixed message $\rho$ for which~\cref{eq:splitnmc} (with $\rho$ in place of $\sigma$) holds only with error larger than $\eps'$.
Therefore, when faced with a maximally mixed message $\sigma_{M\hat{M}}$, we can apply the measurement above to $\hat{M}$ and, if the measurement succeeds and the resulting state is $\rho$, distinguish with advantage larger than $\eps'$.
Since the measurement succeeds with probability at least $2^{-|M|}$, the overall distinguishing advantage is larger than $2^{-|M|}\cdot \eps' = \eps$, which contradicts the average-case $\eps$-non-malleability of the NMC.


\subsubsection{Threshold split-state NMSS schemes for quantum messages}\label{sec:techoverviewNMSS}

At a high level, in order to construct our threshold NMSS schemes for quantum messages, we combine our split-state NMC for quantum messages above with the approach of Goyal and Kumar~\cite{GK16} used to construct NMSS schemes in the classical setting. However, as we will discuss, following this approach in the quantum setting poses various challenges.

Roughly speaking, the approach of Goyal and Kumar~\cite{GK16} for building NMSS schemes proceeds as follows for $p$ parties and a threshold $3\leq t\leq p$. On input a message $m$, first encode it with the split-state NMC to generate two states, $L$ and $R$. Now, apply a standard $t$-out-of-$p$ secret sharing scheme to $L$ (say, Shamir's secret sharing), yielding shares $L_1,\dots,L_p$. Furthermore, apply a $2$-out-of-$p$ \emph{leakage-resilient} secret sharing scheme to $R$, yielding shares $R_1,\dots,R_p$. Intuitively, a secret sharing scheme is leakage-resilient if the input remains private even when the adversary learns an unauthorized subset of shares \emph{plus} bounded side information from every other share. Finally, set the resulting $i$-th share $S_i$ as $S_i=(L_i,R_i)$ for each $i\in[p]$.

To argue the non-malleability of this construction, Goyal and Kumar showed how to transform any tampering attack on the resulting secret sharing scheme into an essentially-as-good tampering attack on the underlying $2$-split-state NMC. The main challenge in designing such a reduction is that the two tampering functions for the underlying NMC must act independently -- we must tamper $L$ without knowledge of $R$, and vice versa. On the other hand, the $i$-th tampering function for the secret sharing scheme tampers $(L_i,R_i)$ into $(L'_i,R'_i)$, and so potentially has access to some information from both $L$ and $R$. For simplicity, let $L'$ and $R'$ denote the tampered secrets reconstructed from $(L'_1,\dots,L'_t)$ and $(R'_1,R'_2)$, respectively.

Showing that we can obtain $R'$ from $R$ without knowing $L$ is easy. This follows because $R'$ is fully determined by $R'_1$ and $R'_2$, which in turn only depend on $L_1$ and $L_2$. Since the $L_i$'s are a $t$-out-of-$p$ Shamir secret sharing of $L$ with threshold $t\geq 3$, the two shares $L_1$ and $L_2$ are independent of $L$. However, arguing that we can (up to small error) obtain $L'$ from $L$ without knowledge of $R$ is a lot trickier. The previous argument clearly does not immediately work since the $t$ shares $L'_1,\dots,L'_t$ depend on $R_1,\dots,R_t$, respectively, which determine $R$. This is where the leakage-resilience property kicks in -- if we see the $L'_i$s as bounded leakage on the $R_i$s, then leakage-resilience guarantees that $R$ is (close to) independent of the leakages $L'_1,\dots,L'_t$, and so is independent of $L'$.

\paragraph{Realizing this approach in the quantum setting.}

We follow the same high-level approach for a quantum message $\sigma_M$. To that end, we first replace the classical split-state NMC by our NMC for quantum messages discussed previously. We use additional key properties of our NMC: (1) The resulting left state $L$ is quantum, but the right state $R$ is \emph{classical}, and (2) we prove that it is actually a $2$-out-of-$2$ split-state NMSS scheme, and so, in particular, learning only one of $L$ and $R$ reveals nothing about $\sigma_M$. Since $L$ is quantum, we now apply a standard $t$-out-of-$p$ (for $p < 2t$) secret sharing scheme for quantum messages, such as ``quantum'' Shamir secret sharing~\cite{CGL99}, to get quantum shares $(L_1,\dots,L_p)$. And, since $R$ is classical, we secret-share it using a $2$-out-of-$p$ scheme satisfying a special leakage-resilience property that we will determine later.

Establishing correctness and privacy of the resulting $t$-out-of-$p$ secret sharing scheme is not difficult using property (2) of our NMC above. It remains to prove non-malleability. Arguing that $R'$ can be obtained from $R$ without knowledge of $L$ still follows easily from the fact that $R$ is shared using a $2$-out-of-$p$ scheme, while $L$ is shared using a $t$-out-of-$p$ scheme with $t\geq 3$. Here, it is crucial that $R$ is classical. Otherwise, a $2$-out-of-$p$ scheme would not exist when $p\geq 4$, and this means that we would not be able to, say, construct $(p/2+1)$-out-of-$p$ split-state NMSS schemes for quantum messages for \emph{any} even $p$ using this approach.

We also want to argue that $L'$ can be obtained from $L$ without knowledge of $R$.
As before, we would like to see $L'_1,\dots,L'_p$ as leakages on the secret shares $R_1,\dots,R_p$, and then exploit the leakage-resilience of the scheme used to share $R$ to conclude that $R$ is independent of $L'_1,\dots,L'_p$, and hence of $L'$.
However, realizing this in our quantum setting requires stronger leakage-resilience properties: First, the tamperings $L'_1,\dots,L'_p$ are now quantum states.
Second, the tampering functions in our setting \emph{share arbitrary entangled states}.
This means that the $i$-th tampering function now sees $(L_i,R_i,W_i)$, where $W_1,\dots, W_p$ are quantum registers holding an arbitrary state.
To overcome these barriers, we introduce \emph{augmented} leakage-resilient secret sharing schemes.
This corresponds to a setting where there are $p$ local adversaries $\cA_1,\dots,\cA_p$ sharing an arbitrary entangled state spread across registers $W_1,\dots, W_p$, respectively, and each one having access to a share $R_1,\dots,R_p$.
We require that $R$ remains hidden even if $\cA_i$ knows the share $R_i$, local bounded quantum leakages $\Leak_1(R_1,W_1),\dots, \Leak_{i-1}(R_{i-1},W_{i-1}),\Leak_{i+1}(R_{i-1},W_{i+1}),\dots, \Leak_{p}(R_{p},W_{p})$ from every other share, \emph{and also the entangled state $W_i$} (hence the ``augmented'' adjective).
We prove that the $2$-out-of-$2$ secret sharing scheme whereby $R$ is shared into $X$ and $Y$ sampled uniformly such that $\langle X,Y\rangle =R$ is augmented leakage-resilient with good parameters.
This relies on the randomness extraction properties of the inner product function and the formalism of ``qpa-states''\footnote{qpa-state stands for quantum purified adversary state.}~\cite{BJK21}.
We can then extend this scheme to a $2$-out-of-$p$ augmented leakage-resilient scheme for classical messages in a standard manner.
Although notions of leakage-resilience against quantum adversaries with shared entanglement have been studied in other recent work~\cite{CGLR23}, these do not cover augmented leakage-resilience. 

Besides the above, proving non-malleability requires dealing with additional subtleties specific to the quantum setting.
The original non-malleability argument in~\cite{GK16} proceeds by fixing the values of certain components (e.g., shares, leakages).
However, we cannot fully emulate this approach in the quantum setting: The left state $L$ is quantum, and so are the bounded leakages that show up in the analysis.
Therefore, we cannot fix them.
Moreover, again because $L$ is quantum, it is modified after the tampering functions are applied, but we still need to access the ``original'' $L$ in the analysis. For this, with some work, we can use the message's canonical purification register $\hat{M}$ to generate a new register $\hat{L}$ which can be thought of as a coherent copy of the original left state $L$.

\subsection{Open problems}

We list here some interesting directions for future research:
\begin{itemize}
    \item Can we design (worst-case) split-state NMCs for quantum messages with a constant rate? This is open even for classical messages against quantum adversaries with shared entanglement. More ambitiously, can we construct (worst-case) split-state NMSS schemes for quantum messages with a constant rate?
\item Can we construct $2$-out-of-$3$ split-state non-malleable secret sharing schemes for quantum messages?
\item Can we design NMSS schemes for quantum messages that are secure against joint tampering of shares?
\item What can we achieve if we consider computationally-bounded adversaries instead?
\end{itemize}

\subsection{Acknowledgements}
We thank Thiago Bergamaschi for insightful discussions about notions of non-malleability in the quantum setting.
We also thank Dakshita Khurana for useful discussions in the initial stage of this project.

J.\ Ribeiro's research was supported by NOVA LINCS (UIDB/04516/2020) with the financial support of FCT - Fundação para a Ciência e a Tecnologia. The work of R.\ Jain is supported by the NRF grant NRF2021-QEP2-02-P05 and the Ministry
of Education, Singapore, under the Research Centres of Excellence program. This work was
done in part while R.\ Jain was visiting the Technion-Israel Institute of Technology, Haifa, Israel and the Simons Institute for the Theory of Computing, Berkeley, CA, USA.

\section{Preliminaries}
\label{sec:prelims}

This section collects basic notation and conventions alongside useful facts and lemmas that we use in the proofs of our main results.
In this work, facts denote results already known from prior work, and lemmas denote auxiliary results that we prove here.

\subsection{Basic general notation}
All the logarithms are evaluated to the base $2$. 
We denote sets by uppercase calligraphic letters such as $\X$ and use uppercase roman letters such as $X$ and $Y$ for both random variables and quantum registers. 
The distinction will be clear from context.
The set $\{1,\dots,n\}$ may be written as $[n]$, and we may also more generally write $[t,n]$ for the set $\{t,t+1,\dots,n\}$.
We denote the uniform distribution over $\{0,1\}^d$ by $U_d$.
For a {\em random variable} $X \in \X$, we use $X$ to denote both the random variable and its distribution, whenever it is clear from context. 
We use $x \leftarrow X$ to denote that $x$ is drawn according to $X$, and, for a finite set $\X$, we use $x \leftarrow \X$ to denote that $x$ is drawn uniformly at random from $\X$. For two random variables $X,Y$ we use $X \otimes Y$ to denote their product distribution. 
We call random variables $X$ and $Y$ {\em copies} of each other if and only if $\Pr[X=Y]=1$.

\subsection{Quantum information theory}

In this section we cover some important basic prerequisites from quantum information theory alongside some useful lemmas and facts.

\subsubsection{Conventions and notation}

Consider a finite-dimensional Hilbert space $\cH$ endowed with an inner-product $\langle \cdot, \cdot \rangle$ (we only consider finite-dimensional Hilbert-spaces). A quantum state (or a density matrix or a state) is a positive semi-definite operator on $\cH$ with trace value  equal to $1$. 
It is called {\em pure} if and only if its rank is $1$. Let $\ket{\psi}$ be a unit vector on $\cH$, that is $\langle \psi,\psi \rangle=1$.  
With some abuse of notation, we use $\psi$ to represent the state and also the density matrix $\ketbra{\psi}$, associated with $\ket{\psi}$. Given a quantum state $\rho$ on $\cH$, the {\em support of $\rho$}, denoted by $\text{supp}(\rho)$, is the subspace of $\cH$ spanned by all eigenvectors of $\rho$ with non-zero eigenvalues.
 
A {\em quantum register} $A$ is associated with some Hilbert space $\cH_A$. Define $\vert A \vert := \log\left(\dim(\cH_A)\right)$. 
For a sequence of registers $A_1,\dots,A_n$ and a set $T\subseteq[n]$, we define the projection according to $T$ as $A_T=(A_i)_{i\in T}$.
Let $\mathcal{L}(\cH_A)$ represent the set of all linear operators on the Hilbert space $\cH_A$. For an operator $O \in \cL(\cH_A)$, we use $O^T$ to represent the transpose of $O$. For operators $O, O'\in \cL(\cH_A)$, the notation $O \leq O'$ represents the L\"{o}wner order, that is, $O'-O$ is a positive semi-definite operator. 
We denote by $\mathcal{D}(\cH_A)$ the set of all quantum states on the Hilbert space $\cH_A$. The state $\rho$ with subscript $A$ indicates that $\rho_A \in \mathcal{D}(\cH_A)$. 
If two registers $A,B$ are associated with the same Hilbert space, we shall represent the relation by $A\equiv B$. 
For two states $\rho$ and  $\sigma$, we write $\rho \equiv \sigma$ if they are identical as states (potentially in different registers). Composition of two registers $A$ and $B$, denoted $AB$, is associated with the Hilbert space $\cH_A \otimes \cH_B$.  For two quantum states $\rho\in \mathcal{D}(\cH_A)$ and $\sigma\in \mathcal{D}(\cH_B)$, $\rho\otimes\sigma \in \mathcal{D}(\cH_{AB})$ represents the tensor product ({\em Kronecker} product) of $\rho$ and $\sigma$. The identity operator on $\cH_A$ is denoted $\id_A$. Let $U_A$ denote the maximally mixed state in $\cH_A$. Let $\rho_{AB} \in \mathcal{D}(\cH_{AB})$. Define
$$ \rho_{B} \defeq \tr_{A}{\rho_{AB}} \defeq \sum_i (\bra{i} \otimes \id_{B})
\rho_{AB} (\ket{i} \otimes \id_{B}) , $$
where $\{\ket{i}\}_i$ is an orthonormal basis for the Hilbert space $\cH_A$.
The state $\rho_B\in \mathcal{D}(\cH_B)$ is referred to as the marginal state of $\rho_{AB}$ on the register $B$. Unless otherwise stated, a missing register from subscript in a state represents partial trace over that register. Given $\rho_A\in\mathcal{D}(\cH_A)$, a {\em purification} of $\rho_A$ is a pure state $\rho_{AB}\in \mathcal{D}(\cH_{AB})$ such that $\tr_{B}{\rho_{AB}}=\rho_A$. Purification of a quantum state is not unique.
Suppose $A\equiv B$. Given $\{\ket{i}_A\}$ and $\{\ket{i}_B\}$ as orthonormal bases over $\cH_A$ and $\cH_B$ respectively, the \textit{canonical purification} of a quantum state $\rho_A$ is a pure state $\rho_{AB} \defeq (\rho_A^{\frac{1}{2}}\otimes\id_B)\left(\sum_i\ket{i}_A\ket{i}_B\right)$. 

A quantum {map} $\cE: \mathcal{L}(\cH_A)\rightarrow \mathcal{L}(\cH_B)$ is a completely positive and trace preserving (CPTP) linear map. A CPTP map $\cE$ is described by the Kraus operators $\{ M_i : \cH_A\rightarrow \cH_B \}_i$ such that $\cE(\rho) = \sum_i M_i \rho M^\dagger_i$ and $\sum_i M^\dagger_i M_i =\id_A$.~A {\em Hermitian} operator $H:\cH_A \rightarrow \cH_A$ is such that $H=H^{\dagger}$. A projector $\Pi \in  \mathcal{L}(\cH_A)$ is a Hermitian operator such that $\Pi^2=\Pi$. A {\em unitary} operator $V_A:\cH_A \rightarrow \cH_A$ is such that $V_A^{\dagger}V_A = V_A V_A^{\dagger} = \id_A$. The set of all unitary operators on $\cH_A$ is  denoted by $\mathcal{U}(\cH_A)$. An {\em isometry}  $V:\cH_A \rightarrow \cH_B$ is such that $V^{\dagger}V = \id_A$. A {\em POVM} element is an operator $0 \le M \le \id$.~We use the shorthand $\bar{M} \defeq \id - M$, where $\id$ is clear from context. 
We use shorthand $M$ to represent $M \otimes \id$, where $\id$ is clear from context.

\subsubsection{Registers, quantum maps, and isometries}

This section collects definitions of certain registers and operations on them.

\begin{definition}[Classical register in a pure state]\label{def:classicalinpurestate}
Let $\X$ be a set. 
A {\em classical-quantum} (c-q) state $\rho_{XE}$ is of the form \[ \rho_{XE} =  \sum_{x \in \X}  p(x)\ket{x}\bra{x} \otimes \rho^x_E , \] where ${\rho^x_E}$ are states. 

Let $\rho_{XEA}$ be a pure state. We call $X$ a \emph{classical register} in $\rho_{XEA}$ if $\rho_{XE}$ (or $\rho_{XA}$) is a c-q state. 
Whenever it is clear from context, we identify the random variable $X$ with the register $X$ via $\Pr[X=x]=p(x)$.

\end{definition}

\begin{definition}[Copy of a classical  register]\label{def:copyofaclassicalregister}
Let $\rho_{X\hat{X}E}$ be a pure state with $X$ being a classical register in $\rho_{X\hat{X}E}$ taking values in $\cX$. Similarly, let $\hat{X}$ be a classical register in $\rho_{X\hat{X}E}$ taking values in $\cX$. Let $\Pi_{\mathsf{Eq}} = \sum_{x \in \cX} \ketbra{x} \otimes \ketbra{x}$ be the \emph{equality} projector acting on the registers $X\hat{X}$. We call $X$ and $\hat{X}$ copies of each other (in the computational basis) if $\tr\left(\Pi_{\mathsf{Eq}} \rho_{X\hat{X}}\right) =1$.
\end{definition}

\begin{definition}[Conditioning] \label{def:conditioning}
Let  
\[ \rho_{XE} =  \sum_{x \in \{0,1\}^n}  p(x)\ket{x}\bra{x} \otimes \rho^x_E , \]
be a c-q state. For an event $\mathcal{S} \subseteq \{0,1\}^n$, define  $$\Pr[\mathcal{S}]_\rho \defeq  \sum_{x \in \mathcal{S}} p(x) \quad \textrm{and} \quad (\rho|X\in \mathcal{S})\defeq \frac{1}{\Pr[\mathcal{S}]_\rho} \sum_{x \in \mathcal{S}} p(x)\ket{x}\bra{x} \otimes \rho^x_E.$$
We sometimes shorthand $(\rho|X\in \mathcal{S})$ as $(\rho|\mathcal{S})$ when the register $X$ is clear from context. 

Let $\rho_{AB}$ be a state with $|A|=n$. We define 
$(\rho|A \in \mathcal{S}) \defeq (\sigma|\mathcal{S})$, where $\sigma_{AB}$ is the c-q state obtained by measuring the register $A$ in $\rho_{AB}$ in the computational basis. In the case where $\mathcal{S}=\{s\}$ is a singleton set, we shorthand $(\rho|A = s) \defeq \tr_A (\rho|A =s)$.
\end{definition}

\begin{definition}[Safe maps] \label{def:safe}
We call an isometry $V: \cH_X \otimes \cH_A \rightarrow \cH_X \otimes \cH_B$, {\em safe} on $X$ if and only if there is a collection of isometries $V_x: \cH_A\rightarrow \cH_B$ such that for all states $\ket{\psi}_{XA} = \sum_x \alpha_x \ket{x}_X \ket{\psi^x}_A$ we have that
\begin{equation*}
    V  \ket{\psi}_{XA} =  \sum_x \alpha_x \ket{x}_X V_x \ket{\psi^x}_A.
\end{equation*}
\end{definition}

\begin{definition}[Extension] \label{def:extension} Let $$\rho_{XE}=  \sum\limits_{x \in \cX}  p(x)\ket{x}\bra{x} \otimes \rho^x_E,$$
be a c-q state. For a function $Z:\cX \rightarrow \cZ$, define the following extension of $\rho_{XE}$, 
\[ \rho_{ZXE} \defeq  \sum_{x\in \cX}  p(x) \ket{Z(x)}\bra{Z(x)} \otimes \ket{x}\bra{x} \otimes  \rho^{x}_E.\]
\end{definition}

For a pure state $\rho_{XEA}$ (with $X$ classical and $X \in \cX$) and a function $Z:\cX \rightarrow \cZ$, define $\rho_{Z\hat{Z}XEA}$ to be a pure state extension of $\rho_{XEA}$ generated via a safe isometry $V: \cH_X \rightarrow \cH_X \otimes \cH_Z \otimes \cH_{\hat{Z}}$ ($Z$ classical with copy $\hat{Z}$). We use the notation $\mathcal{M}_{A}(\rho_{AB})$ to denote measurement in the computational basis on register $A$ in state $\rho_{AB}$.

All isometries considered in this paper are safe on classical registers that they act on. Isometries  applied by adversaries can be assumed without loss of generality as safe on classical registers, by the adversary first making a (safe) copy of classical registers and then proceeding as before. This does not reduce the power of the adversary.
The following notion of a Stinespring isometry extension will also be useful at times.
 \begin{fact}[\protect{Stinespring isometry extension \cite[Theorem 5.1]{WatrousQI}}]\label{fact:stinespring}
 		 Consider a CPTP map $\Phi :    \mathcal{L} (\cH_X ) \rightarrow   \mathcal{L}(\cH_Y )$. 
		 There exists an isometry $V :  \cH_{X} \rightarrow   \cH_{Y} \otimes \cH_{Z}$ (called the \emph{Stinespring isometry extension} of $\Phi$) such that $\Phi(\rho_X)= \tr_{Z}(V \rho_X V^\dagger)$ for every state $\rho_X$.
 \end{fact} 

  \begin{fact}[\protect{Stinespring isometry extension for a classical map}]\label{fact:stinespringclassicalmap}
		 Let $\Phi :    \mathcal{L} (\cH_X ) \rightarrow   \mathcal{L}(\cH_Y )$ be a classical map such that for every $X=x$ we have a fixed function $f$ acting on register $Y$ such that $ \Pr(f(\Phi(x)) =x)=1$. 
        There exists an isometry $V_{\Phi} :  \cH_{X} \rightarrow   \cH_{Y} \otimes \cH_{\hat{Y}}$ (called the \emph{Stinespring isometry extension} of $\Phi$) such that $\Phi(\rho_X)= \tr_{\hat{Y}}(V \rho_X V^\dagger)$ for every state $\rho_X$. Furthermore, for a classical register $\rho_X$, we also have $\Pr(Y =\hat{Y})_{V \rho_X V^\dagger}=1$.
  \end{fact}

\subsubsection{Norms, trace distance, and divergences}

This section collects definitions of some important quantum information-theoretic quantities and related useful properties. 

\begin{definition}[Schatten $p$-norm]
    For $p\geq 1$ and a matrix $A$, the \emph{Schatten $p$-norm} of $A$, denoted by $\|A\|_p$, is defined as $\| A \|_p  \defeq (\tr(A^\dagger A)^{\frac{p}{2}})^{\frac{1}{p}}.$
\end{definition}

\begin{definition}[Trace distance]
    The \emph{trace distance} between two states $\rho$ and $\sigma$ is given by $\|\rho-\sigma\|_1$. We write $\rho\approx_\eps \sigma$ if $\|\rho-\sigma\|_1\leq \eps$.
\end{definition}

\begin{definition}[Fidelity]
    The \emph{fidelity} between two states $\rho$ and $\sigma$, denoted by $\fid(\rho,\sigma)$, is defined as
    \begin{equation*}
        \fid(\rho,\sigma)= \Vert  \sqrt{\rho} \sqrt{\sigma}\Vert_1.
    \end{equation*}
\end{definition}

\begin{definition}[Bures metric]  
The \emph{Bures distance} between two states $\rho$ and $\sigma$, denoted by $\Delta_B(\rho,\sigma)$, is defined as
\begin{equation*}
    \Delta_B(\rho,\sigma)= \sqrt{1-\fid(\rho,\sigma)}.
\end{equation*} 
\end{definition}

\begin{definition}[Max-divergence (\cite{Datta09}, see also~\cite{JainRS02})]\label{def:maxdiv}
    Given states $\rho$ and $\sigma$ such that $\supp(\rho) \subseteq \supp(\sigma)$, the \emph{max-divergence} between $\rho$ and $\sigma$, denoted by $\dmax{\rho}{\sigma}$, is defined as
    \begin{equation*}
        \dmax{\rho}{\sigma} =  \min\{ \lambda \in \mathbb{R} :   \rho  \leq 2^{\lambda} \sigma \}.
    \end{equation*}
\end{definition}

\begin{definition}[Max-information~\cite{Datta09}]\label{def:maxinfo}
  Given a state $\rho_{AB}$, the \emph{max-information} between $A$ and $B$, denoted by $\imax(A:B)_{\rho}$, is given by
  \begin{equation*}
       \imax(A:B)_{\rho} \defeq   \inf_{\sigma_{B}\in \mathcal{D}(\cH_B)}\dmax{\rho_{AB}}{\rho_{A}\otimes\sigma_{B}} .
  \end{equation*}
\end{definition}

We state a useful fact about max-information.
\begin{fact}[\protect{\cite[Fact 5]{BJL21T}}]\label{fact:boundnew}
    Let $\rho_{XBD}$ be a c-q state ($X$ classical) such that $\rho_{XB} = \rho_X \otimes \rho_B$. Then, $$ \imax(X:BD)_\rho \leq  2|D|.$$
\end{fact}

For the facts stated below without citation, we refer the reader to standard textbooks~\cite{NielsenC00,Wat18}. The following facts state some basic properties of trace distance.
\begin{fact}[Data-processing inequality]
\label{fact:data}
Let $\rho, \sigma$ be states and $\cE$ be a CPTP map. Then,
\begin{itemize}
    \item $ \Vert  \cE(\rho)  - \cE(\sigma)\Vert_1  \le \Vert \rho  - \sigma \Vert_1; $
    
     \item $\Delta_B ( \cE(\rho)  , \cE(\sigma))  \le \Delta_B (\rho  , \sigma).$     
\end{itemize}
The inequality above is an equality whenever $\cE$ is a CPTP map corresponding to an isometry.
\end{fact}

\begin{fact}\label{traceavg1}
Let $\rho,\sigma$ be states. Let $\Pi$ be a projector. Then,
\begin{equation*}
    \tr(\Pi \rho)  \left\|  \frac{\Pi \rho_{} \Pi}{\tr(\Pi \rho) }-  \frac{\Pi\sigma_{} \Pi}{\tr(\Pi \sigma)} \right\|_1 \leq \| \rho_{}-\sigma_{} \|_1.
\end{equation*}
\end{fact}

\begin{fact}\label{fact:traceconvex} Let $\rho, \sigma$  be states such that $\rho = \sum_{x} p_x \rho^x$,  $\sigma = \sum_{x} p_x \sigma^x$, $\{\rho^x, \sigma^x\}_x$ are states and $\sum_x p_x =1$. Then, 
\begin{equation*}
    \Vert\rho- \sigma \Vert_1 \leq \sum_x p_x \Vert \rho^x -\sigma^x \Vert_1.
\end{equation*}

\end{fact}

The next result is useful for the rejection sampling step in the analysis of our candidate code.
\begin{fact}[\protect{\cite{JainRS02}}]
\label{rejectionsampling}
Let $\rho_{A'B}, \sigma_{AB}$ be pure states such that
$\dmax{\rho_B}{\sigma_B} \leq k$. Let Alice and Bob share $\sigma_{AB}$. There exists an isometry $V: \cH_A \rightarrow \cH_{A'} \otimes \cH_C$ such that:
\begin{enumerate}
\item  $(V \otimes \id_B) \sigma_{AB}(V \otimes \id_B)^\dagger  = \phi_{A'BC}$, where $C$ is a single qubit register. 
\item Let $C$ be the outcome of measuring $\phi_C$ in the standard basis. Then $\Pr[C=1] \geq 2^{-k}$.
\item Conditioned on outcome $C=1$, the state shared between Alice and Bob is $\rho_{A'B}$.  
\end{enumerate} 
\end{fact}

The following fact connects the fidelity, trace distance, and the Bures metric.
\begin{fact}[\protect{Fuchs-van de Graaf inequalities~\cite{FvdG06} (see also~\cite[Theorem 4.10]{WatrousQI})}]
\label{fact:TracevsFidelityvsBures}
Let $\rho,\sigma$ be two states. 
Then,
\[  1-\fid(\rho,\sigma) \leq   \frac{1}{2} \Vert \rho -\sigma \Vert_1\leq \sqrt{ 1-\fid(\rho,\sigma)^2} \quad \textrm{and} \quad \Delta_B(\rho,\sigma)^2 \leq  \frac{1}{2} \Vert \rho -\sigma \Vert_1 \leq  \sqrt{2}\Delta_B(\rho,\sigma).  \]
\end{fact}
In the above fact, the second set of inequalities follows by noting that
\begin{equation*}
    1-\fid(\rho,\sigma)^2 \leq 2-2\fid(\rho,\sigma) = 2\Delta_B(\rho,\sigma)^2.
\end{equation*}

We now state Uhlmann's Theorem, which relates the  closeness of any two mixed states $\rho, \sigma$ to the closeness of two purification states $\ket{\rho} , \ket{\sigma}$, respectively. 
 \begin{fact}[Uhlmann's Theorem~\cite{uhlmann76}]\label{fact:uhlmann}
Let $\rho_A,\sigma_A\in \mathcal{D}(\cH_A)$ and let $\rho_{AB}\in \mathcal{D}(\cH_{AB})$ be a purification of $\rho_A$ and $\sigma_{AC}\in\mathcal{D}(\cH_{AC})$ be a purification of $\sigma_A$.
Then, there exists an isometry $V$ (from a subspace of $\cH_C$ to a subspace of $\cH_B$) such that
\begin{equation*}
    \Delta_B\left( \ketbra{\rho}_{AB}, \ketbra{\theta}_{AB}) =  \Delta_B(\rho_A,\sigma_A\right),
\end{equation*}
 where $\ket{\theta}_{AB} = (\id_A \otimes V) \ket{\sigma}_{AC}$.
\end{fact}

\subsubsection{Pauli and Clifford operators}

We proceed to define Pauli operators and the associated Pauli and Clifford groups.
\begin{definition}[Pauli operators]\label{def:pauli}
The single-qubit \emph{Pauli operators} are given by
\[ I = \begin{pmatrix} 1 & 0 \\ 0 & 1\end{pmatrix} \quad X = \begin{pmatrix} 0 & 1 \\ 1 & 0\end{pmatrix} \quad Y = \begin{pmatrix} 0 & -i \\ i & 0 \end{pmatrix} \quad Z = \begin{pmatrix} 1 & 0 \\ 0 & -1 \end{pmatrix}.\]

 An $n$-qubit Pauli operator is given by the $n$-fold tensor product of single-qubit Pauli operators. 
 We denote the set of all $\vert A \vert$-qubit Pauli operators on $\cH_A$ by  $\cP(\cH_A)$, where $\vert \cP(\cH_A)\vert =4^{\vert A \vert}$. Any linear operator $L \in \cL(\cH_A)$ can be written as a linear combination of $\vert A \vert$-qubit Pauli operators with complex coefficients as $L = \sum_{P \in \mathcal{P}(\cH_A)} \alpha_P P$. This is called the \emph{Pauli decomposition} of a linear operator.
\end{definition}
\begin{definition}[Pauli group]
The single-qubit \emph{Pauli group} is given by
\begin{equation*}
    \{ +P, -P, \ iP, \ -iP : P \in \{ I, X, Y, Z\} \}.
\end{equation*}
The Pauli group on $\vert A \vert$-qubits is the group generated by the operators described above applied to each of $\vert A \vert$-qubits in the tensor product. We denote the $\vert A \vert$-qubit Pauli group on $\cH_A$ by  $\tilde{\cP}(\cH_A)$.
    
\end{definition}
\begin{definition}[Clifford group]\label{def:clifford}
The \emph{Clifford group} $\mathcal{C}(\cH_A)$ is defined as the group of unitaries that normalize the Pauli group $\tilde{\cP}(\cH_A)$, i.e.,
\begin{equation*}
    \mathcal{C}(\cH_A) = \{ V \in \mathcal{U}(\cH_A) : V \tilde{\cP}(\cH_A) V^\dagger =\tilde{\cP}(\cH_A)\}.
\end{equation*}
The \emph{Clifford unitaries} are the elements of the Clifford group.

\end{definition}

We will also need to work with subgroups of the Clifford group with certain special properties.
The following fact describes these properties and guarantees the existence of such subgroups.
\begin{fact}[Subgroup of the Clifford group~\cite{CLLW16}]\label{lem:subclifford}
There exists a subgroup $\mathcal{SC}(\cH_A)$ of the Clifford group $\mathcal{C}(\cH_A)$ such that given any non-identity Pauli operators $P ,Q \in \cP(\cH_A)$ we have that
\begin{equation*}
    \vert \{ C \in \cSC(\cH_A) \vert C^\dagger P C =Q \} \vert = \frac{\vert \cSC(\cH_A)  \vert}{\vert\cP(\cH_A) \vert -1} \quad \textrm{and} \quad \vert \cSC(\cH_A)  \vert = 2^{5 \vert A \vert }-2^{3 \vert A \vert}.
\end{equation*} 
Informally, applying a random Clifford operator from  $\mathcal{SC}(\cH_A)$ (by conjugation) maps $P$ to a Pauli operator chosen uniformly at random over all non-identity Pauli operators. 
Furthermore, we have that $\cP(\cH_A)\subset \cSC(\cH_A)$.

Additionally, there exists a procedure $\samp$ which when given as input a uniformly random string $R\leftarrow \bits^{5|A|}$ outputs in time $\poly(|A|)$ the classical description $\samp(R)$ of a Clifford operator $C_R\in\mathcal{SC}(\cH_A)$ with the following property:
Let $U_{\mathcal{SC}(\cH_A)}$ denote the uniform distribution over the classical descriptions of Clifford operators in $\cSC(\cH_A)$.
Then, it holds that 
\begin{equation}\label{eq:approxsample}
    \samp(R) \approx_{2^{-2|A|}} U_{\mathcal{SC}(\cH_A)},
\end{equation}
where, as before, $\approx_{2^{-2|A|}}$ means that the statistical distance between the two distributions is at most $2^{-2|A|}$.

\end{fact}

\paragraph{Pauli twirling and related facts.}

The analysis of our construction will require the use of several facts related to Pauli twirling. We collect them below, beginning with the usual version of the Pauli twirl.
\begin{fact}[Pauli twirl~\cite{DCEL09}]\label{lem:paulitwirl}
 Let $\rho \in \cD( \cH_A)$ be a state and $P , P' \in \cP(\cH_A)$ be Pauli operators such that $P \ne P'$. Then,
$$\sum_{Q \in \cP(\cH_A)}   Q^\dagger P Q \rho Q^\dagger P'^\dagger Q  =0 .$$

\end{fact}

\begin{fact}[Subgroup Clifford twirl~\cite{BBJ23}]\label{lem:cliffordtwirl}
Let $\rho \in \cD( \cH_A)$ be a state and $P , P' \in \cP(\cH_A)$ be Pauli operators such that $P \ne P'$. Let $\cSC(\cH_A)$ be the subgroup of Clifford group as defined in \cref{lem:subclifford}. Then, 
    \begin{equation*}
        \sum_{C \in \cSC(\cH_A)}   C^\dagger P C \rho C^\dagger P'^\dagger C  = 0.
    \end{equation*}
  As an immediate corollary, we conclude that for any normal operator $M \in \cL( \cH_A)$ such that $M^\dagger M=MM^\dagger$ we have that
  \begin{equation*}
      \sum_{C \in \cSC(\cH_A)}   C^\dagger P C M C^\dagger P'^\dagger C  =0,
  \end{equation*}
since $M$ has an eigen-decomposition. 
\end{fact}

\begin{fact}[Modified twirl~\cite{BBJ23}]\label{lem:cliffordtwirl1}
 Let $\rho_{\hat{A}A}$ be the canonical purification of $\rho_A$. Let $P \in \cP(\cH_A),P' \in \cP(\cH_A)$ be Pauli operators such that $P \ne P'$. Let $\cSC(\cH_A)$ be the subgroup of Clifford group as defined in \cref{lem:subclifford}. Then, 
    $$ \sum_{C \in \cSC(\cH_A)}   (\id \otimes C^\dagger PC) \rho_{\hat{A}A} (\id \otimes C^\dagger P' C)  = 0.$$
\end{fact}

\begin{fact}[Uniform Pauli conjugation]\label{lem:paulitwirl2}
  Let $P \in \cP(\cH_A)$ be a Pauli operator. If $P=\id_A$, then 
  $$\frac{1}{\vert \cP(\cH_A) \vert} \sum_{Q \in \cP(\cH_A)}    Q P Q^\dagger  =\id_A ,$$
  else if $P \ne \id_A$
 $$\frac{1}{\vert \cP(\cH_A) \vert} \sum_{Q \in \cP(\cH_A)}    Q P Q^\dagger  =0 .$$  
\end{fact}
\begin{proof}
If $P=\id_A$, the statement trivially follows. If $P \ne \id_A$, the statement follows from \cref{lem:paulitwirl} by considering $(P,P',\rho)$ in \cref{lem:paulitwirl} to be $(P, \id_A, U_A)$.
\end{proof}
Similarly, the following result is a consequence of \cref{lem:cliffordtwirl}.
\begin{fact}[Uniform conjugation]\label{lem:cliffordtwirl2}
  Let $P \in \cP(\cH_A)$ be a Pauli operator. If $P=\id_A$, then 
  $$\frac{1}{\vert \cSC(\cH_A) \vert} \sum_{C \in \cSC(\cH_A)}    C P C^\dagger  =\id_A .$$
  Else, if $P \ne \id_A$, then
 $$\frac{1}{\vert \cSC(\cH_A) \vert} \sum_{C \in \cSC(\cH_A)}    C P C^\dagger  =0 .$$  
\end{fact}

The following fact is folklore.  
We provide a proof for completeness. 
\begin{fact}[Pauli $1$-design]\label{fact:bellbasis}
    Let $\rho_{AB}$ be a state. Then,
$$\frac{1}{\vert \cP(\cH_A) \vert} \sum_{Q \in \cP(\cH_A) } (Q \otimes \id)   \rho_{A B}  ( Q^\dagger \otimes \id  )  = U_{A } \otimes  \rho_B.$$
\end{fact}

\begin{proof}
Let 
\begin{align*}
   \rho_{AB}  & =\sum_{P \in \cP(\cH_{A}), Q \in \cP(\cH_{B})} \alpha_{PQ} (P \otimes Q)  \\
   &= \sum_{P \in \cP(\cH_{A})}  \left(P \otimes \left(\sum_{Q \in \cP(\cH_{B})} \alpha_{PQ} Q\right)\right)  \\
    & = \sum_{P \in \cP(\cH_{A})}  \left(P \otimes M^P\right),
\end{align*}
where $M^P \defeq \sum_{Q \in \cP(\cH_{B})} \alpha_{PQ} Q$.
It is easily seen that
 \begin{equation}\label{eq:qnm111}
     \rho_B=\tr_A(\rho_{AB}) =\tr(\id_A) M^{\id_A}. 
 \end{equation}
Then, we have that
\begin{align*}
& \frac{1}{\vert \cP(\cH_A)\vert } \sum_{Q \in \cP(\cH_A)} (Q  \otimes \id) \left(\sum_{P \in \cP(\cH_{A})}  \left(P \otimes M^P\right) \right)( Q^\dagger  \otimes \id)  \\ 
     &=  \sum_{P \in \cP(\cH_{A})}   \left( \frac{1}{\vert \cP(\cH_A)\vert } \sum_{Q \in \cP(\cH_A)} (Q P Q^\dagger  )   \otimes M^P \right)  \\ 
       &=  \sum_{P \in \cP(\cH_{A}) \setminus \id_A}   \left( \frac{1}{\vert \cP(\cH_A)\vert } \sum_{Q \in \cP(\cH_A)} (Q P Q^\dagger  )   \otimes M^P \right) \\& \quad \quad  + \sum_{P = \id_A}   \left( \frac{1}{\vert \cP(\cH_A)\vert } \sum_{Q \in \cP(\cH_A)} (Q P Q^\dagger  )   \otimes M^P \right)  \\ 
       & = \id_A \otimes M^{\id_A} & \mbox{(\cref{lem:paulitwirl2})} \\
       & = U_A \otimes \rho_B. & \mbox{(\cref{eq:qnm111})}
\end{align*}
\end{proof}
The proof of the next fact follows similarly to that of \cref{fact:bellbasis} using  \cref{lem:cliffordtwirl2}.
\begin{fact}[$1$-design]\label{fact:notequal}     Let $\rho_{AB}$ be a state. Let $\cSC(\cH_A)$ be the subgroup of Clifford group as defined in \cref{lem:subclifford}. Then,
$$\frac{1}{\vert \cSC(\cH_A) \vert} \sum_{C \in \cSC(\cH_A) } (C \otimes \id)   \rho_{A B}  ( C^\dagger \otimes \id  )  = U_{A } \otimes  \rho_B.$$
\end{fact}

\subsubsection{The transpose method}
The transpose method (see, e.g.,~\cite{MarisQI}) is one of the most important
tools for manipulating maximally entangled states. 
Note that the canonical purification of a maximally mixed state $\rho_A=U_A$, denoted $\rho_{A\hat{A}}$, is a maximally entangled state.
Roughly speaking, the transpose method corresponds to the statement that some local action on one half of the maximally entangled state (say register $A$) is equivalent to performing the transpose of the same action on the other half of that state (register $\hat{A}$). 
We now state this formally.
\begin{fact}[Transpose method]\label{fact:transposetrick}
    Let $\rho_{A\hat{A}}$ be the canonical purification of $\rho_A = U_A$. For any $M \in \cL(\cH_A)$ it holds that
$$(M \otimes \id_{\hat{A}}) \rho_{A\hat{A}} (M^\dagger \otimes \id_{\hat{A}}) =  (\id_{A} \otimes M^T) \rho_{A\hat{A}} (\id_A \otimes (M^T)^\dagger).$$ 
\end{fact}

\subsection{Quantum-secure randomness extractors}

Randomness extractors are key objects in our constructions of non-malleable codes and secret sharing schemes.
We introduce the relevant notions and auxiliary results here.

\subsubsection{Min-entropy and weak sources}

This section introduces information-theoretic formalism useful for our definitions of randomness extractors.

\begin{definition}[Min-entropy]\label{def:condminentropy}
    Given a state $\rho_{XE}$, the min-entropy of $X$ conditioned on $E$, denoted by $\hmin{X}{E}_\rho$, is defined as
    \begin{equation*}
        \hmin{X}{E}_\rho = - \inf_{\sigma_E \in  \mathcal{D}(\cH_{E}) } \dmax{\rho_{XE}}{\id_X \otimes \sigma_E}    .
    \end{equation*}
\end{definition}

We will require the following fact about the min-entropy.
\begin{fact}[\protect{\cite[Lemma 2.14]{CLW14}}]
	\label{fact:PostProcHmin}  
	 Let $\Phi:\mathcal{L} (\cH_M ) \rightarrow   \mathcal{L}(\cH_{M'} )$ be a CPTP map and define $\sigma_{XM'} =(\id \otimes \Phi)(\rho_{XM})$. 
    We have that
    \begin{equation*}
        \hmin{X}{M'}_\sigma  \geq \hmin{X}{M}_\rho.
    \end{equation*}
\end{fact}

The formalism of qpa-states from~\cite{BJK21} will be useful when building leakage-resilient secret sharing schemes.
\begin{definition}[$\qpas$]\label{qmadvk1k2}
A pure state $\sigma_{X\hat{X}Y\hat{Y}W_1 W_2}$ with $XY$ classical and $\hat{X}$ and $\hat{Y}$ copies of $X$ and $Y$, respectively, is  a  $(k_1,k_2)\mhyphen\qpas$ if
\begin{equation*}
    \hmin{X}{W_2 Y\hat{Y}}_\sigma \geq k_1 \quad \textrm{and} \quad \hmin{Y}{W_1 X\hat{X}}_\sigma \geq k_2.
\end{equation*}
\end{definition}

\subsubsection{Quantum-secure $2$-source extractors}

A $2$-source extractor transforms any two independent weak sources of randomness into an unbiased random string, provided that the input sources contain enough min-entropy.
We will use the fact that the inner product function $\IP:\F_q^N\times\F_q^N\to\F_q$ given by
\begin{equation*}
    \IP(x,y) = \sum_{i=1}^N x_i y_i,
\end{equation*}
with operations performed over $\F_q$, is a $2$-source extractor for qpa-states with good properties.
This is captured in the following fact.
\begin{fact}[\protect{\cite[Claim 5]{BJK21}}]\label{fact:IPqpa}
Let $\rho_{X \hat{X} Y \hat{Y} W_1 W_2}$ be an arbitrary $(k_1,k_2)$-$\qpas$ with $X,Y\in\F_q^N$ and
\begin{equation*}
    k_1+k_2\geq (N+1)\log q + 8\log(1/\eps)+40.
\end{equation*}
Set $Z=\IP(X,Y)$.
We have that
\begin{equation*}
    \|\rho_{Z X W_1} - U_Z  \otimes \rho_{X W_1} \|_1 \leq \eps \quad \textrm{and} \quad \|\rho_{Z Y W_2} - U_Z  \otimes \rho_{Y W_2} \|_1 \leq \eps,
\end{equation*}
where $U_Z$ is uniformly distributed over $\F_q$.
\end{fact}

\subsubsection{Quantum-secure $2$-source non-malleable extractors}

Roughly speaking, a $2$-source non-malleable extractor $\nmext$ takes as input two (not necessarily) uniformly random and independent strings $X$ and $Y$ and outputs a string $R=\nmext(X,Y)$ that is statistically close to uniform distribution.
It is called \emph{non-malleable} because learning $\nmext(f(X),g(Y))$ for any known tampering functions $f$ and $g$ (without fixed points) reveals essentially nothing about $R=\nmext(X,Y)$, in the sense that $R$ should still be close to uniformly random given the tampered version $\nmext(f(X),g(Y))$.
We may thus see $(X,Y)$ as a form of split-state encoding of $R$.
Since our split-state tampering adversaries have quantum capabilities \emph{and} access to shared quantum entanglement, we cannot use an arbitrary classical $2$-source non-malleable extractor to generate $R$.
Instead, we make use of an explicit \emph{quantum-secure} $2$-source non-malleable extractor recently constructed by Batra, Boddu, and Jain~\cite{BBJ23}, which remains secure against such quantum adversaries and whose properties we detail below.

\begin{fact}[Quantum-secure $2$-source non-malleable extractor~\cite{BBJ23}]\label{lem:qnmcodesfromnmext}
Consider the split-state tampering experiment in \cref{fig:splitstate5} with a split-state tampering adversary $\cA=(U,V,\ket\psi_{W_1 W_2})$.
Based on this figure, define $p_\sm=\Pr[(X,Y)=(X',Y')]_{\hat\rho}$ and the conditioned quantum states
\begin{equation*}
    \rhosame = (\nmext\otimes\nmext)(\hat\rho | (X,Y)=(X',Y'))
\end{equation*}
and
\begin{equation*}
    \rhotamp = (\nmext\otimes\nmext)(\hat\rho | (X,Y)\neq(X',Y')).
\end{equation*}
For any given constant $\delta>0$, there exists an explicit function $\nmext : \{ 0,1\}^n \times \{0,1 \}^{\delta n} \rightarrow \{0,1 \}^{r} $ with output length $r=(1/2-\delta)n$ such that for independent sources $X\leftarrow\bits^n$ and $Y \leftarrow \bits^{\delta n}$ and any such split-state tampering adversary $\cA=(U,V,\ket\psi_{W_1 W_2})$ it holds that
\begin{enumerate}
    \item $ \Vert  \nmext(X,Y)X - U_{r} \otimes U_n  \Vert_1 \leq \eps$ and $\Vert  \nmext(X,Y)Y - U_{r} \otimes U_{\delta n}  \Vert_1 \leq \eps,$

    \item $p_\sm\Vert \rhosame_{RW_2}-  U_r \otimes \rhosame_{W_2} \Vert_1 +  (1-p_\sm)  \Vert \rhotamp_{RR'W_2}-  U_r \otimes \rhotamp_{R'W_2}  \Vert_1  \leq \eps$,
\end{enumerate}
with $\eps =2^{-n^{\Omega_\delta(1)}}$. Furthermore, $\nmext(x,y)$ can be computed in time $\poly(n)$. 
\end{fact}

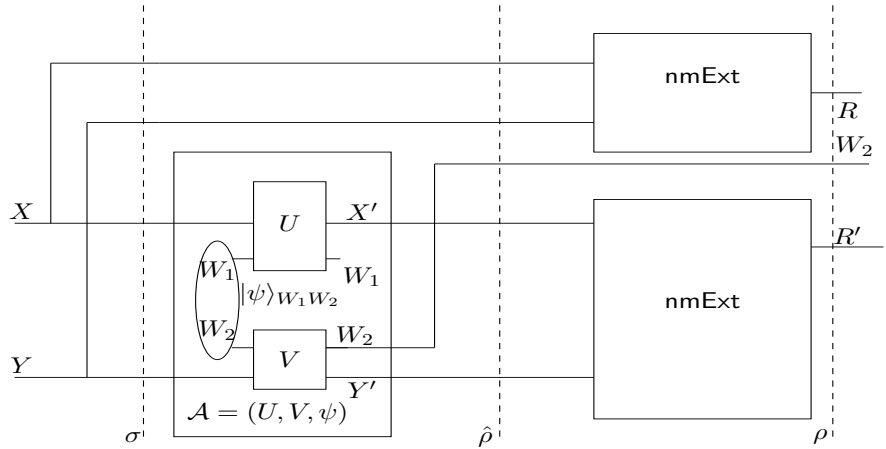
\begin{figure}[h]
\centering

\resizebox{12cm}{6cm}{

\begin{tikzpicture}

\draw (4,4.5) -- (5,4.5);
\draw (5,4.5) -- (11,4.5);

\draw (3.5,5.5) -- (5,5.5);
\draw (5,5.5) -- (11,5.5);
\node at (14.5,4.7) {$R$};
\node at (14.6,4.1) {$W_2$};

\node at (12.5,5.3) {$\nmext$};

\draw [dashed] (4.78,-0.8) -- (4.78,6.5);
\draw [dashed] (9.7,-0.8) -- (9.7,6.5);
\draw [dashed] (14.3,-0.8) -- (14.3,6.5);

\node at (6.8,1.6) {$\ket{\psi}_{W_1W_2}$};
\node at (4.64,-0.8) {$\sigma$};
\node at (9.5,-0.8) {$\hat{\rho}$};
\node at (14.13,-0.8) {$\rho$};

\node at (3.1,3) {$X$};
\node at (7.8,3) {$X'$};
\draw (3,2.8) -- (6.3,2.8);
\draw (7.3,2.8) -- (11,2.8);
\node at (14.5,2.6) {$R'$};
\draw (14,2.4) -- (15,2.4);

\node at (3.1,0.4) {$Y$};
\node at (7.8,0.0) {$Y'$};
\draw (3,0.2) -- (6.3,0.2);
\draw (7.3,0.2) -- (11,0.2);
\draw (3.5,2.8) -- (3.5,5.5);
\draw (4,0.2) -- (4,4.5);

\draw (14,5) -- (14.7,5);

\draw (7.3,0.7) -- (8.8,0.7);
\draw (8.8,0.7) -- (8.8,3.8);
\draw (8.8,3.8) -- (14.8,3.8);

\draw (6.3,2) rectangle (7.3,3.5);
\node at (6.8,2.8) {$U$};
\draw (6.3,0) rectangle (7.3,1);
\node at (6.8,0.5) {$V$};

\node at (6.5,-0.4) {$\mathcal{A}=(U,V,\psi)$};
\draw (5.2,-0.8) rectangle (8.2,4);

\draw (5.8,1.5) ellipse (0.3cm and 1cm);
\node at (5.8,2) {$W_1$};
\draw (6,2.2) -- (6.3,2.2);
\node at (7.8,1.9) {$W_1$};
\draw (7.3,2.2) -- (7.5,2.2);
\node at (5.8,1) {$W_2$};
\draw (6,0.7) -- (6.3,0.7);
\node at (7.7,0.9) {$W_2$};
\draw (7.3,0.7) -- (7.6,0.7);


\draw (11,4) rectangle (14,6);
\draw (11,-0.5) rectangle (14,3.2);

\node at (12.5,1.5) {${\nmext}$};

\end{tikzpicture} }
\caption{Split-state tampering experiment for quantum-secure $2$-source non-malleable extractors.}\label{fig:splitstate5}
\end{figure}

Intuitively, Item 1 in \cref{lem:qnmcodesfromnmext} guarantees that $R=\nmext(X,Y)$ remains close to uniformly random even when one of the input sources $X$ and $Y$ is revealed.
This property is usually called \emph{strong extraction}.
Item 2 spells out the non-malleability guarantees of $\nmext$: If the tampering attack does not change $X$ and $Y$ (i.e., $(X',Y')=(X,Y)$), then $R$ should be close to uniformly random even given one of the updated entangled states shared by the adversaries attacking each source.
On the other hand, if the tampering attack changed $X$ and $Y$ (i.e., $(X',Y')\neq (X,Y)$), then $R$ should be close to uniformly random even given the additional output $R'=\nmext(X',Y')$ \emph{and} one of the updated entangled states.

\subsection{Secret sharing schemes for quantum messages}

Secret sharing schemes for quantum messages were first introduced in~\cite{HBB99,CGL99,KKI99}.
As in the classical case, they allow a dealer to share a quantum message among $p$ parties in a way that only authorized subsets of parties can reconstruct the quantum message, while unauthorized subsets of parties learn nothing about it.
The no-cloning theorem imposes constraints on the type of access structures that can be realized by a secret sharing scheme for quantum messages.
For example, there exists a $t$-out-of-$p$ secret sharing scheme for quantum messages if and only if $t> p/2$.
Conversely, it is known that any access structure that does not violate no-cloning\footnote{This amounts to requiring that if $T\subseteq[p]$ is authorized, then its complement $[p]\setminus T$ must be unauthorized.} can be realized by a secret sharing scheme for quantum messages~\cite{Got00,Smi00}.

We proceed to define secret sharing schemes for quantum messages and then state a known efficient construction of such threshold schemes due to Cleve, Gottesman, and Lo~\cite{CGL99} (a quantum analogue of Shamir's secret sharing scheme). 
We define secret sharing schemes for threshold access structure for simplicity, but this definition can be generalized to any access structure that does not violate no-cloning.

Let 
$\qshare : \cL( \cH_M) \to \cL(\cH_{S_1}\otimes \cH_{S_2} \otimes \cdots \otimes \cH_{S_p})$ and  $\qrec  :  \cL( \bigotimes_{i \in T} \cH_{S_i}) \to  \cL(\cH_{M})$, be the sharing and reconstruction procedures respectively, where $\cL(\cH)$ is the space of all linear operators in the Hilbert space $\cH$.  The reconstruction procedure $\qrec$ acts on any authorized subset of shares $T$ to reconstruct the original message.

\begin{definition}[Threshold secret sharing scheme for quantum messages]
The coding scheme $(\qshare,\qrec)$ is said to be a \emph{$t$-out-of-$p$ secret sharing scheme for quantum messages} if for any quantum message ${\sigma_M} \in \cD(\cH_M)$ (with canonical purification  $\sigma_{M\hat{M}}$) the following two properties are satisfied:
\begin{itemize}
    \item \textbf{Correctness:} For any $T \subseteq [p]$ such that $ \vert T \vert \geq t$  it holds that
    \begin{equation*}
        \qrec (\qshare( \sigma_{M \hat{M}})_{S_{T}}) = \sigma_{M\hat{M}},
    \end{equation*}
    where we recall that $S_T=(S_i)_{i\in T}$.

    \item \textbf{Perfect privacy:} For any $T \subseteq [p]$ such that $ \vert T \vert \leq t-1$ it holds that
    \begin{equation*}
        (\qshare( \sigma))_{\hat{M}S_{T}}  \equiv \sigma_{\hat{M}} \otimes  \zeta_{S_{T}} .
    \end{equation*}
      where $\zeta_{S_{T}} $ is a fixed quantum state for every $\sigma_{M \hat{M}}$.
 \end{itemize}
\end{definition}

\begin{fact}[\cite{CGL99}]\label{fact:qss}
    For any number of parties $p$ and threshold $t$ such that $p/2<t\leq p$, there exists a $t$-out-of-$p$ secret sharing scheme $(\qshare,\qrec)$ for quantum messages of length $b$ with share size at most $2\max(p,b)$, where both the sharing and reconstruction procedures run in time $\poly(p,b)$. 
\end{fact}

The following results state some useful properties satisfied by the threshold secret sharing scheme for quantum messages from \cref{fact:qss}.
The first fact states that the shares corresponding to any unauthorized set of parties are jointly uniformly distributed.
\begin{fact}[\cite{CGL99}]\label{fact:qsshide}
    Let $(\qshare,\qrec)$  be the $t$-out-of-$p$ secret sharing scheme from \cref{fact:qss}.  Let $\sigma_{M}$ be a state such that $\vert M \vert =b$ and $\rho = \qshare(\sigma_{M})$. For any $T \subseteq [p]$ such that $ \vert T \vert \leq t-1$ it holds that
    \begin{equation*}
        \rho_{ S_T } = U_{S_T }.
    \end{equation*}
\end{fact}

\begin{lemma}\label{fact:qsshidepauli}
    Let $(\qshare,\qrec)$  be the $t$-out-of-$p$ secret sharing scheme from \cref{fact:qss}. Let $P \in \cP(\cH_M)$ be any non-identity Pauli operator.  For any $T \subseteq [p]$ such that $ \vert T \vert \leq t-1$ it holds that
    $(\qshare(P))_{S_T} =0$. 
\end{lemma}
\begin{proof}
Consider the spectral decomposition of $P$. Let $P = \sum_{i=1}^{2^{\vert M \vert}} a_i \ketbra{\psi}_i$. Note that half of the $a_i$'s are exactly equal to $1$, while the other half is $-1$. 
Since $\qshare(\ketbra{\psi}_i)_{S_T} = U{S_T}$ for every $i$ in $[2^{\vert M \vert}]$, we arrive at the desired result that $(\qshare(P))_{S_T} = 0.$
\end{proof}

The next lemma captures privacy properties in the presence of entangled quantum side information. 
\begin{lemma}\label{fact:qsshiding}
Let $(\qshare,\qrec)$  be the $t$-out-of-$p$ secret sharing scheme from \cref{fact:qss}. 
    Let $\sigma_{EM}$ be a state such that $\vert M \vert =b$ and $\rho = \qshare(\sigma_{EM})$. For any $T \subseteq [p]$ such that $ \vert T \vert \leq t-1$ it holds that
    \begin{equation*}
        \rho_{E S_T } \equiv \rho_E \otimes U_{S_T }\equiv \sigma_E \otimes U_{S_T }.
    \end{equation*}
\end{lemma}
\begin{proof} 
For each $Q \in \cP(\cH_{M})$ and $Q' \in \cP(\cH_{E})$ there exist coefficients $\alpha_{Q' Q}$ such that we may write
\begin{align*}
   \sigma_{EM}  & =\sum_{Q \in \cP(\cH_{M}),Q' \in \cP(\cH_{E})} \alpha_{Q' Q} (Q'\otimes Q)\\
   &= \sum_{Q \in \cP(\cH_{M})}  \left(\left(\sum_{Q' \in \cP(\cH_{E})} \alpha_{Q' Q} Q'\right) \otimes Q \right)  \\
    & = \sum_{Q \in \cP(\cH_{M})}  \left(G^{Q} \otimes Q\right),
\end{align*}
where
\begin{equation*}
    G^Q \defeq \sum_{Q' \in \cP(\cH_{E})} \alpha_{Q'Q} Q'.
\end{equation*}
In particular, this implies that
\begin{equation}\label{eq:qss}
    \sigma_E = G^{\id_M} 2^{\vert M \vert}.
\end{equation}
Then, we conclude that 
\begin{align*}
\rho_{E S_T }  
& =  \sum_{Q \in \cP(\cH_{M})}  \left(G^{Q} \otimes  \left(\qshare(Q)\right)_{S_T }\right) \\
&=  \sum_{Q = \id_M}  \left(G^{Q} \otimes  \left(\qshare(Q)\right)_{S_T }\right) + \sum_{Q \ne \id_M \land Q \in \cP(\cH_{M})}  \left(G^{Q} \otimes  \left(\qshare(Q)\right)_{S_T }\right) \\ 
&=  G^{\id_M} \otimes \left(\qshare \left(  2^{\vert M \vert} \cdot \frac{\id_M}{2^{\vert M \vert}}  \right)\right)_{S_T } \\ 
     &= G^{\id_M} \otimes 2^{\vert M \vert} U_{S_T}\\
     & = \sigma_E \otimes U_{S_T }.
\end{align*}
The third equality above follows from \cref{fact:qsshidepauli}.
The fourth equality holds because of \cref{fact:qsshide}.
The final equality is a consequence of \cref{eq:qss}.
\end{proof}

\section{Split-state non-malleable codes for quantum messages}\label{sec:qnmc67}
In this section, we describe and analyze our split-state non-malleable coding scheme for quantum messages. We begin by describing the encoding and decoding procedures $\enc$ and $\dec$. In the analysis, we first show that $(\enc,\dec)$ is an average-case non-malleable code with low error for quantum messages. This yields \cref{thm:mainqnmc1}. Then, to conclude the argument and obtain \cref{thm:mainqnmc}, we show that every such average-case non-malleable code is also \emph{worst-case} split-state non-malleable at the price of a blow-up in the error as a function of the input quantum message length.

\subsection{Our candidate coding scheme for quantum messages}\label{sec:codedesc}

We proceed to describe our explicit candidate coding scheme for quantum messages.
Suppose that we wish to encode a quantum state $\sigma_M$ with canonical purification $\sigma_{M\hat{M}}$.
Let $b=|M|$ denote the message length and fix an arbitrary constant $\delta>0$.
We will invoke the explicit quantum-secure $2$-source $\eps$-non-malleable extractor $\nmext:\bits^\ell\times\bits^{\delta \ell}\to\bits^r$ guaranteed by \cref{lem:qnmcodesfromnmext} with output length $r$ satisfying
\begin{equation}\label{eq:constraintOutLength}
    r=(1/2-\delta)\ell \geq 5b
\end{equation}
and error 
\begin{equation*}
    \eps=2^{-\ell^{\Omega_\delta(1)}},
\end{equation*}
where $\Omega_\delta(\cdot)$ hides constants which depend only on $\delta$.

The encoding CPTP map $\enc$ works as follows on input $\sigma_{M}$:
\begin{enumerate}
    \item Sample classical bitstrings $X\leftarrow\bits^\ell$ and $Y\leftarrow \bits^{\delta \ell}$;
    
    \item Compute the classical key $R=\nmext(X,Y)\in\bits^r$;
    
    \item Let $\mathcal{SC}(\cH_M)$ be the subgroup of the Clifford group described in \cref{lem:subclifford} and $\samp$ be the associated sampling procedure.
    Compute $C_R=\samp(R)$ and the masked state $\psi_Z=C_R\sigma_M C^\dagger_R$.
    Note that the constraint in \cref{eq:constraintOutLength} guarantees that $R$ is long enough to sample such a Clifford operator as per \cref{lem:subclifford}.
    
    \item Output registers $X$ and $(Y,Z)$ as the two parts of the split-state encoding. Note that $X$ and $Y$ are classical strings while $Z$ is a quantum state.
\end{enumerate}

It is clear that $\enc$ can be computed efficiently if $\nmext$ is explicit. The decoding procedure $\dec$ is straightforward, and proceeds as follows on input possibly tampered registers $(X,Y) \to (X', Y')$ and $\psi_Z \to \tau_Z$:
\begin{enumerate}
    \item Compute the candidate key $R'=\nmext(X',Y')$;
    
    \item Let $\tau_{Z}$ denote the possibly tampered quantum state stored in register $Z$. Then, compute the candidate message $\eta_{M}=C_{R'}^\dagger \tau_{Z}C_{R'}$ where, as in the encoding procedure $\enc$ above, $C_{R'}=\samp(R')$ with $\samp$ the sampling procedure corresponding to the subgroup $\mathcal{SC}(\cH_M)$ of the Clifford group from \cref{lem:subclifford};
    
    \item Output $\eta_{M}$.
\end{enumerate}
It is easy to check that this coding scheme $(\enc,\dec)$ satisfies the basic correctness property
\begin{equation*}
    \dec(\enc(\sigma_{M\hat{M}})) = \sigma_{M\hat{M}}.
\end{equation*}

\subsection{Average-case non-malleability}

In this section, we show that the coding scheme $(\enc,\dec)$ described in \cref{sec:codedesc} is average-case non-malleable with low-error.
More precisely, we prove the following result.
\begin{theorem}[Average-case non-malleable codes with constant rate]\label{thm:avgcaseNMC}
    For any fixed constant $\delta>0$,
    the coding scheme $(\enc,\dec)$ described in \cref{sec:codedesc} with codewords of length $n$ and message size $b\leq \left(\frac{1}{11}-\delta\right)n$ is average-case $\eps'$-non-malleable with $\eps'=2^{-n^{\Omega_\delta(1)}}$.
    Moreover, both $\enc$ and $\dec$ can be computed in time $\poly(n)$.
\end{theorem}

The claim about the computational complexity of $\enc$ and $\dec$ in \cref{thm:avgcaseNMC} can be easily verified from the description in \cref{sec:codedesc}, using the fact that the codeword length $n$ satisfies $n=O(\ell)$.

We prove the remainder of \cref{thm:avgcaseNMC} through a sequence of lemmas.
Throughout this proof we assume that $C_R$, the Clifford operator sampled from $R$ via the sampling procedure $\samp$ from \cref{lem:subclifford} used in Step 3 of $\enc(\sigma)$ in \cref{sec:codedesc}, is uniformly distributed over $\cSC(\cH_M)$.
Based on \cref{eq:approxsample} in \cref{lem:subclifford}, this assumption will lead to an additive factor of $2^{-2b}$ in the final non-malleability error, which we add at the end of our argument.

For ease of readability, a detailed diagram of the complete split-state tampering experiment on $(\enc,\dec)$ in \cref{fig:splitstate2}.

Taking into account the notation from \cref{fig:splitstate2}, in order to prove \cref{thm:avgcaseNMC} it suffices to show that for any split-state adversary $\cA=(U,V,\psi)$ and $\sigma_{M\hat{M}}$ maximally entangled it holds that
\begin{equation}\label{eq:finalgoal}
    (\sigma_3)_{\hat{M}M} \approx_{\eps'=2^{-n^{\Omega_\delta(1)}}} p_{\mathcal{A}} \sigma_{\hat{M}M}  + (1-p_\mathcal{A}) (U_{\hat{M}} \otimes  \gamma^{\mathcal{A}}_{M}),
\end{equation}
where $(p_{\mathcal{A}}, \gamma^{\mathcal{A}}_{M})$ depend only on the split-state adversary $\cA$.

We begin by noting that in order to establish \cref{eq:finalgoal} we can equivalently consider the modified tampering experiment described in \cref{fig:splitstate3}, where the (transposed) Clifford operator is applied to $\sigma_{\hat{M}}$ instead.
More precisely, we have the following lemma.
\begin{lemma}\label{lem:modtamptranspose}
    For any fixed split-state adversary $\cA$ it holds that the states $\rho,\rho_1,\rho_2,\rho_3$ in \cref{fig:splitstate3} are equal to $\sigma,\sigma_1,\sigma_2,\sigma_3$ in \cref{fig:splitstate2}.
\end{lemma}
\begin{proof}
    The desired statement follows if we show that $\sigma_1$ and $\rho_1$ are equal.
     This is a direct consequence of the transpose method (\cref{fact:transposetrick}).
\end{proof}

We can further note that delaying the generation of $R$ and the application of the corresponding Clifford operator $C^T_R$ on the $\hat{M}$ register until the very end has no effect on the final state.
Therefore, we can focus on the modified tampering experiment described in \cref{fig:splitstate4}.
We capture this formally in the next brief lemma statement.
\begin{lemma}\label{lem:modtampdelay}
    For any fixed split-state adversary $\cA$ it holds that the final state $\theta_4$ in \cref{fig:splitstate4} is equal to the final state $\rho_3$ in \cref{fig:splitstate3}.
\end{lemma}

In particular, combining \cref{lem:modtamptranspose,lem:modtampdelay} implies that to establish \cref{eq:finalgoal} it suffices to show that for any split-state adversary $\cA=(U,V,\psi)$ and $\theta_{M\hat{M}}$ maximally entangled it holds that
\begin{equation}\label{eq:finalgoaltranspose}
    (\theta_4)_{\hat{M}M} \approx_{\eps'} p_{\mathcal{A}} \theta_{\hat{M}M}  + (1-p_\mathcal{A}) (U_{\hat{M}} \otimes  \gamma^{\mathcal{A}}_{M}),
\end{equation}
where $(p_{\mathcal{A}}, \gamma^{\mathcal{A}}_{M})$ depend only on the split-state adversary $\cA$.
Here, the states $\theta,\theta_1,\theta_2,\theta_3,\theta_4$ correspond to the intermediate states of the modified tampering experiment in \cref{fig:splitstate4}.

We now set up some helpful definitions before proceeding to the next lemma.
We may write $\theta_2$ as
\begin{equation*}
    \theta_2 = (U \otimes V)  (\theta_1 \otimes \ketbra{\psi}_{W_1W_2}) (U \otimes V)^\dagger.
\end{equation*}

Our analysis will proceed by cases, depending on whether the $X$ and $Y$ registers are modified by the tampering experiment in \cref{fig:splitstate4} (i.e., $XY\neq X' Y'$) or not.
To this end, we consider two different conditionings of $\theta_2$ based on these two cases.
More precisely, taking into account \cref{def:conditioning}, we define
\begin{equation*}
    \theta^{\mathsf{tamp}}_2= \theta_2 \vert  (XY \ne X'Y')
\end{equation*}
and
\begin{equation*}
    \theta^{\mathsf{same}}_2= \theta_2 \vert  (XY = X'Y').
\end{equation*}
Note that we may write $\theta_2 = p_{\sm} \theta^{\mathsf{same}}_2+(1-p_{\sm})\theta^{\mathsf{tamp}}_2$, where $p_{\sm}= \Pr[ XY= X'Y' ]_{\theta_2}$ is the probability that $XY$ are not modified in the tampering experiment from \cref{fig:splitstate4}.

Further taking into account that
\begin{equation*}
    \theta_3 = (\nmext_{XY} \otimes \nmext_{X'Y'})\theta_2,
\end{equation*}
let
\begin{equation*}
    \thetatamp_3 = (\nmext_{XY} \otimes \nmext_{X'Y'})\thetatamp_2
\end{equation*}
and
\begin{equation*}
    \thetasame_3 = (\nmext_{XY} \otimes \nmext_{X'Y'})\thetasame_2.
\end{equation*}

Let $D_{R\hat{M}}$ denote the controlled Clifford operator $C^T_{R}$ acting on register $\hat{M}$. Similarly, let $\tilde{D}_{R'Z}$ denote the controlled Clifford operator $C^\dagger_{R'}$ acting on register $Z$. 
We can then write
\begin{equation*}
    (\theta_4)_{\hat{M}M} =  (D_{R\hat{Z}} \otimes \tilde{D}_{R'Z} )  (\theta_3)_{RR'Z\hat{M}}    (D^\dagger_{R\hat{Z}} \otimes \tilde{D}^\dagger_{R'Z} ).
\end{equation*}
Analogously to the previous cases, we define
 \begin{equation}\label{eq:17thm1}
    (\thetatamp_4)_{\hat{M}M} \defeq  (D_{R\hat{Z}} \otimes \tilde{D}_{R'Z} )  (\thetatamp_3)_{RR'Z\hat{M}}    (D^\dagger_{R\hat{Z}} \otimes \tilde{D}^\dagger_{R'Z} )    
 \end{equation}
 and
  \begin{equation}\label{eq:17fvthm1}
    (\thetasame_4)_{\hat{M}M} \defeq  (D_{R\hat{Z}} \otimes \tilde{D}_{R'Z} )  (\thetasame_3)_{RR'Z\hat{M}}    (D^\dagger_{R\hat{Z}} \otimes \tilde{D}^\dagger_{R'Z} ).     
 \end{equation}
    
Recall from \cref{eq:finalgoaltranspose} that, based on \cref{lem:modtamptranspose,lem:modtampdelay}, our end goal is to show that
\begin{equation*}
    (\theta_4)_{\hat{M}M} \approx_{\eps'} p_{\mathcal{A}} \theta_{\hat{M}M}  + (1-p_\mathcal{A}) (U_{\hat{M}} \otimes  \gamma^{\mathcal{A}}_{M})
\end{equation*}
for some quantity $p_\cA\in[0,1]$ and state $\gamma^{\mathcal{A}}_{M}$ depending only on the split-state adversary $\cA$.
Towards establishing this, we begin by writing
\begin{equation*}
    (\theta_4)_{\hat{M}M} = p_{\sm} (\thetasame_4)_{\hat{M}M} + (1-p_{\sm}) (\thetatamp_4)_{\hat{M}M}.
\end{equation*}
We will focus on each of the two terms on the right-hand side of this equation separately.
We invoke the two lemmas below, whose proofs we defer to later dedicated sections.

The first lemma is relevant for handling $\thetatamp_4$, and intuitively states that when tampering occurs and $X'Y'\neq XY$ the outcome of the tampering experiment is close (in trace distance) to being an unentangled message.
\begin{lemma}\label{lem:thetatamp}
    We have that
    \begin{equation*}
        (1-p_{\sm})  \left\|    (\thetatamp_4)_{\hat{M}M} -U_{\hat{M}} \otimes    (\thetatamp_4)_{M} \right\|_1 \leq  \eps.
    \end{equation*}
\end{lemma}

The second lemma is relevant for handling $\thetasame_4$, and intuitively states that when $X'Y'= XY$ the outcome of the tampering experiment is close (in trace distance) to being either the original message or a completely independent and unentangled message.
\begin{lemma}\label{lem:thetasame}
    There exists a constant $p_\epr\in[0,1]$ depending only on the split-state adversary $\cA$ such that
    \begin{equation*}
        p_{\sm}\cdot \left\|    (\thetasame_4)_{\hat{M}M}-\left( p_\epr\cdot \theta_{\hat{M}M} + (1-p_\epr) (U_{\hat{M}} \otimes U_{M}) \right)\right\|_1 \leq  \eps+2 \cdot 4^{-b}.
    \end{equation*}
\end{lemma}

Proofs of \cref{lem:thetatamp,lem:thetasame} can be found in \cref{sec:thetatamp,sec:thetasame}, respectively.
We now show how to wrap up the argument and establish \cref{eq:finalgoaltranspose} using \cref{lem:thetatamp,lem:thetasame}.
We have
\begin{align}
    (\theta_4 )_{\hat{M}M} &= p_{\sm}  ( \thetasame_4 )_{\hat{M}M}  + (1-p_{\sm})  ( \thetatamp_4 )_{\hat{M}M}\nonumber\\
     & \approx_{\eps + 2 \cdot 4^{-b}} p_{\sm}  \left( p_\epr \theta_{M\hat{M}} + (1-p_\epr) (U_{\hat{M}} \otimes U_{M}) \right)+(1-p_{\sm})(\thetatamp_4)_{\hat{M}M}\label{eq:applyLemThetasame}\\
     &\approx_{\eps} p_{\sm}  \left( p_\epr\cdot \theta_{M\hat{M}} + (1-p_\epr) (U_{\hat{M}} \otimes U_{M}) \right)+(1-p_{\sm})(U_{\hat{M}} \otimes (\thetatamp_4)_{M}),\label{eq:applyLemThetatamp}
\end{align}
where \cref{eq:applyLemThetasame} follows from \cref{lem:thetasame} and \cref{eq:applyLemThetatamp} follows from \cref{lem:thetatamp}.
Combining \cref{eq:applyLemThetasame,eq:applyLemThetatamp} with a triangle inequality shows that
\begin{align}
    (\theta_4 )_{\hat{M}M} &\approx_{2(\eps+4^{-b})}p_{\sm}  \left( p_\epr\cdot \theta_{M\hat{M}} + (1-p_\epr) (U_{\hat{M}} \otimes U_{M}) \right)+(1-p_{\sm})(U_{\hat{M}} \otimes (\thetatamp_4)_{M})\nonumber\\
    &= p_{\sm}\cdot p_\epr \cdot \theta_{M\hat{M}}\nonumber \\
    &+ (1-p_{\sm}\cdot p_\epr)\left(U_{\hat{M}}\otimes\left(\frac{p_\sm(1-p_\epr)}{1-p_{\sm}\cdot p_\epr}\cdot U_{M}+\frac{1-p_\sm}{1-p_{\sm}\cdot p_\epr}\cdot (\thetatamp_4)_{M}\right)\right).\label{eq:finalbeforesetting}
\end{align}
Consider setting
\begin{equation*}
    p_{\cA} = p_\sm\cdot p_\epr
\end{equation*}
and
\begin{equation*}
    \gamma^\cA_{M} = \frac{p_\sm(1-p_\epr)}{1-p_{\sm}\cdot p_\epr}\cdot U_{M}+\frac{1-p_\sm}{1-p_{\sm}\cdot p_\epr}\cdot (\thetatamp_4)_{M}
\end{equation*}
in \cref{eq:finalbeforesetting}. 
We claim that $p_\cA$ and $\gamma^\cA_{M}$ depend only on the split-state adversary $\cA$.
This holds because:
\begin{itemize}
    \item $p_{\sm}= \Pr[ XY= X'Y' ]_{\theta_2}$ is a function of $\cA$ and $(X,Y)$, which are sampled independently of the message $\sigma_M$;
    
    \item \cref{lem:thetasame} guarantees that $p_\epr$ is a function of $\cA$ only;
    
    \item The state $(\thetatamp_4)_{M}$ is a function of $(\thetatamp_3)_{ZX'Y'}$, and one can prepare the latter by running an independent tampering experiment on registers $XYZ=U_l \otimes U_{\delta l} \otimes U_Z$.
    
\end{itemize}

Therefore, we conclude that both $p_\cA$ and $\gamma^\cA_{M}$ only depend on $\cA$.
In this case, we can write
\begin{equation*}
    (\theta_4 )_{\hat{M}M} \approx_{2(\eps+4^{-b})} p_\cA \cdot \theta_{M\hat{M}}
    + (1-p_\cA)\left(U_{\hat{M}}\otimes\gamma^\cA_{M}\right).
\end{equation*}
By \cref{lem:modtamptranspose,lem:modtampdelay}, this implies that we have
\begin{equation}\label{eq:finalfinal}
    (\sigma_3)_{\hat{M}M} \approx_{2(\eps+4^{-b})} p_\cA \cdot \sigma_{M\hat{M}}
    + (1-p_\cA)\left(U_{\hat{M}}\otimes\gamma^\cA_{M}\right)
\end{equation}
in the original tampering experiment from \cref{fig:splitstate2}.

\paragraph{Setting parameters for \cref{thm:avgcaseNMC}.}
To conclude the proof of \cref{thm:avgcaseNMC}, it remains to argue that, given any constant $\delta>0$, the coding scheme $(\enc,\dec)$ is average-case $(\eps'=2^{-n^{\Omega_\delta(1)}})$-non-malleable for message length $b\leq (\frac{1}{11}-\delta)n$.

Note that the total codeword length is
\begin{equation*}
    n=|X|+|Y|+|Z|=\ell+\delta \ell+b=(1+\delta+1/10+\delta/5)\ell
\end{equation*}
and that the error $\eps$ of $\nmext$ satisfies $\eps=2^{-\ell^{\Omega_\delta(1)}}=2^{-n^{\Omega_\delta(1)}}$, since $n=O(\ell)$.
Furthermore, recall from \cref{eq:constraintOutLength} in \cref{sec:codedesc} that we may set the message length to be as large as $r/5=\frac{(1/2-\delta)\ell}{5}$.
In fact, even if $b<r/5$ we may always extend a message so that its length is exactly $b'=r/5$ by appending dummy independent uniformly random bits to it.
Since the length-$b$ message $\sigma_{M}$ is maximally mixed, the resulting extended length-$b'$ message will also be maximally mixed. 
Therefore, taking into account \cref{eq:finalfinal} and recalling that our argument in this section assumed that the Clifford operator $C_R$ is sampled uniformly at random from $\cSC(\cH_A)$ while \cref{lem:subclifford} only guarantees that $C_R\approx_{2^{-b'}}U_{\cSC(\cH_A)}$ with $U_{\cSC(\cH_A)}$ uniformly distributed over $\cSC(\cH_A)$,  the resulting coding scheme is average-case $\eps'$-non-malleable with
\begin{equation*}
    \eps'=2(\eps+4^{-b'})+2^{-2b'}=2^{-n^{\Omega_\delta(1)}},
\end{equation*}
as desired.

Finally, simple algebraic manipulation yields
\begin{equation*}
    r/5 = \frac{(1/2-\delta)\ell}{5}\geq \left(\frac{1}{11}-c\right)n
\end{equation*}
for some constant $c=O(\delta)$.
This means that we may encode messages of length at least $\left(\frac{1}{11}-c\right)n$, and the constant $c>0$ can be made arbitrarily small.
This concludes the proof of \cref{thm:avgcaseNMC}.

\begin{figure}
\centering
\resizebox{12cm}{6cm}{
\begin{tikzpicture}

\draw (1,5.5) ellipse (0.3cm and 1.5cm);
\draw (3,3.9) rectangle (4.4,5.4);
\draw (12.6,3.9) rectangle (14,5.4);
\node at (1,4.5) {$M$};

\node at (3.6,4.6) {$C_R$};
\node at (13.3,4.6) {$C^\dagger_{R'}$};
\node at (1,6.5) {$\hat{M}$};
\node at (14.8,6.5) {$\hat{M}$};

\draw (1.2,6.5) -- (14.6,6.5);

\draw (1.2,4.5) -- (3,4.5);
\draw (4.4,4.5) -- (6.3,4.5);
\draw (7.3,4.5) -- (12.6,4.5);
\draw (14,4.5) -- (15,4.5);
\node at (14.8,4.7) {$M$};
\node at (0.8,2) {$Y$};
\node at (0.8,1) {$X$};
\node at (3.4,1.2) {$R$};
\draw (1,1) -- (1.8,1);
\draw (1,2) -- (1.8,2);
\draw (3.2,1.5) -- (4,1.5);
\draw (4,1.5) -- (4,3.9);
\draw (1.2,0.2) -- (4.6,0.2);
\draw (1.2,2.8) -- (4.6,2.8);
\draw (1.2,0.2) -- (1.2,1);
\draw (1.2,2) -- (1.2,2.8);
\draw (1.8,0.5) rectangle (3.2,2.5);
\draw (1.5,-0.5) rectangle (4.8,5.9);
\node at (3.3,6.1) {$\enc$};
\node at (2.5,1.5) {$\nmcenc$};

\draw (9.4,-0.5) rectangle (14.2,5.9);
\node at (11.5,6.1) {$\dec$};


\draw [dashed] (1.4,-1.4) -- (1.4,7.2);
\draw [dashed] (4.98,-1.4) -- (4.98,7.2);
\draw [dashed] (9.1,-1.4) -- (9.1,7.2);
\draw [dashed] (14.4,-1.4) -- (14.4,7.2);


\node at (6.8,1.6) {$\ket{\psi}_{W_1W_2}$};
\node at (1.25,-1.4) {$\sigma$};
\node at (4.8,-1.4) {$\sigma_1$};
\node at (8.9,-1.4) {$\sigma_2$};
\node at (14.2,-1.4) {$\sigma_3$};

\node at (4.6,4.7) {$Z$};
\node at (8.5,4.7) {$Z$};
\node at (4.6,3) {$Y$};
\node at (8.5,3) {$Y'$};
\draw (4.5,2.8) -- (6.3,2.8);
\draw (7.3,2.8) -- (10.5,2.8);

\node at (12.7,2.6) {$R'$};
\draw (12.5,2.4) -- (13.5,2.4);
\draw (13.5,2.4) -- (13.5,3.9);

\node at (4.6,0.4) {$X$};
\node at (8.5,0.0) {$X'$};
\draw (4.5,0.2) -- (6.3,0.2);
\draw (7.3,0.2) -- (10.5,0.2);

\draw (6.3,2) rectangle (7.3,5);
\node at (6.8,3.5) {$V$};
\draw (6.3,0) rectangle (7.3,1);
\node at (6.8,0.5) {$U$};

\node at (6.5,-0.4) {$\mathcal{A}=(U,V,\psi)$};
\draw (5.2,-0.8) rectangle (8.2,5.3);

\draw (5.8,1.5) ellipse (0.3cm and 1cm);
\node at (5.8,2) {$W_2$};
\draw (6,2.2) -- (6.3,2.2);
\node at (7.7,2) {$W_2$};
\draw (7.3,2.2) -- (7.5,2.2);
\node at (5.8,1) {$W_1$};
\draw (6,0.7) -- (6.3,0.7);
\node at (7.7,0.9) {$W_1$};
\draw (7.3,0.7) -- (7.45,0.7);



\draw (10.5,-0.2) rectangle (12.5,3.0);
\node at (11.5,1.5) {${\nmcdec}$};

\end{tikzpicture}}

    \caption{
        Split-state tampering experiment.
    }\label{fig:splitstate2}
\end{figure}
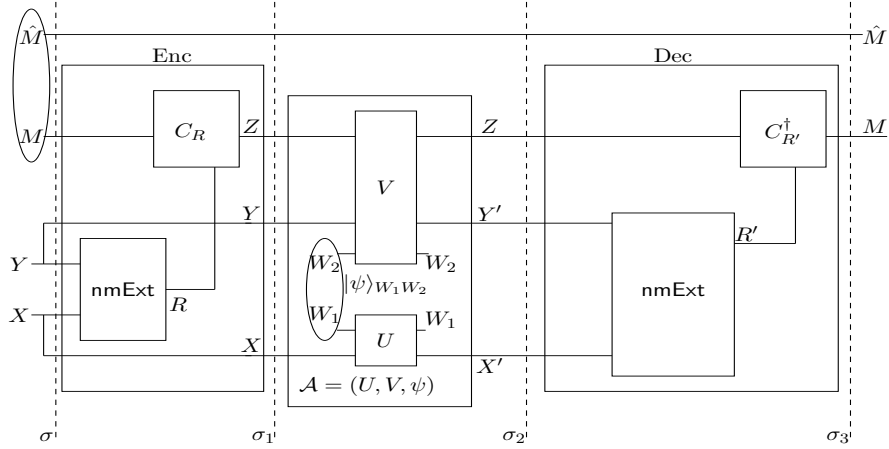

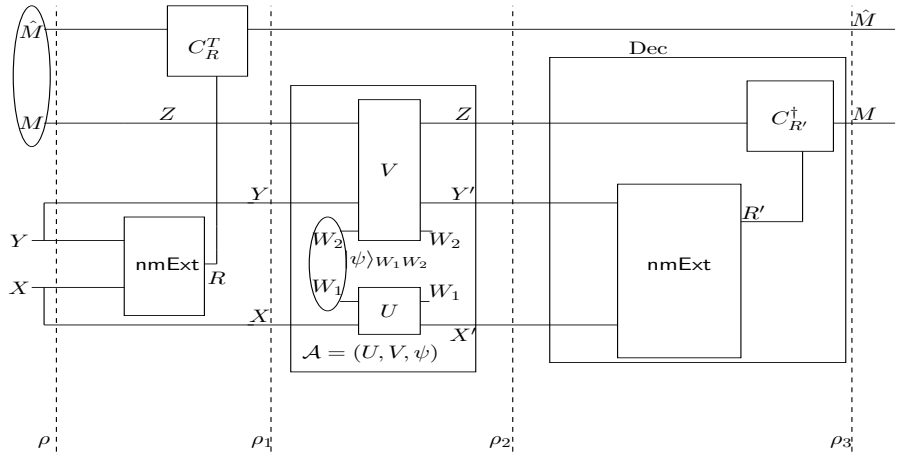
\begin{figure}
\centering
\resizebox{12cm}{6cm}{
\begin{tikzpicture}

\draw (1,5.5) ellipse (0.3cm and 1.5cm);
\draw (3.2,5.5) rectangle (4.5,7);
\draw (12.6,3.9) rectangle (14,5.4);
\node at (1,4.5) {$M$};
\node at (3.8,6.1) {$C^T_{R}$};
\node at (13.3,4.6) {$C^\dagger_{R'}$};
\node at (1,6.5) {$\hat{M}$};
\node at (14.5,6.75) {$\hat{M}$};
\draw (1.2,6.5) -- (3.2,6.5);
\draw (4.5,6.5) -- (15,6.5);

\draw (2.5,0.4) rectangle (3.8,2.5);
\node at (0.8,2) {$Y$};
\node at (0.8,1) {$X$};
\node at (4,1.2) {$R$};
\draw (1,1) -- (2.5,1);
\draw (1,2) -- (2.5,2);
\draw (3.8,1.5) -- (4,1.5);
\draw (4,1.5) -- (4,1.8);
\draw (4,1.8) -- (4,2.3);
\draw (4,2.3) -- (4,5.5);

\draw (1.2,0.2) -- (4.6,0.2);
\draw (1.2,2.8) -- (4.6,2.8);
\draw (1.2,0.2) -- (1.2,1);
\draw (1.2,2) -- (1.2,2.8);

\draw (1.2,4.5) -- (1.7,4.5);
\draw (1.7,4.5) -- (2.5,4.5);
\draw (2.5,4.5) -- (6.3,4.5);
\draw (7.3,4.5) -- (12.6,4.5);
\draw (14,4.5) -- (15,4.5);
\node at (14.5,4.7) {$M$};
\node at (3.2,1.5) {$\nmext$};


\draw [dashed] (1.4,-2.5) -- (1.4,7);
\draw [dashed] (4.88,-2.5) -- (4.88,7);
\draw [dashed] (8.8,-2.5) -- (8.8,7);
\draw [dashed] (14.3,-2.5) -- (14.3,7);

\node at (6.75,1.6) {$\ket{\psi}_{W_1W_2}$};
\node at (1.2,-2.3) {$\rho$};
\node at (4.74,-2.3) {$\rho_1$};
\node at (8.6,-2.3) {$\rho_2$};
\node at (14.13,-2.3) {$\rho_3$};

\node at (3.2,4.7) {$Z$};
\node at (8,4.7) {$Z$};
\node at (4.7,3) {$Y$};
\node at (8,3) {$Y'$};
\draw (4.5,2.8) -- (6.3,2.8);
\draw (7.3,2.8) -- (10.5,2.8);

\node at (12.7,2.6) {$R'$};
\draw (12.5,2.4) -- (13.5,2.4);
\draw (13.5,2.4) -- (13.5,3.9);

\node at (4.7,0.4) {$X$};
\node at (8,0.0) {$X'$};
\draw (4.5,0.2) -- (6.3,0.2);
\draw (7.3,0.2) -- (10.5,0.2);


\draw (6.3,2) rectangle (7.3,5);
\node at (6.8,3.5) {$V$};
\draw (6.3,0) rectangle (7.3,1);
\node at (6.8,0.5) {$U$};

\node at (6.5,-0.4) {$\mathcal{A}=(U,V,\psi)$};
\draw (5.2,-0.8) rectangle (8.2,5.3);

\draw (5.8,1.5) ellipse (0.3cm and 1cm);
\node at (5.8,2) {$W_2$};
\draw (6,2.2) -- (6.3,2.2);
\node at (7.7,2) {$W_2$};
\draw (7.3,2.2) -- (7.5,2.2);
\node at (5.8,1) {$W_1$};
\draw (6,0.7) -- (6.3,0.7);
\node at (7.7,0.9) {$W_1$};
\draw (7.3,0.7) -- (7.45,0.7);



\draw (10.5,-0.5) rectangle (12.5,3.2);
\node at (11.5,1.5) {${\nmcdec}$};



\draw (9.4,-0.6) rectangle (14.2,5.9);
\node at (11,6.1) {$\dec$};
\end{tikzpicture}}
\caption{Split-state tampering experiment after applying the transpose method.
}\label{fig:splitstate3}
\end{figure}

\begin{figure}
\centering
\resizebox{12cm}{6cm}{
\begin{tikzpicture}

\draw (1,5.5) ellipse (0.3cm and 1.5cm);
\draw (14.6,5.7) rectangle (16,7);
\draw (12.9,3.9) rectangle (14.3,5.4);
\node at (1,4.5) {$M$};
\node at (15.3,6.1) {$C^T_{R}$};
\node at (13.5,4.6) {$C^\dagger_{R'}$};
\node at (1,6.5) {$\hat{M}$};
\draw (1.2,6.5) -- (14.6,6.5);
\draw (16,6.5) -- (16.5,6.5);

\draw (1.2,4.5) -- (1.7,4.5);
\draw (1.7,4.5) -- (2.5,4.5);
\draw (2.5,4.5) -- (6,4.5);
\draw (7,4.5) -- (12.9,4.5);
\draw (14.3,4.5) -- (16,4.5);
\node at (16.5,4.7) {$M$};
\node at (16.5,6.2) {$\hat{M}$};
\draw (4.5,-1) -- (10.5,-1);
\draw (4.5,-1.6) -- (10.5,-1.6);
\draw (15.5,-1.6) -- (12.4,-1.6);

\draw  (15.5,-1.8) -- (15.5,4.1);
\draw [dashed] (15.5,4.1) -- (15.5,4.7);
\draw  (15.5,4.7) -- (15.5,5.7);

\draw [dashed] (1.4,-2.5) -- (1.4,7);
\draw [dashed] (4.88,-2.5) -- (4.88,7);
\draw [dashed] (9,-2.5) -- (9,7);
\draw [dashed] (12.8,-2.5) -- (12.8,7);
\draw [dashed] (16.1,-2.5) -- (16.1,7);

\node at (6.6,1.6) {$\ket{\psi}_{W_1W_2}$};
\node at (1.2,-2.3) {$\theta$};
\node at (4.74,-2.3) {$\theta_1$};
\node at (8.8,-2.3) {$\theta_2$};
\node at (13,-2.3) {$\theta_3$};
\node at (15.9,-2.3) {$\theta_4$};

\node at (3.2,4.7) {$Z$};
\node at (8,4.7) {$Z$};
\node at (0.8,3) {$Y$};
\draw (4.5,2.8) -- (1,2.8);
\draw (4.5,2.8) -- (4.5,0.5);
\draw [dashed] (4.5,0.5) -- (4.5,-0.5);
\draw (4.5,-0.5) -- (4.5,-1);

\draw (1,0.2) -- (4.5,0.2);
\draw (3.5,0.2) -- (3.5,-1.6);
\draw (3.5,-1.6) -- (4.5,-1.6);
\node at (8,3) {$Y'$};
\draw (4.5,2.8) -- (6,2.8);
\draw (7,2.8) -- (10.5,2.8);

\node at (12.6,2.6) {$R'$};
\draw (12.4,2.4) -- (13.5,2.4);
\draw (13.5,2.4) -- (13.5,3.9);

\node at (0.8,0.4) {$X$};
\node at (8,0.0) {$X'$};
\draw (4.5,0.2) -- (6,0.2);
\draw (7,0.2) -- (10.5,0.2);

\draw (6,2) rectangle (7,5);
\node at (6.5,3.5) {$V$};
\draw (6,0) rectangle (7,1);
\node at (6.5,0.5) {$U$};

\node at (6.5,-0.4) {$\mathcal{A}=(U,V,\psi)$};
\draw (5.05,-0.8) rectangle (7.75,5.3);

\draw (5.5,1.5) ellipse (0.3cm and 1.2cm);
\node at (5.5,2) {$W_2$};
\draw (5.7,2.2) -- (6,2.2);

\node at (7.4,1.9) {$W_2$};
\draw (7,2.2) -- (8.5,2.2);
\draw (8.5,2.2) -- (8.5,6.3);
\draw (8.5,6.3) -- (13.8,6.3);
\node at (14.1,6.3) {$W_2$};

\node at (5.5,1) {$W_1$};
\draw (5.7,0.7) -- (6,0.7);
\node at (7.4,0.9) {$W_1$};
\draw (7.0,0.7) -- (7.2,0.7);


\draw (10.5,-2) rectangle (12.4,-0.8);
\draw (10.5,-0.4) rectangle (12.4,3.1);
\node at (11.5,1.5) {${\nmcdec}$};
\node at (11.5,-1.4) {${\nmcdec}$};

\node at (15.3,-1.6) {$R$};

\draw (9.4,-0.6) rectangle (14.5,5.9);
\node at (11,6.1) {$\dec$};
\end{tikzpicture}}
\caption{
Split-state tampering experiment after applying the transpose method and delaying both the generation of register $R$ and the application of the corresponding Clifford operator.
}\label{fig:splitstate4}
\end{figure}
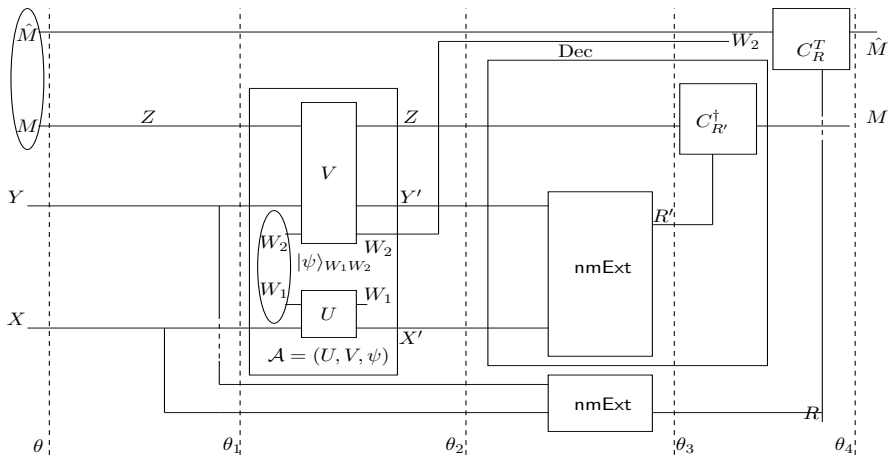

\subsubsection{Proof of \cref{lem:thetatamp}}\label{sec:thetatamp}

We begin by writing
\begin{equation*}
     (\theta_3)_{RR'W_2Z\hat{M}}= p_{\sm} (\thetasame_3)_{ RR'W_2Z\hat{M}} + (1-p_{\sm}) (\thetatamp_3)_{ RR'W_2Z\hat{M}}.
    \end{equation*}
    Observe that $(\theta_1)_{\hat{M}M}$ is a pure state (thus independent of other registers in $\theta_1$). 
    We invoke Item 2 of \cref{lem:qnmcodesfromnmext} with the assignment of states/registers
    (the states/registers on the left below are from the statement of \cref{lem:qnmcodesfromnmext}, while the states/registers on the right are those from this proof)
    \begin{equation*}
        \left(\sigma, \hat{\rho}, \rho, \rho^\sm, \rhotamp , W_1,W_2,X,Y\right)
    \leftarrow 
    \left( \theta_1, \theta_2, \theta_3, \thetasame_3,  \thetatamp_3, W_1, W_2M \hat{M},X,Y \right)
    \end{equation*}
    to conclude that
    \begin{multline}\label{eq:boundtheta3nmext}
        p_{\sm} \Vert (\thetasame_3)_{ RW_2Z\hat{M}}-  U_r \otimes (\thetasame_3)_{W_2Z\hat{M}} \Vert_1 + \\ (1-p_{\sm})\Vert (\thetatamp_3)_{RR'W_2Z\hat{M}}-  U_r \otimes (\thetatamp_3)_{R'W_2Z\hat{M}} ) \Vert_1  \leq \eps.
\end{multline}
In particular, it holds that
\begin{equation}\label{eq:specificboundtheta3nmext}
   (1-p_{\sm})  \Vert (\thetatamp_3)_{ RR'W_2Z\hat{M}}-  U_m \otimes (\thetatamp_3)_{R'W_2Z\hat{M}}  \Vert_1 \leq  \eps.
\end{equation}

As a result, we have that
\begin{align*}
    &(1-p_{\sm})  \Vert  (\thetatamp_4)_{\hat{M}M} -U_{\hat{M}} \otimes    (\thetatamp_4)_{M} \Vert_1 \\
     & = (1-p_{\sm})  \Vert  ( \id \otimes D)( \tilde{D}_{} \otimes \id)(\thetatamp_3)_{ RR'Z\hat{M}} ( \tilde{D}^\dagger_{} \otimes \id)( \id \otimes D) \\ & \quad - ( \id \otimes D)\left(U_{r} \otimes ( \tilde{D}_{} \otimes \id)   (\thetatamp_3)_{R'Z\hat{M}} ( \tilde{D}^\dagger_{} \otimes \id) \right)( \id \otimes D)\Vert_1 \\
    &\leq(1-p_{\sm})  \Vert   ( \tilde{D}_{} \otimes \id)(\thetatamp_3)_{ {R}R'Z\hat{M}} ( \tilde{D}^\dagger_{} \otimes \id) -U_{r} \otimes    ( \tilde{D}_{} \otimes \id)   (\thetatamp_3)_{R'Z\hat{M}} ( \tilde{D}^\dagger_{} \otimes \id) \Vert_1 & \mbox{(\cref{fact:notequal} and \cref{fact:data})}  \\
    & \leq (1-p_{\sm})   \Vert  (\thetatamp_3)_{ RR'Z\hat{M}} -  U_{r} \otimes  (\thetatamp_3)_{R'Z\hat{M}}\Vert_1 & \mbox{(\cref{fact:data})}\\
    & \leq (1-p_{\sm})  \Vert (\thetatamp_3)_{ RR'W_2Z\hat{M}}-  U_r \otimes (\thetatamp_3)_{R'W_2Z\hat{M}}  \Vert_1 & \mbox{(\cref{fact:data})}\\
    & \leq \eps. & \mbox{(\cref{eq:specificboundtheta3nmext})}
\end{align*}
The first inequality invokes \cref{fact:notequal} with the assignment of registers $(A,B,\rho_{ } ) \longleftarrow ( \hat{M}, M, \thetatamp_{4})$, where the registers on the left hand side correspond to those in the statement \cref{fact:notequal} and the registers on the right hand side are those used in this proof.

\subsubsection{Proof of \cref{lem:thetasame}}\label{sec:thetasame}

Before we proceed to the proof of \cref{lem:thetasame}, we prove two auxiliary lemmas.
\begin{lemma}\label{lem:equalreq}
Let $\rho_{A\hat{A}}$ be the canonical purification of $\rho_A = U_A$, $\cSC(\cH_A)$ be the subgroup of Clifford group as defined in \cref{lem:subclifford}, and $P, Q \in \cP(\cH_A)$ be any two Pauli operators. 
If $P  \ne Q$, then
\begin{equation*}
    \frac{1}{\vert \cSC(\cH_A)\vert } \sum_{C \in \cSC(\cH_A)} (C^T  \otimes C^\dagger) (\id \otimes P) \rho_{A \hat{A}} (\id \otimes Q^\dagger) ( (C^T)^\dagger  \otimes C) = 0.
\end{equation*}
Else, if $P=Q = \id_A$, then
\begin{equation*}
    \frac{1}{\vert \cSC(\cH_A)\vert } \sum_{C \in \cSC(\cH_A)} (C^T  \otimes C^\dagger) (\id \otimes P) \rho_{A \hat{A}} (\id \otimes P^\dagger) ( (C^T)^\dagger  \otimes C) = \rho_{A\hat{A}}.
\end{equation*}
Else, if $P =Q \ne \id_A$, then
\begin{equation*}
    \frac{1}{\vert \cSC(\cH_A)\vert } \sum_{C \in \cSC(\cH_A)} (C^T  \otimes C^\dagger) (\id \otimes P) \rho_{A \hat{A}} (\id \otimes P^\dagger) ( (C^T)^\dagger  \otimes C) \approx_{\frac{2}{ \vert \cP(\cH_A) \vert}} \rho_A \otimes \rho_{\hat{A}}.
\end{equation*}
\end{lemma}
\begin{proof}
We have that
\begin{align*}
   & \frac{1}{\vert \cSC(\cH_A)\vert } \sum_{C \in \cSC(\cH_A)} (C^T  \otimes C^\dagger) (\id \otimes P) \rho_{A \hat{A}} (\id \otimes Q^\dagger) ( (C^T)^\dagger  \otimes C)  \\ 
      & =\frac{1}{\vert \cSC(\cH_A)\vert } \sum_{C \in \cSC(\cH_A)} (\id  \otimes C^\dagger P) (C^T  \otimes \id)  \rho_{A \hat{A}}  ( (C^T)^\dagger  \otimes \id)  ( \id  \otimes Q^\dagger C) \\ 
       & =\frac{1}{\vert \cSC(\cH_A)\vert } \sum_{C \in \cSC(\cH_A)} (\id  \otimes C^\dagger P) (\id  \otimes C)  \rho_{A \hat{A}}  ( \id \otimes C^\dagger  )  ( \id  \otimes Q^\dagger C) & \mbox{(\cref{fact:transposetrick})} \\
    & =\frac{1}{\vert \cSC(\cH_A)\vert } \sum_{C \in \cSC(\cH_A)} (\id  \otimes C^\dagger P C)   \rho_{A \hat{A}}  ( \id \otimes C^\dagger Q^\dagger C  ) . 
\end{align*}
When $P \neq Q$ we can invoke \cref{lem:cliffordtwirl1} to conclude that
\begin{equation*}
    \frac{1}{\vert \cSC(\cH_A)\vert } \sum_{C \in \cSC(\cH_A)} (\id  \otimes C^\dagger P C)   \rho_{A \hat{A}}  ( \id \otimes C^\dagger Q^\dagger C  ) = 0 .
\end{equation*}
When $P=Q = \id_A$, it holds that
\begin{align*}
    &\frac{1}{\vert \cSC(\cH_A)\vert } \sum_{C \in \cSC(\cH_A)} (\id  \otimes C^\dagger P C)   \rho_{A \hat{A}}  ( \id \otimes C^\dagger P^\dagger C  ) \\
    & =\frac{1}{\vert \cSC(\cH_A)\vert } \sum_{C \in \cSC(\cH_A)} (\id  \otimes C^\dagger C)   \rho_{A \hat{A}}  ( \id \otimes C^\dagger C  )  \\
    &=  \rho_{A \hat{A}}. & \mbox{($C^\dagger C = \id$)}
\end{align*} 
Finally, when $P=Q \ne \id_A$ we have that
\begin{align*}
    & \frac{1}{\vert \cSC(\cH_A)\vert } \sum_{C \in \cSC(\cH_A)} (\id  \otimes C^\dagger P C)   \rho_{A \hat{A}}  ( \id \otimes C^\dagger P^\dagger C  )  \\
     & =\frac{1}{\vert \cP(\cH_A)\vert -1} \sum_{Q \in \cP(\cH_A)\setminus\id } (\id  \otimes Q)   \rho_{A \hat{A}}  ( \id \otimes Q^\dagger   ) & \mbox{(\cref{lem:subclifford})} \\
    &= \frac{\vert \cP(\cH_A)\vert ( \rho_{A} \otimes \rho_{\hat{A}} ) -\rho_{A\hat{A}} }{\vert \cP(\cH_A)\vert-1}. & \mbox{(\cref{fact:bellbasis})}
\end{align*} 
The desired result follows since
\begin{equation*}
    \left\|\frac{\vert \cP(\cH_A)\vert ( \rho_{A} \otimes \rho_{\hat{A}} ) -\rho_{A\hat{A}} }{\vert \cP(\cH_A)\vert-1}   - ( \rho_{A} \otimes \rho_{\hat{A}} )   \right\|_1 \leq \frac{2}{ \vert \cP(\cH_A) \vert}.
\end{equation*}
\end{proof}

\begin{lemma}\label{lem:equal101}
 Let $\ket{\psi}_{\hat{A}A}$ be the canonical purification of $\psi_A = U_A$, $\rho_{\hat{A}A}$ be any state, and $\cSC(\cH_A)$ be the subgroup of Clifford group as defined in \cref{lem:subclifford}. 
 Define $\Pi =\ketbra{\psi}$. 
 Then, we have that
 \begin{multline}\label{eq:equal101}
    \frac{1}{\vert \cSC(\cH_A)\vert } \sum_{C \in \cSC(\cH_A)} (C^T  \otimes C^\dagger) \rho_{\hat{A}A}( (C^T)^\dagger  \otimes C) \\ =  \tr(\Pi \rho) \psi + (1- \tr(\Pi \rho))  \frac{\vert \cP(\cH_A)\vert (U_{\hat{A} } \otimes U_{{A}} ) -\psi_{A\hat{A}} }{\vert \cP(\cH_A)\vert-1}\\ \approx_{\frac{2}{4^{\vert A \vert}}}  \tr(\Pi \rho) \psi + (1- \tr(\Pi \rho)) (U_{\hat{A}} \otimes U_A) .
 \end{multline}
\end{lemma}
\begin{proof}
  Let $\ket{\phi}_{\hat{A}A}$ be an eigenvector of $\rho_{\hat{A}A}$. Consider the decomposition 
  $$\ket{\phi}_{\hat{A}A} = \sum_{P \in \cP(\cH_{A}) } \alpha_P (\id \otimes P) \ket{\psi}_{\hat{A}A},$$ 
  where $\sum_{P \in \cP(\cH_{A})  } \vert \alpha_P  \vert^2 =1$.
  Define
  \begin{equation*}
      \tau(P,Q)\defeq (\id \otimes P) \ketbra{\psi} ( \id \otimes Q^\dagger ).
  \end{equation*}
  Then, we have that
\begin{align*}
  &\frac{1}{\vert \cSC(\cH_A)\vert } \sum_{C \in \cSC(\cH_A)} (C^T  \otimes C^\dagger) \ketbra{\phi} ( (C^T)^\dagger  \otimes C)  \\ 
  &=\frac{1}{\vert \cSC(\cH_A)\vert } \sum_{C \in \cSC(\cH_A)} (C^T  \otimes C^\dagger) \left(\sum_{P,Q \in \cP(\cH_{A})} \alpha_{P} \alpha^*_{Q} \tau(P,Q) \right)( (C^T)^\dagger  \otimes C)  \\ 
   &=  \sum_{P,Q \in \cP(\cH_{A})} \alpha_{P} \alpha^*_{Q}   \left( \frac{1}{\vert \cSC(\cH_A)\vert } \sum_{C \in \cSC(\cH_A)} (C^T  \otimes C^\dagger)  \tau(P,Q) ( (C^T)^\dagger  \otimes C) \right)  \\
   & =  \sum_{P,Q \in \cP(\cH_{A}) \land (P \ne Q)}\alpha_{P} \alpha^*_{Q}   \left( \frac{1}{\vert \cSC(\cH_A)\vert } \sum_{C \in \cSC(\cH_A)} (C^T  \otimes C^\dagger)  \tau(P,Q) ( (C^T)^\dagger  \otimes C) \right) \\
   & +  \sum_{P,Q \in \cP(\cH_{A})  \land (P=Q=\id_A)} \alpha_{P} \alpha^*_{Q}   \left( \frac{1}{\vert \cSC(\cH_A)\vert } \sum_{C \in \cSC(\cH_A)} (C^T  \otimes C^\dagger)  \tau(P,Q) ( (C^T)^\dagger  \otimes C) \right) \\
   & +  \sum_{P,Q \in \cP(\cH_{A})  \land (P=Q \ne \id_A)} \alpha_{P} \alpha^*_{Q}   \left( \frac{1}{\vert \cSC(\cH_A)\vert } \sum_{C \in \cSC(\cH_A)} (C^T  \otimes C^\dagger)  \tau(P,Q) ( (C^T)^\dagger  \otimes C) \right) \\
    &= \vert \alpha_{\id_A} \vert^2  \ketbra{\psi}\\
    &+ \sum_{P \in \cP(\cH_{A}) \setminus \id_A}  \vert \alpha_{P} \vert^2   \left( \frac{1}{\vert \cSC(\cH_A)\vert } \sum_{C \in \cSC(\cH_A)} (C^T  \otimes C^\dagger) \tau(P,P) ( (C^T)^\dagger  \otimes C) \right)  \\
      &= \vert \alpha_{\id_A} \vert^2  \ketbra{\psi}+ (1- \vert \alpha_{\id_A} \vert^2 )  \left(\frac{\vert \cP(\cH_A)\vert (U_{A} \otimes U_{\hat{A}} ) -\psi_{A\hat{A}} }{\vert \cP(\cH_A)\vert-1} \right).
\end{align*} 
The last equality follows from \cref{lem:equalreq}. Now, the first equality in \cref{eq:equal101} from the lemma statement follows by observing that $\rho_{\hat{A}A}$ is a convex combination of its eigenvectors, and the approximation in \cref{eq:equal101} follows from \cref{fact:traceconvex} by observing that
\begin{equation*}
    \left\|\frac{\vert \cP(\cH_A)\vert ( U_{A} \otimes U_{\hat{A}} ) -\psi_{A\hat{A}} }{\vert \cP(\cH_A)\vert-1}   - ( U_{A} \otimes U_{\hat{A}} )   \right\|_1 \leq \frac{2}{ \vert \cP(\cH_A) \vert}.
\end{equation*}
\end{proof}

We are now ready to prove \cref{lem:thetasame}.
\begin{proof}[Proof of \cref{lem:thetasame}]
Recall from \cref{eq:boundtheta3nmext} that, in particular,
\begin{equation}\label{eq:boundtheta3same}
  p_{\sm} \Vert (\thetasame_3)_{ RW_2Z\hat{M}}-  U_r \otimes (\thetasame_3)_{W_2Z\hat{M}}  \Vert_1 \leq  \eps.
\end{equation}
Recall also that $\tilde{D}_{R'Z}$ denotes the controlled Clifford unitary $C^\dagger_{R'}$ acting on register $Z$ and that ${D}_{R\hat{M}}$ denotes the controlled Clifford operator $C^T_{R}$ acting on register $\hat{M}$. 

By the conditioning in $\thetasame_3$, we have that
\begin{equation*}
    \Pr[R=R']_{ (\thetasame_3)} =1.
\end{equation*}
Let $\Pi$ denote the maximally entangled state in registers $Z\hat{M}$ and write $p_\epr = \tr(\Pi ( \thetasame_3)_{Z\hat{M}} ))$. Consider the state $\tau_{RR'}$ such that $\Pr[R=R']_\tau=1$ and $\tau_R=U_r$, and for ease of readability define
\begin{equation*}
\beta \defeq \tau_{RR'} \otimes ( \thetasame_3)_{Z\hat{M}}  \quad ; \quad 
    \gamma \defeq (\tilde{D} \otimes D) \beta  (\tilde{D}^\dagger \otimes D^\dagger).
\end{equation*}
Intuitively, the states $\beta$ and $\gamma$ are ``approximate'' versions of $(\thetasame_3)_{RR'Z\hat{M}}$ and $(\thetasame_4)_{\hat{M}M}$ which are easier to reason about.
Our argument proceeds in two steps. 
First, we show that $\gamma_{\hat{M}M}$ is close in trace distance to
\begin{equation*}
    \left( p_\epr \Pi + (1-p_\epr) (U_{\hat{M}} \otimes U_{M}) \right).
\end{equation*}
Then, we argue that $(\thetasame_4)_{\hat{M}M}$ and $ \gamma_{\hat{M}M}$ are close in trace distance, and an application of the triangle inequality to these two statements finishes the proof (after arguing that $p_\epr$ depends only on $\cA$). 
More formally, we have that
\begin{align}
    & p_{\sm} \Vert  (\thetasame_4)_{\hat{M}M}-   \left( p_\epr \Pi + (1-p_\epr) (U_{\hat{M}} \otimes U_{M}) \right)\Vert_1   \nonumber\\
      & \leq p_{\sm} \Vert   (\thetasame_4)_{\hat{M}M}-  \gamma_{\hat{M}M}  \Vert_1 \nonumber\\
      & \quad \quad \quad  +p_\sm \Vert \gamma_{\hat{M}M} -  \left( p_\epr \Pi + (1-p_\epr) (U_{\hat{M}} \otimes U_{M}) \right) \Vert_1 \nonumber\\
        & \leq p_{\sm}  \Vert   (\thetasame_4)_{\hat{M}M}-  \gamma_{\hat{M}M}  \Vert_1+ 2 \cdot 4^{-\vert M \vert} \label{eq:lem1equal10}\\
         & \leq p_{\sm}  \Vert (\tilde{D} \otimes D) ( \thetasame_3 )_{RR'\hat{M}M} (\tilde{D}^\dagger \otimes D^\dagger)-   \gamma_{}  \Vert_1+ 2 \cdot 4^{-\vert M \vert} \label{eq:1data1} \\
     & \leq p_{\sm}\Vert(\thetasame_3)_{ RR'Z\hat{M}}-  \tau_{RR'} \otimes (\thetasame_3)_{Z\hat{M}}  \Vert_1 + 2 \cdot 4^{-\vert M \vert}  \label{eq:1data2}  \\
      & \leq p_{\sm}\Vert( \thetasame_3)_{ RZ\hat{M}}-  U_r \otimes (\thetasame_3)_{Z\hat{M}}  \Vert_1 + 2 \cdot 4^{-\vert M \vert} \label{eq:1data3} \\
        & \leq p_{\sm} \Vert (\thetasame_3)_{ RW_2Z\hat{M}}-  U_r \otimes (\thetasame_3)_{W_2Z\hat{M}}  \Vert_1 + 2 \cdot 4^{-\vert M \vert} \label{eq:1traceavg1} \\
        & \leq \eps + 2 \cdot 4^{-\vert M \vert}. \label{eq:122thm1}
    \end{align}
    \cref{eq:lem1equal10} follows from \cref{lem:equal101} with the assignment of registers (the registers on the left below are from \cref{lem:equal101} and the registers on the right are the registers in this proof)
    \begin{equation*}
        (\hat{A},A , \rho) \longleftarrow ( \hat{M}, Z, \thetasame_3).
    \end{equation*}
    \cref{eq:1data1,eq:1data2,eq:1data3,eq:1traceavg1} follow from \cref{fact:data}.
    We provide some more detail on \cref{eq:1data3}.
    The difference between the expressions in \cref{eq:1data3,eq:1data2} is that the (classical) register $R'$ is removed. 
    The stated inequality (which is even an equality) follows from data-processing because in the states $\thetasame_3$ and $\tau$ we have that $R'$ is a copy of $R$, and so can be generated based on $R$. 
    This is because in these states we have $R = \nmext(X,Y)$, $R' = \nmext(X',Y')$, and $(X,Y)=(X',Y')$.
    The final \cref{eq:122thm1} follows from \cref{eq:boundtheta3same}. 
    
    To conclude the proof, it remains to argue that $p_\epr$ depends only on the split-state adversary $\cA$.
    This holds because $p_\epr$ is a function of the state $(\thetasame_3)_{Z\hat{M}}$, which can be prepared by running an independent tampering experiment with another maximally entangled state independent of the original input message.
\end{proof}

\subsection{From average-case to worst-case non-malleability}\label{sec:avgtoworst}

In this section, we show that every average-case non-malleable code is also worst-case non-malleable with larger error, provided that the message length is not too large.
More precisely, we have the following.
\begin{lemma}\label{thm:avgtoworst}
    If $(\enc,\dec)$ is an average-case $\eps$-non-malleable code for quantum messages of length $b$, then it is also a (worst-case) $\eps'$-non-malleable code for quantum messages of length $b$, where $\eps' = 2^b \cdot \eps$.
\end{lemma}

Combining \cref{thm:avgcaseNMC} with \cref{thm:avgtoworst} immediately implies the following result.
\begin{theorem}\label{thm:worstcaseNMC}
    There exists a constant $c\in (0,1)$ such that the following holds: The coding scheme $(\enc,\dec)$ described in \cref{sec:codedesc} with codewords of length $n$ and message size $b\leq n^c$ is (worst-case) $\eps$-non-malleable with $\eps=2^{-n^{\Omega(1)}}$.
    Moreover, both $\enc$ and $\dec$ can be computed in time $\poly(n)$.
\end{theorem}

Note that the only barrier towards obtaining a constant rate code in \cref{thm:worstcaseNMC} is the fact that the average-case non-malleable code from \cref{thm:avgcaseNMC} has error $\eps=2^{-n^{\Omega(1)}}$, where $n$ is the codeword size.
Therefore, \cref{thm:avgtoworst} only allows for quantum messages of length up to $n^{\Omega(1)}$, as otherwise the resulting error ($2^b\cdot \eps$) becomes trivial.

We now proceed to prove \cref{thm:avgtoworst}.
\begin{proof}[Proof of \cref{thm:avgtoworst}]
    Fix a split-state adversary $\cA=(U,V,\ket{\psi}_{W_1W_2})$.
    Let $\rho_{M\hat{M}}$ be the canonical purification of $\rho_M=U_M$, and denote by $\eta$ the outcome of the tampering experiment of $\cA$ on the encoding of $\rho$, i.e.,
    \begin{equation*}
        \eta=  \dec(  (U \otimes V) (\enc (\rho_{\hat{M}M}) \otimes \ketbra{\psi} )(U^\dagger \otimes V^\dagger)).
    \end{equation*}
    Since $(\enc,\dec)$ is an average-case $\eps$-non-malleable code and $\rho_{M\hat{M}}$ is maximally entangled (and so $\rho_M$ is maximally mixed), we know that
    \begin{equation}\label{eq:guaranteeavgcase}
        \eta_{\hat{M}M} \approx_{\eps} p_{\mathcal{A}}\rho_{\hat{M}M}+ (1-p_{\mathcal{A}}) \rho_{\hat{M}} \otimes \gamma^{\mathcal{A}}_{M} ,
    \end{equation}
    where $(p_{\mathcal{A}}, \gamma^{\mathcal{A}}_{M} )$ depend only on the split-state adversary $\mathcal{A}$. 
    Note also that $\eta_{\hat{M}} =\rho_{\hat{M}}= U_{\hat{M}}.$ 
    
    Towards showing that $(\enc,\dec)$ is also worst-case non-malleable, consider an arbitrary quantum message $\sigma_M$ of size $b$ with canonical purification $\sigma_{M\hat{M}}$.
    If
    \begin{equation*}
        \zeta = \dec(  (U \otimes V) (\enc (\sigma_{\hat{M}M}) \otimes \ketbra{\psi} )(U^\dagger \otimes V^\dagger) )
    \end{equation*}
    is the outcome of the tampering experiment of $\cA$ on message $\sigma$, to derive the theorem statement it suffices to show that 
    \begin{equation}\label{eq:worstcasefinalgoal}
        \zeta_{\hat{M}M} \approx_{2^b \eps} p_{\mathcal{A}}\sigma_{\hat{M}M}+ (1-p_{\mathcal{A}}) 
        (\sigma_{\hat{M}} \otimes  \gamma^{\mathcal{A}}_{M}),
    \end{equation}
    where $(p_{\mathcal{A}},\gamma^{\mathcal{A}}_{M})$ depend only on $\cA$ and come from \cref{eq:guaranteeavgcase}.
    
    Recalling the definition of max-divergence (\cref{def:maxdiv}), we have that
    \begin{equation*}
        \dmax{\sigma_M}{\rho_M} \leq b,
    \end{equation*}
    since $\sigma_M \leq \id_M$ and $\rho_M=\id_M/2^b$. 
    Therefore, by \cref{rejectionsampling} there exists a measurement $\{\Pi , \bar{\Pi} \}$ on register $\hat{M}$ in state $\rho$ with success probability $\tr(\Pi \rho) \geq 2^{-b}$ such that conditioned on success we have that
    \begin{equation*}
        \sigma_{\hat{M}M} =\frac{\Pi \rho \Pi}{ \tr(\Pi \rho)}.
    \end{equation*} 
    First, note that $\Pi$ commutes with the operations of ($\enc, \dec$) and $\mathcal{A}$, since it is applied on register $\hat{M}$. 
    Consequently, again conditioned on the measurement being successful, we have that
    \begin{equation*}
        \zeta_{\hat{M}M}  =\frac{\Pi \eta_{\hat{M}M} \Pi}{ \tr(\Pi \rho)},
    \end{equation*}
    and so 
    \begin{align*}
        \zeta_{\hat{M}M}  & = \frac{\Pi \eta_{\hat{M}M} \Pi}{ \tr(\Pi \rho)} \\
        & \approx_{2^{b}\eps} p_{\mathcal{A}}\frac{\Pi \rho_{\hat{M}M} \Pi}{ \tr(\Pi \rho)}+ (1-p_{\mathcal{A}}) \frac{\Pi \rho_{\hat{M}} \Pi}{ \tr(\Pi \rho)} \otimes \gamma^{\mathcal{A}}_{M} \\
        & = p_{\mathcal{A}}\sigma_{\hat{M}M}+ (1-p_{\mathcal{A}}) 
        (\sigma_{\hat{M}} \otimes  \gamma^{\mathcal{A}}_{M}).
    \end{align*}
    The approximation above follows from \cref{traceavg1} and by noting that
    \begin{equation*}
        \eta_{\hat{M}M} \approx_{\eps} p_{\mathcal{A}}\rho_{\hat{M}M}+ (1-p_{\mathcal{A}}) \rho_{\hat{M}} \otimes \gamma^{\mathcal{A}}_{M}.
    \end{equation*}
    This completes the proof. 
\end{proof}

\subsection{Privacy of our split-state NMC for quantum messages}\label{sec:privNMC}

We show that our split-state (worst-case) non-malleable code for quantum messages is also a $2$-out-of-$2$ secret sharing scheme for quantum messages.
\begin{theorem}\label{thm:22nmss}
    Let $(\enc,\dec)$ be the $(\eps'=2^{-n^{\Omega(1)}})$-non-malleable code described in \cref{sec:codedesc}.
    Then, $(\enc,\dec)$ is also an $(\epspriv=2^{-n^{\Omega(1)}}, \epsnm=2^{-n^{\Omega(1)}})$-$2$-out-of-$2$ non-malleable secret sharing scheme for quantum messages. 
\end{theorem}
\begin{proof}
Since $(\enc,\dec)$ is already known to be $(\eps'=2^{-n^{\Omega(1)}})$-non-malleable by \cref{thm:mainqnmc}, it remains to show that it also satisfies statistical privacy with error $\epspriv=2^{-n^{\Omega(1)}}$.

Let $\sigma_M$ be an arbitrary message.
Recall that $\enc(\sigma_M)$ is composed of two parts $X$ and $(Y,C_R \sigma_M C^\dagger_R)$, where $X\leftarrow \bits^\ell$, $Y\leftarrow\bits^{\delta\ell}$ are independent and $R=\nmext(X,Y)$ with $\nmext$ being the explicit function described in \cref{lem:qnmcodesfromnmext}.
Additionally, we have $C_R=\samp(R)$ with $\samp$ being the sampling procedure for the special Clifford subgroup from \cref{lem:subclifford}.

First, note that $X$, the first part of $\enc(\sigma_M)$, is clearly sampled independently of the message $\sigma_M$, and so statistical privacy trivially holds when an adversary only has access to this first part.
Therefore, we only need to focus on the case where an adversary holds the second part $(Y,C_R\sigma_M C^\dagger_R)$ of $\enc(\sigma_M)$.
Item 1 of \cref{lem:qnmcodesfromnmext} guarantees that
\begin{equation*}
    RY = \nmext(X,Y)Y \approx_{2^{-n^{\Omega(1)}}} U_r\otimes U_{\delta n}.
\end{equation*}
This means that we can from here onwards assume in our argument that $R$ and $Y$ are independent and both uniformly distributed (i.e., $RY =U_r\otimes U_{\delta n}$) at the expense of an extra additive factor $2^{-n^{\Omega(1)}}$ in the final statistical privacy error $\epspriv$.

Observe that, for a fixed message $\sigma_M$, the state $(Y,C_R\sigma_M C^\dagger_R)$ corresponding to the second part of $\enc(\sigma_M)$ is a (randomized) function of $RY$ only.
Since we are assuming that $Y$ is independent of $R$, it suffices to argue that
\begin{equation}\label{eq:approxencryption}
    C_R \sigma_{\hat{M}M}C^\dagger_R \approx_{2^{-n^{\Omega(1)}}}  \sigma_{\hat{M}} \otimes U_M
\end{equation}
for any message $\sigma_{\hat{M}M}$.
If $C_R$ was uniformly distributed over the Clifford subgroup $\cSC(\cH_M)$, then \cref{fact:notequal} would guarantee that
\begin{equation}\label{eq:dreamequiv}
      C_R \sigma_{\hat{M}M}C^\dagger_R = \sigma_{\hat{M}} \otimes U_M.
\end{equation}
Although $C_R$ is not uniformly distributed over the Clifford subgroup $\cSC(\cH_M)$, we know from \cref{lem:subclifford} that
\begin{equation}\label{eq:approxsampleguarantee}
    \samp(R)\approx_{2^{-n^{\Omega(1)}}} U_{\cSC(\cH_M)},
\end{equation}
where $\samp(R)$ is the classical description of $C_R$ and $U_{\cSC(\cH_M)}$ is uniformly distributed over classical descriptions of Clifford operators in $\cSC(\cH_M)$.
Combining \cref{eq:dreamequiv,eq:approxsampleguarantee} yields \cref{eq:approxencryption}.
Finally, summing the various trace distances along the way yields the final statistical privacy error $\epspriv =2^{-n^{\Omega(1)}}$ via the triangle inequality.
\end{proof}

\section{Threshold non-malleable secret sharing schemes for quantum messages}\label{sec:nmssquantummsg}

In this section, we present and analyze our non-malleable threshold secret sharing schemes for quantum messages. Before doing so, we must first introduce auxiliary objects, which we call \emph{quantum-secure augmented leakage-resilient secret sharing schemes for classical messages.}

Our final construction of a non-malleable secret sharing scheme will combine our split-state non-malleable code for quantum messages from \cref{thm:mainqnmc1} with a standard secret sharing scheme for quantum messages and a quantum-secure leakage-resilient secret sharing scheme for classical messages. This overall approach can be seen as the quantum analogue of the argument by Goyal and Kumar~\cite{GK16}.

\subsection{Quantum-secure augmented leakage-resilient secret sharing schemes for classical messages}\label{sec:2lrss}

Roughly speaking, a (locally) leakage-resilient secret sharing scheme allows a dealer to share a secret among $p$ parties in such a way that authorized subsets of parties can reconstruct the secret, but unauthorized subsets of parties gain almost no information about the secret even when they learn bounded quantum leakage from other shares.

We proceed to define what we call, quantum-secure augmented leakage-resilient secret sharing schemes for classical messages. Let $\sigma_M$ be a classical message in a register $M$, and let $\sigma_{M\hat{M}}$ be such that $\hat{M}$ is a copy of $M$. We wish to share $\sigma_M$ among $p$ parties in a way that any subset of at least $t$ parties can reconstruct the secret, but subsets of at most $t-1$ parties gain almost no information about the secret, even if they learn some \emph{quantum} leakage from every other share, and even when these $t-1$ parties share arbitrary entangled states with the leakage adversaries for the remaining shares.

Towards this end, we consider a coding scheme given by a classical encoding map
$\lrenc : \cL( \cH_M) \to \cL(\cH_{S_1}\otimes \cH_{S_2} \otimes \cdots \otimes \cH_{S_p})$ and a classical decoding map $\lrdec  : \cL(\bigotimes_{i \in T} \cH_{S_i}) \to  \cL(\cH_{M})$, where $\cL(\cH)$ is the space of all linear operators in the Hilbert space $\cH$. The reconstruction procedure $\lrdec$ acts on any authorized subset of shares $T$ to reconstruct the original message $M$.

We now describe the adversarial leakage model we work under. 
Let $\rho_{ S_1 S_2 \dots S_p}=\lrenc(\sigma_M)$ be the secret sharing of message $\sigma_M$ and fix a subset $T$ of parties.
An \emph{$\ell$-bounded local leakage adversary} $\cA_T$ is specified by $p-|T|$ leakage maps
$\Phi_j:\cL(\cH_{S_j}\otimes \cH_{W_j})\to\cL(\cH_{Z_j})$ for $j \in [p]\setminus T$ along with a quantum state $\ket\psi_{W_1 W_2 \dots W_p}$ which captures the shared entanglement between non-communicating local leakage adversaries.
The outcome of the leakage operation on the $j$-th share $S_j$, is stored in register $Z_j$ with $Z_j$ being the leakage register, is obtained by applying $\Phi_j$ to the contents of registers $S_j$ and $W_j$. 
We assume that the leakage $Z_j$ is bounded, in the sense that $\vert Z_j \vert \leq \ell$ for every $j \in [p]\setminus T$, where $\ell$ is some leakage bound parameterizing the leakage adversary.
We denote the outcome of the leakage experiment against $\cA_T$ by $\tau^{\cA_T}$, which may be written as
\begin{equation*}
    \tau^{\cA_T} = \left(\left(\bigotimes_{i\in T} \id\right)\otimes\left(\bigotimes_{j\not\in T}\Phi_j\right)\right)\left( \lrenc(\sigma_{M \hat{M}}) \otimes \ketbra{\psi}_{W_1 W_2 \dots W_p} \right).
\end{equation*}

We are now ready to define quantum-secure augmented leakage-resilient secret sharing schemes for classical messages. It is easy to extend this definition to more general access structures.
For simplicity and because it suffices for our needs, we focus on threshold access structures only.
\begin{definition}[Quantum-secure threshold augmented leakage-resilient secret sharing]\label{def:2naqlrss}
    We say that the coding scheme $(\lrenc , \lrdec)$ is an \emph{$(\ell,\eps)$-quantum-secure $t$-out-of-$p$ augmented leakage-resilient secret sharing scheme} if the following properties hold:
    \begin{itemize}
        \item \textbf{Correctness:} For any set $T\subseteq[p]$ of size $|T|\geq t$ it holds that
        \begin{equation*}
        \lrdec(\lrenc(\sigma_{})_{\hat{M}T})=\sigma_{\hat{M}M}.
        \end{equation*}

        \item \textbf{Leakage-resilience:} For any set $T\subseteq[p]$ of size $|T|= t-1$, every $\ell$-bounded local leakage adversary
    \begin{equation*}
        \mathcal{A}_T= ((\Phi_j)_{j\not\in T}, \ket{\psi}_{W_1W_2 \dots W_p}),
    \end{equation*}
    and every classical message $\sigma_M$ (with the register $\hat{M}$ being a copy of $M$) it holds that 
    \begin{equation*}
        \tau^{\cA_T}_{\hat{M}S_TW_T (Z_j)_{j\not\in T}} \approx_{\eps} \sigma_{\hat{M}} \otimes \gamma^{\cA_T}_{S_T W_T (Z_j)_{j\not\in T} },
    \end{equation*}
    where $\gamma^{\cA_T}_{S_T W_T (Z_j)_{j\not\in T} }$ is a fixed state depending only\footnote{By this, we mean that the state $\gamma^{\cA_T}_{S_T W_T (L_j)_{j\not\in T}}$ can be prepared without the knowledge of the input message $\sigma_{M\hat{M}}$.} on the adversary
    \begin{equation*}
        \mathcal{A}_T= ((\Phi_j)_{j\not\in T}, \ket{\psi}_{W_1W_2 \dots W_p}).
    \end{equation*}
    \end{itemize}

    Moreover, we say that the coding scheme $(\lrenc , \lrdec)$ is an \emph{$(\ell,\eps)$-quantum-secure average-case $t$-out-of-$p$ augmented leakage-resilient secret sharing scheme} if the above holds for a uniformly random message $\sigma_M=U_M$.
\end{definition}

As was the case for our non-malleability definitions, our notion of average-case leakage-resilience is akin to requiring that the \emph{average} leakage-resilience error over the uniform choice of input message be small.

We use the adjective ``augmented'' because we consider a leakage experiment where the distinguisher learns not only the shares $S_T$ corresponding to the parties in $T$ and the leakage $(Z_j)_{j\not\in T}$, but (crucially) also the entangled states $W_T$ corresponding to parties in $T$.
Note that every $(\ell,\eps)$-quantum-secure $t$-out-of-$p$ augmented leakage-resilient secret sharing scheme is also average-case with the same parameters.

\subsubsection{A $2$-out-of-$2$ scheme}

On the way to constructing a quantum-secure $2$-out-of-$p$ augmented leakage-resilient secret sharing scheme, we first focus on the special case where $t=p=2$.
We show that a well-known secret sharing scheme based on the inner product extractor is a quantum-secure augmented leakage-resilient scheme with good parameters.

Suppose that we wish to share a $b$-bit classical message.
Let $q=2^b$ and recall the inner product function $\IP:\F_q^N\times\F_q^N\to\F_q$ given by
\begin{equation*}
    \IP(x,y) =\sum_{i=1}^N x_i y_i,
\end{equation*}
where operations are performed over $\F_q$. Consider the following coding scheme $(2\lrenc,2\lrdec)$.
On input a message $s\in\F_q$, the sharing procedure $2\lrenc(s)$ works as follows:
\begin{enumerate}
    \item Sample $(X,Y)\leftarrow \IP^{-1}(s)=\{(x,y)\in\F_q^N\times\F_q^N:\IP(x,y)=s\}$;
    \item Set the two shares as $S_1=X$ and $S_2=Y$.
\end{enumerate}
To reconstruct the secret, we set $\lrdec(S_1,S_2)=\IP(S_1,S_2)$.
Note that both the sharing and reconstruction procedures run in time $\poly(N,b)$.
The following theorem, proved using the extractor properties of $\IP$ (see \cref{fact:IPqpa}), states the main properties of this scheme.

 \begin{theorem}\label{thm:22augmented}
    The coding scheme 
    $(2\lrenc , 2\lrdec)$ defined above is an $(\ell,\eps)$-quantum-secure $2$-out-of-$2$ augmented leakage-resilient secret sharing scheme provided that the share size $N\cdot b$ (in bits) satisfies
    \begin{equation*}
        N\cdot b \geq 9b + 2\ell+ 8 \log(1/\eps)+40.
    \end{equation*}
    Moreover, the sharing and reconstruction procedures run in time $\poly(N,b)$.
 \end{theorem}
 \begin{proof}
     The correctness and runtime of this coding scheme are straightforward.
     We focus on establishing leakage-resilience.
     We briefly discuss the intuition behind the proof.
     First, we consider an ``ideal'' case where the message $\sigma_M$ is uniformly random.
     This corresponds to the two shares, $X$ and $Y$, being independent and uniformly random.
     In this unrealistic case we can invoke the randomness extraction properties of $\IP$ from \cref{fact:IPqpa} to conclude that $\IP(X,Y)$ will still be extremely close to uniformly random even when conditioned on, say, $X$ along the entangled state $W_1$ and the leakage $Z_2$ from $Y$.
     Then, we move to the real case where $\IP(X,Y)=s$ for an arbitrary fixed message $s$ by conditioning the ideal distribution above on this event, which happens with probability $2^{-b}$.
     This incurs a blowup factor of $2^b$ on the final statistical error when moving from the ideal case to the real case, which we can survive.

     Without loss of generality, fix $T=\{1\}$ and some leakage adversary $\cA=(\Phi_2,\ket\psi_{W_1 W_2})$.
     First, we consider the ``ideal'' case where $X$ and $Y$ are independent and uniformly distributed over $\F_q^N$.
    Define $\rho = \rho_{X \hat{X}} \otimes  \rho_{Y \hat{Y}} $, where $\hat{X}$ and $\hat{Y}$ are canonical purification registers of $X$ and $Y$, respectively. 
    Let $Z_2$ be the leakage register after applying the leakage function $\Phi_2$ to the contents of registers $Y$ and $W_2$.
    Invoking \cref{fact:stinespring}, let $V_{\Phi_2} : \cH_Y \otimes \cH_{W_2} \to  \cH_Y \otimes \cH_{Z_2}  \otimes \cH_{W'_2}$ be the Stinespring isometry corresponding to the leakage map $\Phi_2$. Denote by $\gamma$ the state $ \rho \otimes \ketbra{\psi}$ 
    and $\tau$ by the state obtained by applying the leakage attack defined by $\cA$ to registers $(Y,W_2)$, i.e.,
    \begin{equation*}
        \tau = V_{\Phi_2} \gamma V^\dagger_{\Phi_2}.
    \end{equation*}

We begin by lower bounding the min-entropy of certain random variables in order to be able to apply the extractor properties of $\IP$.
By independence of $X$ and $Y$, we have that
\begin{equation}\label{eq:hminboundsnew}
        \hmin{X}{Y\hat{Y}W_2}_\gamma \geq N \log q \quad \textrm{and} \quad \hmin{Y}{X\hat{X} W_1}_\gamma \geq N \log q  .
    \end{equation}
In fact, we claim that we actually have the inequalities
    \begin{equation}\label{eq:hminbounds}
        \hmin{X}{Y\hat{Y} W'_2}_\tau \geq N \log q \quad \textrm{and} \quad \hmin{Y}{X\hat{X} W_1 Z_2}_\tau \geq N \log q -2\ell.
     \end{equation}
The leftmost inequality of \cref{eq:hminbounds} follows by combining the leftmost inequality of \cref{eq:hminboundsnew} with \cref{fact:PostProcHmin}, since 
$\tau_{YZ_2 W'_2} = V_{\Phi_2} \gamma_{YW_2}V^\dagger_{\Phi_2}$. To see the rightmost inequality of \cref{eq:hminbounds}, first note that $\tau_{YX\hat{X}W_1} = \tau_Y \otimes \tau_{X\hat{X}W_1}=U_Y\otimes \tau_{X\hat{X}W_1}$. 
Furthermore, we have that
\begin{equation}\label{eq:imaxUB}
    \imax(Y: X\hat{X}W_1 Z_2)_\tau \leq 2 \vert Z_2 \vert \leq 2\ell,
\end{equation}
where we recall that $\imax$ is the max-information and the inequality follows from \cref{fact:boundnew}. The rightmost inequality of \cref{eq:hminbounds} now follows from \cref{def:maxinfo} and \cref{def:condminentropy} after noting that $\tau_Y = U_Y$.

Now, set $Z=\IP(X,Y)$ and store this string in register $M$. 
Taking into account \cref{eq:hminbounds}, we invoke \cref{fact:IPqpa} with the assignment of states/registers
    (the states/registers on the left below are from the statement of \cref{fact:IPqpa}, while the states/registers on the right are those from this proof)
    \begin{equation*}
        \left(\rho, W_1,W_2,X,Y\right)
    \leftarrow 
    \left( \tau, W_1Z_2 ,W_2',X,Y \right)
    \end{equation*}
to conclude that
     \begin{equation}\label{eq:finalgoalunif}
         \tau_{M X W_1 Z_2} \approx_{\eps'=\eps\cdot 2^{-b}} U_M \otimes \tau_{X W_1 Z_2},
     \end{equation}
     provided that $N\log q = N\cdot b \geq 9b + 2\ell +8\log(1/\eps)+40$.
     
     To conclude the proof, we would like to show that \cref{eq:finalgoalunif} holds not only when $X$ and $Y$ are uniformly random over $\F_q^N$, but also when $X$ and $Y$ are sampled uniformly at random from any preimage $\IP^{-1}(s)$ with $\sigma_M=s\in\F_q$ any fixed message.
     Let $\tau'$ denote the state $\tau$ conditioned on the event that $\tau_M=\sigma_M$.
     Then, $\tau'$ corresponds exactly to the state obtained by sharing $s$ using $2\lrenc$ and then applying the leakage attack defined by $\cA$ above on the resulting shares.
     Since $\Pr[U_M=\sigma_M]=2^{-b}$, conditioning both sides of \cref{eq:finalgoalunif} on the event that $\tau_M=\sigma_M$ yields
     \begin{equation*}
         \tau'_{M X W_1 Z_2} \approx_{\eps'\cdot 2^b=\eps} \sigma_M \otimes \tau_{X W_1 Z_2},
     \end{equation*}
     as desired.
 \end{proof}

\subsubsection{A $2$-out-of-$p$ scheme}\label{sec:2nLRSS}

We now use our $2$-out-of-$2$ augmented leakage-resilient secret sharing scheme to obtain a $2$-out-of-$p$ scheme for any number of parties $p\geq 3$.
This is accomplished by following the high-level approach of Goyal and Kumar~\cite{GK16} in the classical setting.

Let $(2\lrenc,2\lrdec)$ be our secret sharing scheme from \cref{thm:22augmented} for $b$-bit messages with share size $N\cdot b$.
Then, we define our $2$-out-of-$p$ secret sharing scheme $(\lrenc,\lrdec)$ for $b$-bit messages as follows.
To share a message $s\in\bits^b$ among $p$ parties, the sharing procedure $\lrenc(s)$ proceeds as follows:
\begin{enumerate}
    \item For each pair $(i,j)\in[p]\times[p]$ with $i<j$, compute the $2$-out-of-$2$ secret sharing $(X^j_i,X^i_j)\leftarrow 2\lrenc(s)$.

    \item Set the final $i$-th share as $S_i=(X^1_i,X^2_i,\dots,X^{i-1}_i,X^{i+1}_i,\dots,X^p_i)$ for each $i\in[p]$.
\end{enumerate}
For an authorized subset of parties $T$ (i.e., $|T|\geq 2$), the reconstruction procedure $\lrdec(S_T)$ works as follows: 
Pick any two indices $i,j\in T$ such that $i<j$.
Then, output $s=2\lrdec(X^j_i,X^i_j)$.

Note that $(\lrenc,\lrdec)$ above has share size $(p-1)\cdot N\cdot b$.
The following theorem states the secret sharing and leakage-resilience properties of $(\lrenc,\lrdec)$. 
\begin{theorem}\label{thm:2naugmented}
    The coding scheme $(\lrenc,\lrdec)$ is an $(\ell,\eps)$-quantum-secure $2$-out-of-$p$ augmented leakage-resilient secret sharing scheme whenever the share size $(p-1)\cdot N\cdot b$ satisfies
    \begin{equation*}
        (p-1)\cdot N\cdot b\geq (p-1)(9b + 2\ell+ 8 \log(1/\eps)+16\log p+40).
    \end{equation*}
    Moreover, the sharing and reconstruction procedures run in time $\poly(p,N,b)$.  Furthermore, the fixed leaked state in the leakage-resilience property of~\cref{def:2naqlrss}  can be obtained by running the leakage tampering experiment for a uniform input  message.
\end{theorem}
\begin{proof}
    The correctness of this coding scheme is straightforward.
    Therefore, we focus on showing leakage-resilience, which follows via a hybrid argument analogous to that of~\cite[Theorem 6]{GK16}.

    Without loss of generality we may set $T=\{1\}$.
    Fix a secret $s\in\F_q$ and an arbitrary leakage adversary $\cA=(\Phi_2,\dots,\Phi_p,\ket\psi_{W_1 W_2\dots W_p})$, where $\Phi_i : \cL(\cH_{S_i} \otimes \cH_{W_i} ) \to \cL(\cH_{Z_i} )$ with $Z_i$ being the leakage register for every $i \in [2,p]$.

  Let $\sigma^s_{S_1S_2 \ldots S_p} $ where $S_i \equiv X^1_i\dots X^{i-1}_iX^{i+1}_i\dots X^p_i$ be the $p$ shares corresponding to the encoding of secret $s$, and let 
    \[
        \rho^s_{S_1W_1Z_2 \dots Z_p} = (\Phi_2 \otimes \ldots \otimes  \Phi_p)  (\sigma^s_{S_1S_2 \ldots S_p} \otimes \ketbra{\psi}_{W_1 W_2\dots W_p}). 
    \]
 Let $H^s_{0,0}$ denote the original leakage experiment when run on $\sigma^s_{S_1S_2 \ldots S_p}$, which outputs $\rho^s_{S_1W_1Z_2 \dots Z_p}$. 
    For all pairs $(i,j)\in[p]\times[p]$ with $i<j$, let $Y^j_i$ and $Y^i_j$ be independent and uniformly distributed over $\F_q^N$. We consider hybrids $H^s_{i,j}$ where for all pairs $(i',j')$ up to $(i,j)$ in lexicographic order we replace $X^j_i$ and $X^i_j$ by $Y^j_i$ and $Y^i_j$, respectively.
    
    Note that the output of $H^s_{p-1,p}$ is independent of the secret $s$, since it can be obtained by running the leakage tampering experiment for a uniform input message. Furthermore, there are $\binom{p}{2} < p^2$ hybrids. 
    Therefore, the theorem statement will follow if we show that the outputs of any two consecutive hybrids are $(\eps/p^2)$-close in trace distance and then repeatedly apply the triangle inequality.
    We show here that
    \begin{equation}\label{eq:hybclose}
        H^s_{0,0}\approx_{\eps/p^2} H^s_{1,2}.
    \end{equation}
    The desired analogous statement for any given pair of consecutive hybrids follows in exactly the same manner.
    Recall that the only thing that changes from $H^s_{0,0}$ to $H^s_{1,2}$ is that $X^2_1$ and $X^1_2$ are replaced by $Y^2_1$ and $Y^1_2$, respectively.
    Furthermore, observe that, since $s$ is fixed, we may sample any other pair $(X^j_i,X^i_j)\leftarrow\IP^{-1}(s)$ for index pairs $(i,j)\neq (1,2)$ using randomness independent of $(X^2_1,X^1_2,Y^2_1,Y^1_2)$.
    Consider an arbitrary fixing $(X^j_i,X^i_j)=(x^j_i,x^i_j)$ of all pairs $(i,j)\neq (1,2)$ with $i<j$ such that $\IP(x^j_i,x^i_j)=s$. Let us denote the arbitrary fixing as $\mathsf{fix}$, and $\sigma^{\mathsf{fix}}_{\mathsf{rest}}$  to mean that the rest of the registers except $(X_1^2, X_2^1)$ are set to $\mathsf{fix}$. Let $ H^{s,\mathsf{fix}}_{0,0}$ be the original tampering experiment when run after fixing $\mathsf{fix}$.
    Then, we have that
    \begin{equation*}
        H^{s,\mathsf{fix}}_{0,0} =  (\Phi_2 \otimes \ldots \otimes  \Phi_p) ( \sigma^s_{X_1^2 X_2^1} \otimes \sigma^{\mathsf{fix}}_{\mathsf{rest}}  \otimes \ketbra{\psi}_{W_1 \cdots W_p}   )
    \end{equation*}
    and
    \begin{equation*}
        H^{s,\mathsf{fix}}_{1,2} = (\Phi_2 \otimes \ldots \otimes  \Phi_p) (  U_{Y_1^2} \otimes U_{Y_2^1}  \otimes \sigma^{\mathsf{fix}}_{\mathsf{rest}}  \otimes \ketbra{\psi}_{W_1 \cdots W_p}   ).
    \end{equation*}
    Furthermore, it holds that
    \[ H^{s}_{0,0}  = \E_{\mathsf{fix}} H^{s,\mathsf{fix}}_{0,0} \quad \textrm{and} \quad H^{s}_{1,2}  = \E_{\mathsf{fix}} H^{s,\mathsf{fix}}_{1,2},  \]where $\E_{\mathsf{fix}}$ stands for expectation taken over all the fixings. If we show that for every fixing $\mathsf{fix}$, $ H^{s,\mathsf{fix}}_{0,0} \approx_{\eps/p^2} H^{s,\mathsf{fix}}_{1,2},$ we will be done as this implies $ H^{s\mathsf{}}_{0,0} \approx_{\eps/p^2} H^{s\mathsf{}}_{1,2}$ by invoking \cref{fact:traceconvex}. We now proceed to show that $ H^{s,\mathsf{fix}}_{0,0} \approx_{\eps/p^2} H^{s,\mathsf{fix}}_{1,2}.$

    By post-processing of trace distance (\cref{fact:data}) it suffices to show that
    \begin{equation}\label{eq:preprocess}
       \Phi_2  ( \sigma^s_{X_1^2 X_2^1} \otimes \sigma^{\mathsf{fix}}_{\mathsf{rest}}  \otimes \ketbra{\psi}_{W_1 \cdots W_p}   ) \approx_{\eps/p^2} \Phi_2  (  U_{Y_1^2} \otimes U_{Y_2^1}  \otimes \sigma^{\mathsf{fix}}_{\mathsf{rest}}  \otimes \ketbra{\psi}_{W_1 \cdots W_p}  ).
    \end{equation}
    Noting that $\Phi_2$ is an $\ell$-bounded leakage function we see that \cref{eq:preprocess} follows from \cref{thm:22augmented} and the definition of quantum-secure augmented leakage-resilient secrete sharing, provided that $N\log q =  N\cdot b \geq 9b + 2\ell +8\log(1/\eps)+ 16 \log p+40$. This completes the proof. 
\end{proof}

\subsection{Our candidate non-malleable secret sharing scheme for quantum messages}\label{sec:qnmssscheme}

We proceed to leverage our split-state non-malleable code for quantum messages from \cref{thm:mainqnmc1} and quantum-secure augmented leakage-resilient secret sharing scheme for classical messages from \cref{thm:2naugmented} to construct non-malleable threshold secret sharing schemes for quantum messages.
The high-level approach we follow can be seen as the quantum analogue of the construction of classical non-malleable secret sharing schemes due to Goyal and Kumar~\cite{GK16}.
We exploit two crucial properties of our split-state non-malleable code for quantum messages: First, one of its states is fully classical.
Second, it is also a $2$-out-of-$2$ secret sharing scheme for quantum messages.

Fix a threshold $t\geq 3$ and a number of parties $p$ such that $t\leq p\leq 2t-1$.
We will require the following objects:
\begin{itemize}
    \item A $2$-out-of-$2$ $(\epsnm=\eps,\epspriv=\eps)$-non-malleable secret sharing scheme $(2\nmenc,2\nmdec)$ for quantum messages of length $b$ with a quantum left share of length $b_1$ and a classical right share of length $b_2$, guaranteed by \cref{thm:22nmss};

    \item A $t$-out-of-$p$ secret sharing scheme $(\qshare,\qrec)$ for quantum messages of length $b_1$ with shares of length at most $\ell$, guaranteed by \cref{fact:qss};

    \item A $2$-out-of-$p$ $(\ell,\epslk)$-quantum-secure 
    augmented leakage-resilient secret sharing scheme $(\lrenc,\lrdec)$ for classical messages of length $b_2$, guaranteed by \cref{thm:2naugmented}.
\end{itemize}

We proceed to describe our candidate $t$-out-of-$p$ scheme $(\nmshare,\nmrec)$.
On input a quantum message $\sigma_M$ with canonical purification $\sigma_{M\hat{M}}$, the sharing algorithm $\nmshare(\sigma)$ proceeds as follows:
\begin{enumerate}
    \item Compute the split-state encoding $\rho_{LR}=2\nmshare(\sigma_M)$, where $R$ is a classical register;

     \item Apply $\qshare$ to the contents of register $L$ to obtain $p$ quantum shares stored 
    \item Apply $\lrenc$ to the contents of register $R$ to obtain $p$ classical shares stored in registers $R_1,\dots,R_p$; 

    \item Form the $i$-th final share $S_i=(L_i,R_i)$.
\end{enumerate}

The correctness and runtime of this coding scheme are clear.
In the following sections we establish its statistical privacy and non-malleability.

\subsection{Statistical privacy}\label{sec:statpriv}

Consider an arbitrary message $\sigma_M$ with canonical purification $\sigma_{M\hat{M}}$.
Let $\rho_{\hat{M} LR} = 2\nmshare(\sigma_{M\hat{M}})$ be the split-state encoding of $\sigma$ with $L$ the quantum register and $R$ the classical register. Let $\gamma = \qshare(\rho)$ denote the state after encoding the contents of register $L$ using the $t$-out-of-$p$ secret sharing scheme $(\qshare,\qrec)$ for quantum messages, and let $\theta= \nmshare(\sigma_{M\hat{M}})= \lrenc (\gamma)$ denote the final state after encoding $\sigma$.
In order to establish the desired privacy property for $(\nmshare,\nmrec)$, we must show that for any subset of parties $T\subseteq[p]$ of size $|T|\leq t-1$ it holds that
\begin{equation*}
    \theta_{\hat{M}S_T}\approx_{\epspriv} \sigma_{\hat{M}}\otimes \zeta_{S_T},
\end{equation*}
for a fixed state $\zeta_{S_T}$ (recall the notation $S_T = (S_i)_{i \in T}$).

Since $\gamma = \qshare(\rho)$, invoking \cref{fact:qsshiding} with the assignment of registers (where the registers on the left are those from the statement of \cref{fact:qsshiding} and the registers on the right are the registers used in this proof)
\begin{equation*}
    (\rho, \sigma, \sigma_{E}, \sigma_M) \leftarrow ( \gamma, \rho, \rho_{\hat{M}R}, \rho_{L})
\end{equation*}
ensures that
\begin{equation}\label{eq:leftsharesindep}
    \gamma_{\hat{M}R L_{T}}  \equiv \rho_{\hat{M}R} \otimes U_{L_T}.
\end{equation} 
Since $(2\nmshare,2\nmrec)$ satisfies $\epspriv$-privacy, we also have that
\begin{equation}\label{eq:22privacy}
   \rho_{\hat{M}R} \approx_{\epspriv}  \sigma_{\hat{M}} \otimes  U_{R}.
\end{equation}
Thus, combining \cref{eq:leftsharesindep,eq:22privacy} with the triangle inequality yields
\begin{equation*}
    \gamma_{\hat{M}RL_{T}} \approx_{\epspriv}  \sigma_{\hat{M}} \otimes  U_{R}\otimes U_{L_T}.
\end{equation*}
Since $\theta_{} =  \lrenc (\gamma)$, the post-processing property of trace distance (\cref{fact:data}) implies that
\begin{equation}\label{eq:statprivacy1}
    \theta_{\hat{M} R_{T}L_{T}}  \approx_{\epspriv} \sigma_{\hat{M}} \otimes  \lrenc (U_R)_{R_T} \otimes U_{L_{T}}.
\end{equation}
We conclude that
\begin{equation*}
    \theta_{\hat{M} S_T}\equiv \theta_{\hat{M} R_{T}L_{T}}  \approx_{\epspriv} \sigma_{\hat{M}} \otimes \zeta_{S_{T}},
\end{equation*}
where $\zeta_{S_{T}}=  \lrenc (U_R)_{R_T} \otimes U_{L_{T}}$ is the fixed state.

\subsection{Average-case non-malleability}
\label{sec:avgnmqss}

We proceed to show average-case non-malleability for our coding scheme $(\nmshare,\nmrec)$.
In this case we take $\sigma_M=U_M$, and so the canonical purification $\sigma_{M\hat{M}}$ is maximally entangled.
Without loss of generality, we may take $T=\{1,\dots,t\}$.

Fix an arbitrary tampering adversary $\cA=(U_1,U_2, \ldots, U_t , \ketbra{\psi}_{E_1E_2 \ldots E_p})$, where the quantum registers $E_1,\dots, E_p$ store an arbitrary quantum state to be used as shared entanglement across the several tampering adversaries.
Our goal is to reduce the non-malleability of $(\nmshare,\nmrec)$ to the non-malleability of $(2\nmshare,2\nmrec)$.
In order to do this, we show that if $\cA$ breaks the  $(\epsnm+2 \sqrt{\epslk})$-non-malleability of $(\nmshare,\nmrec)$ via the authorized subset $T$, then we can build, using $\cA$, a split-state adversary $\cA'=(U,V,\ketbra\gamma_{W_1W_2})$ that breaks the $\epsnm$-non-malleability of $(2\nmshare,2\nmrec)$.
Since we know that $(2\nmshare,2\nmrec)$ is $\epsnm$-non-malleable, the desired result follows.

To build the $2$-split-state tampering adversary $\cA'$ from our original adversary $\cA$, we would like to follow the high-level approach of Goyal and Kumar~\cite{GK16} in the classical setting and extend it to the quantum setting. 
However, as already mentioned in \cref{sec:techoverviewNMSS}, naively replicating this argument in the quantum setting is not possible since the quantum side information cannot be fixed.
We require a quantum-secure \emph{augmented} leakage-resilient secret sharing scheme to overcome this issue.
Another thing that complicates the approach in the quantum setting is
that the register $L$ in the split-state encoding $(L,R)$ of the message $\sigma_M$ is quantum and so is lost after tampering. 
We get around this issue by considering ``coherent copies'' of $L$ and $R$ which we will be able to access in our analysis even after $L$ has been destroyed.

To be more precise and because it will be useful during the following argument to generate coherent copies of registers, we recall some properties of $(2\nmshare,2\nmrec)$. 
First, $2\nmshare(\sigma_M)$ is composed of two parts, $R$ and $L=(L^1,L^2)$, where $R=U_R$ and $L^1=U_{L^1}$ are classical registers, and $L^2$ is the quantum register. 
Moreover, we have that $L^2 = C_K \sigma_M C_K^\dagger =U_{L^2}$ with $K=\nmext(R,L^1)$, where $\nmext$ is the explicit quantum-secure 2-source non-malleable extractor from \cref{lem:qnmcodesfromnmext}, and $C_K$ is an appropriate Clifford gate. It follows from the procedure that 
\begin{equation}\label{eq:qnmssencodingrho}
   2\nmshare(\sigma_M)=  U_L \otimes U_R.
\end{equation}

Note that $2\nmshare$ is a CPTP map and not an isometry -- to access ``coherent copies'' of $L$ and $R$, we will have to consider the Stinespring isometry extention of $2\nmshare$. 
Let  $V_{2\nmshare} :  \cH_M \to (\cH_{L} \otimes  \cH_{R}   \otimes \cH_{ F})$ be the Stinespring isometry extension (see \cref{fact:stinespring}) of $2\nmshare$, where the register $F$ in the definition of $V_{2\nmshare}$  is the external ancilla register required to implement the operation $2\nmshare$ as an isometry. 
Let $ \nu  = V_{2\nmshare} \sigma_{M \hat{M}} V^\dagger_{2\nmshare}.$ 
Observe that $\nu$ is a pure state and  $2\nmshare(\sigma_M)=  \tr_{\hat{M}F} \left( V_{2\nmshare} \sigma_{M \hat{M}} V^\dagger_{2\nmshare} \right). $ Combining this with \cref{eq:qnmssencodingrho}, we have that
\begin{equation}\label{eq:qnmssencodingphi}
   \tr_{\hat{M}F} \left(\nu \right) =  U_L \otimes U_R.
\end{equation}

Let $\rho_{L\hat{L}R\hat{R} }$ be the canonical purification of $\rho_{LR} = U_L \otimes U_R$.
We now invoke Uhlmann's theorem (\cref{fact:uhlmann}) with the following assignment of registers (where the registers on the left are those from the statement of \cref{fact:uhlmann} and the registers on the right are the registers used in this proof):
\begin{equation*}
    (\sigma_A, \rho_A, \sigma_{AC}, \rho_{AB}) \leftarrow ( \nu_{LR} ,  \rho_{LR}, \nu_{LRF\hat{M} }, \rho_{LR\hat{L}\hat{R}}).
\end{equation*}
This guarantees the existence of an isometry $\mathsf{Uhlmann} : \cH_{\hat{M}} \otimes \cH_F  \to \cH_{\hat{L}} \otimes \cH_{\hat{R}}$ such that
\begin{equation*}
   \Delta_B \left( \rho_{LR\hat{L}\hat{R}} , (\mathsf{Uhlmann}) \nu_{LRF\hat{M} } (\mathsf{Uhlmann}^\dagger) \right) = \Delta_B( \rho_{LR},  \nu_{LR}) =0,
\end{equation*}
 where the rightmost inequality holds because $\rho_{LR} = \nu_{LR} = U_L \otimes U_R$. 
 Thus, we conclude that the state $\rho$ as depicted in \cref{fig:22nmssq} satisfies
\begin{equation}\label{eq:qnmssrho}
    \rho_{L\hat{L}R\hat{R} } = \rho_{L\hat{L}} \otimes \rho_{ R\hat{R} }  \quad \textrm{and} \quad  \rho_{LR} = U_L \otimes U_R.
\end{equation}

With access to coherent copies of $\hat{L}, \hat{R}$ in state $\rho$, we now move on to consider the two relevant tampering experiments are described in \cref{fig:22nmssq,fig:nmssqq1} for some adversaries $\cA'$ and $\cA$, respectively.
We will refer to states in those figures throughout our argument.
First, observe that the state $\rho$ is the same in \cref{fig:22nmssq,fig:nmssqq1}. 
Moreover, the procedures applied to the states $\tau'$ and $\tau$ in \cref{fig:22nmssq,fig:nmssqq1}, respectively, are also the same.
This means that in order to prove the desired result it suffices to show that given any adversary $\cA$ in \cref{fig:nmssqq1} we can come up with a corresponding $2$-split-state adversary $\mathcal{A}' = (U,V, \ketbra{\gamma}_{W_1W_2})$ (which transforms the state $\rho$ into $\tau'$ in \cref{fig:22nmssq}) such  that $\tau'$ in \cref{fig:22nmssq} is  appropriately close (in trace distance) to the state $\tau$ in \cref{fig:nmssqq1}.

As a stepping stone towards defining $\cA'$, we begin by studying the state $\theta$ in \cref{fig:nmssqq1}, which is obtained from $\rho$ through an application of the $\qshare$ and $\lrenc$ procedures.
Note that $\theta$ is not pure -- to facilitate the analysis, we will from here onwards work with the pure state extension of $\theta$, which, because the context is clear, we will also call $\theta$.
More precisely, let  $V_{\qshare} :  \cH_L \to (\cH_{L_1} \otimes \cdots \otimes \cH_{L_p}   \otimes \cH_{ \tilde{L}})$ be the Stinespring isometry extension (see \cref{fact:stinespring}) of $\qshare$, and let
$V_{\lrenc} :  \cH_R \to (\cH_{R_1} \otimes \cdots \otimes \cH_{R_p}   \otimes \cH_{ \tilde{R}_1} \otimes \cdots \otimes \cH_{\tilde{R}_p})$ be the Stinespring isometry extension (see \cref{fact:stinespringclassicalmap}) of $\lrenc$.
The registers $\tilde{L}$ and $\tilde{R}_1,\dots,\tilde{R}_p$ in the definitions of $V_{\qshare}$ and $V_{\lrenc}$, respectively, are the external ancilla registers required to implement the operations $\lrenc$ and $\qshare$ (which are CPTP maps) as isometries. Then, we define
\begin{equation*}
    \theta \defeq (V_{\qshare}  \otimes V_{\lrenc} ) \rho (V^\dagger_{\qshare}  \otimes V^\dagger_{\lrenc} ).
\end{equation*}

We also consider the following intermediate state $\theta'$ obtained from $\theta$ by tampering with the first $t-1$ shares according to $\cA$, i.e.,
\begin{equation*}\label{eq:qnmssthetaprime}
    \theta' \defeq (U_1 \otimes \ldots \otimes U_{t-1} \otimes \id_t \otimes \ldots \id_p ) (\theta \otimes \ketbra{\psi}_{E_1 \ldots E_p}) (U^\dagger_{1}  \otimes  \ldots \otimes U^\dagger_{t-1}\otimes \id^\dagger_t \otimes \ldots \id^\dagger_p  ).
\end{equation*}
The registers associated with $\theta'$ are found in \cref{fig:nmssqq1} after applying $U_1,\dots,U_{t-1}$ to $\theta$.
In what follows we will need to refer to several registers of $\theta'$.
To help with readability, we define
\begin{equation*}
    W'_1 \equiv R_t L_{[t-1]} E_{[t,p]}  \tilde{L}_{} L_{[t,p]} \quad \textrm{and}  \quad  W'_2 \equiv  R_{[t-1]}E_{[t-1]} R_{[t+1,p]} \tilde{R}_{[p]}  
\end{equation*}
as shorthand, where we recall that $[t,p]$ denotes the set $\{t,t+1,\dots,p\}$.
The notation $W'_1$ and $W'_2$ is justified by the fact that we will use $\gamma_{W_1 W_2}\equiv \theta'_{W_1 W_2}$ with
\begin{equation*}
    W_1 \equiv R_t L_{[t-1]} E_{[t,p]} \hat{L} \tilde{L} L_{[t,p]} \quad \textrm{and}  \quad W_2 \equiv  R_{[t-1]} E_{[t-1]} \hat{R} R_{[t+1,p]} \tilde{R}_{[p]}  ,
\end{equation*}
obtained by adding $\hat{R}$ to $W'_2$ and $\hat{L}$ to $W'_1$, respectively, as the shared entanglement for the $2$-split-state adversary $\cA'$. 

We now present some important properties of the state $\theta'$. Their proofs appear later in \cref{sec:qss1,sec:qss2}.
The first claim states that the coherent copy $\hat{L}$ defined above is independent of $\hat{R}$ even when $W'_2$ is revealed as side information.
\begin{restatable}{claim}{qssone}\label{claim:qss1}
    We have that
    \begin{equation*}
        \theta'_{\hat{L} \hat{R}W'_2 } =  \theta'_{\hat{L}} \otimes   \theta'_{\hat{R}W'_2}.
    \end{equation*}
\end{restatable}
The second claim states that $\hat{R}$ is close to being independent of $\hat{L}$ given the side information $W'_1$.
This uses the augmented leakage-resilience property of the $(\lrenc,\lrdec)$ scheme.
\begin{restatable}{claim}{qsstwo}\label{claim:qss2}
    We have that
    \begin{equation*}
        \theta'_{\hat{R} \hat{L}W'_1 } \approx_{\epslk} \theta'_{\hat{R} } \otimes \theta'_{\hat{L} W'_1 }.
    \end{equation*}
\end{restatable}

We use these claims as starting points towards constructing the tampering maps $U$ and $V$ of the $2$-split-state adversary $\cA'$.
First, we show that there exist maps $U' : \cH_{L} \otimes \cH_{W_1 } \to \cH_{W'_1 }$ and $V' : \cH_{R  } \otimes \cH_{W_2} \to \cH_{W'_2}$  and a pure state $\gamma_{W_1  W_2}$ that depends only on $(U_1, U_2, \ldots , U_{t-1}, \ketbra{\psi}_{E_1 \ldots E_p} )$ such that
\begin{equation}\label{eq:thetaprimeleak}
    \theta'_{W_1W_2} \approx_{2\sqrt{\epslk}} (U' \otimes V')   ( \rho \otimes \gamma_{W_1 W_2} ) ((U')^\dagger \otimes (V')^\dagger).
\end{equation}
\cref{eq:thetaprimeleak} is a consequence of a careful instantiation of the following fact.
A variation of the fact below in terms of mutual information already appears in~\cite[Lemma II.15]{JPY14}. For the sake of completeness, we provide a proof of the version we need in \cref{sec:qsstechnical}.
\begin{restatable}{fact}{qsstechnical}\label{lem:qsstechnical}
    Let $\phi_{XABY}$ be a pure state such that 
    \begin{equation*}
        \phi_{X BY  } \approx_{\eps_1} \phi_{X  } \otimes \phi_{ BY  } \quad \textrm{and} \quad \phi_{YAX   } \approx_{\eps_2} \phi_{Y  } \otimes \phi_{AX} .
    \end{equation*}
    Furthermore, let $\kappa$ be a pure state such that
    \begin{equation*}
        \kappa_{X_1 \hat{X}_1 Y_1 \hat{Y}_1} = \kappa_{X_1 \hat{X}_1  } \otimes \kappa_{Y_1 \hat{Y}_1},
    \end{equation*}
    where $\hat{X}_1$ and $\hat{Y}_1$ are the canonical purification registers of $X_1$ and $Y_1$, respectively, and
    \begin{equation*}
        \kappa_{X_1} \equiv \phi_{X} \quad \textrm{and} \quad \kappa_{Y_1} \equiv \phi_Y.
    \end{equation*} 
    Then, there exist two isometries $U' : \cH_{\hat{X}_1} \otimes \cH_A \otimes \cH_{X} \to \cH_{A}$ and $V' : \cH_{\hat{Y}_1} \otimes \cH_B \otimes \cH_{Y} \to \cH_{B}$ alongside a pure state $\gamma_{XABY}$ such that 
    \begin{equation*}
        \phi_{XABY} \approx_{2(\sqrt{\eps_1}  + \sqrt{\eps_2})} (U' \otimes V') ( \kappa \otimes \gamma) ((U')^\dagger \otimes (V')^\dagger).
    \end{equation*} 
\end{restatable}

To be precise, we obtain \cref{eq:thetaprimeleak} from \cref{lem:qsstechnical} via the assignment of registers (below the registers on the left are from~\cref{lem:qsstechnical} and the registers on the right are the registers in this proof)
\begin{equation*}
    (\phi, \phi_X,\phi_Y,\phi_A,\phi_B, \kappa, \kappa_{X_1},\kappa_{X_2} , \kappa_{\hat{X}_1}, \kappa_{\hat{X}_2}) \leftarrow ( \theta' ,  \theta'_{\hat{L}} , \theta'_{\hat{R}}  ,\theta'_{W'_1},   \theta'_{W'_2} , \rho, \rho_{\hat{L}}, \rho_{\hat{R} }, \rho_{L}, \rho_{R}),
\end{equation*}
and by setting $\eps_1=0$ and $\eps_2=\epslk$, which is justified by \cref{claim:qss1,claim:qss2}, respectively.
Furthermore, we have used \cref{claim:qss1,claim:qss2} to satisfy the conditions of state $\phi$ (which will be state $\theta'$ in this proof) in~\cref{lem:qsstechnical}.

We use $U'$ and $V'$ to construct the CPTP $2$-split-state tampering maps $U : 
 \cL(\cH_{L  } \otimes \cH_{ { W}_1 } ) \to  \cL(\cH_{{ L} })$ and $V : 
 \cL(\cH_{R  } \otimes \cH_{ { W}_2 } ) \to  \cL(\cH_{{ R} })$ of $\cA'$ as follows:
\begin{itemize}
  \item On input the contents of registers $\rho_L$ and $\gamma_{W_1}$, the map $U$ first applies the unitary $U':\cH_L\otimes \cH_{W_1}\to \cH_{W'_1}$ to generate the contents of register $\theta'_{W'_1}$ (approximately).
    In particular, note that registers $L_{[t-1]} R_t L_t E_t$ are part of $W'_1$. Then, it uses the unitary $U_t$ followed by $\qrec$ to generate $\tau_L$ and traces out the remaining registers of $W'_1$.
    \item On input the contents of registers $\rho_R$ and $\gamma_{W_2}$, the map $V$ first applies the unitary $V':\cH_R\otimes \cH_{W_2}\to \cH_{W'_2}$ to generate the contents of register $\theta'_{W'_2}$ (approximately).
    In particular, note that $R_{[2]}$ is part of $W'_2$ (this crucially uses the fact that $t>2$).
    Then, it uses $\lrdec$ to generate $\tau_R$ from $\theta'_{R_{[2]}}$ and traces out the remaining registers of $W'_2$;

\end{itemize}

Let $\tau'_{\hat{L} \hat{R} LR }  = ( U \otimes V)   ( \rho \otimes \gamma_{{ W}_1  { W}_2} )$
be the state obtained after applying the $2$-split-state tampering attack defined by $\cA'=(U,V,\gamma_{W_1 W_2})$ on $\rho$ in \cref{fig:22nmssq}.
By inspecting the definition of $\theta'$ above, the state $\tau_{\hat{L} \hat{R} LR }$ obtained after applying the tampering attack defined by $\cA$ on $\rho$ in \cref{fig:nmssqq1} can be written as
\begin{equation}\label{eq:exprtau}
    \tau_{\hat{L} \hat{R} LR } = (\qrec \otimes \lrdec ) \left(  \left( U_t \theta' U^\dagger_t  \right)_{\hat{L} \hat{R} L_{[t]} R_{[2]}} \right).
\end{equation}
Recall from our discussion above that our main goal is to show that these two states $\tau$ and $\tau'$ are appropriately close in trace distance.
Combining \cref{eq:exprtau} with \cref{eq:thetaprimeleak} and the postprocessing property of trace distance (\cref{fact:data}) yields
\begin{equation}\label{eq:qssfinal}
 \tau_{\hat{L} \hat{R} LR } \approx_{2\sqrt{\epslk}} \tau'_{\hat{L} \hat{R} LR }   ,   
\end{equation}
To see this, first note that registers $\theta'_{L_1\dots L_{t-1}L_tR_tE_t}$ are obtained by applying $U'$ to $\rho_L \otimes \gamma_{W_1}$. Then, the register $\tau_L$ is generated using $\qrec (U_t \theta'_{L_1\dots L_{t-1}L_tR_tE_t} U^\dagger_t )$.
On the other hand, the register $\tau_R$ is generated using $\lrdec$ with input the registers $\theta'_{R_{[2]}}$, which are themselves generated from $\rho_R$ and $\gamma_{W_2}$ via $V'$.

Recalling that the operations applied to $\tau'$ and $\tau$ in \cref{fig:22nmssq,fig:nmssqq1}, respectively, are the same, we conclude from \cref{eq:qssfinal} and postprocessing (\cref{fact:data}) that the final decoded tampered state $\eta$ in \cref{fig:nmssqq1} satisfies
 \begin{equation}\label{eq:qssfinal1}
    \eta_{\hat{M} M } \approx_{2\sqrt{\epslk}}  \tr_F \left( (\mathsf{Uhlmann}^\dagger \otimes 2\nmrec ) (  ( U \otimes  V)   ( \rho \otimes \gamma_{{ W}_1  { W}_2} ) )\right) .
\end{equation}
Moreover, from the $\epsnm$-non-malleability of $(2\nmshare, 2\nmrec)$ we get that 
  \begin{equation}\label{eq:qssfinal2}
   \tr_F \left((\mathsf{Uhlmann}^\dagger  \otimes 2\nmrec ) (  ( U \otimes  V)   ( \rho \otimes \gamma_{{ W}_1  { W}_2} ) ) \right)  \approx_{\epsnm} p_\mathcal{A} \sigma_{\hat{M}M } +(1-p_\mathcal{A} )  \sigma_{\hat{M}} \otimes \zeta^{\mathcal{A}}_{M} ,   
\end{equation}
 where $(p_{\mathcal{A}}, \zeta^{\mathcal{A}}_{M})$ depend only on $(U, V, \gamma_{{ W}_1  { W}_2})$, which further depend only on the initial adversary $\mathcal{A}= (U_1,U_2, \ldots, U_t, \ketbra{\psi}_{E_1\ldots E_p})$.
 Finally, combining \cref{eq:qssfinal1,eq:qssfinal2} with the triangle inequality leads to
 \begin{equation*}
     \eta_{\hat{M} M } \approx_{\epsnm+2\sqrt{\epslk}} p_\mathcal{A} \sigma_{\hat{M}M } +(1-p_\mathcal{A} )  \sigma_{\hat{M}} \otimes \zeta^{\mathcal{A}}_{M},
 \end{equation*}
 which completes the proof of average-case non-malleability.

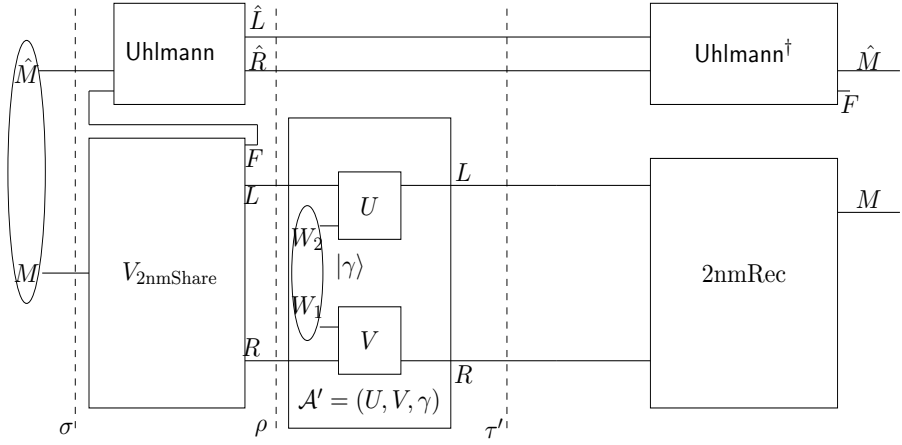
\begin{figure}
\centering
\resizebox{12cm}{6cm}{
\begin{tikzpicture}

\node at (1,4.5) {$\hat{M}$};
\node at (1,1.5) {$M$};
\draw (1.2,4.5) -- (2.4,4.5);
\draw (2.4,4) rectangle (4.5,5.5);
\draw (11,4) rectangle (14,5.5);
\draw (4.5,4.5) -- (11,4.5);
\draw (14,4.5) -- (15,4.5);
\draw (4.5,5) -- (11,5);
\node at (14.5,4.7) {$\hat{M}$};

\draw (14,4.2) -- (14.2,4.2);
\node at (14.2,4) {${F}$};

\node at (4.7,5.3) {$\hat{L}$};
\node at (4.7,4.7) {$\hat{R}$};

\draw (1.25,1.5) -- (2,1.5);
\draw (2,-0.5) rectangle (4.5,3.5);
\node at (3.3,1.5) {$V_{2\nmshare}$};
\node at (3.3,4.8) {$\mathsf{Uhlmann}$};
\node at (12.5,4.8) {$\mathsf{Uhlmann}^\dagger$};
\draw (1.0,3.0) ellipse (0.3cm and 1.95cm);
\draw [dashed] (1.78,-0.8) -- (1.78,5.5);
\draw [dashed] (5,-0.8) -- (5,5.5);
\draw [dashed] (8.7,-0.8) -- (8.7,5.5);

\node at (6.2,1.6) {$\ket{\gamma}_{}$};
\node at (1.64,-0.8) {$\sigma$};
\node at (4.76,-0.8) {$\rho$};
\node at (8.5,-0.8) {${\tau'}$};

\draw (2,4.2) -- (2.4,4.2);
\draw (2,3.7) -- (2,4.2);
\draw (4.7,3.7) -- (2,3.7);
\draw (4.7,3.4) -- (4.7,3.7);
\draw (4.5,3.4) -- (4.7,3.4);
\node at (4.65,3.2) {$F$};
\node at (4.6,2.65) {$L$};
\node at (8,3) {$L$};
\draw (4.5,2.8) -- (6,2.8);
\draw (7,2.8) -- (9.5,2.8);
\draw (9.5,2.8) -- (11,2.8);
\node at (14.5,2.6) {$M$};
\draw (14,2.4) -- (15,2.4);

\node at (4.6,0.4) {$R$};
\node at (8,0.0) {$R$};
\draw (4.5,0.2) -- (6,0.2);
\draw (7,0.2) -- (9.5,0.2);
\draw (9.5,0.2) -- (11,0.2);

\draw (6,2) rectangle (7,3);
\node at (6.5,2.5) {$U$};
\draw (6,0) rectangle (7,1);
\node at (6.5,0.5) {$V$};

\node at (6.5,-0.4) {$\mathcal{A}'=(U,V,\gamma)$};
\draw (5.2,-0.8) rectangle (7.8,3.8);

\draw (5.5,1.5) ellipse (0.25cm and 1cm);
\node at (5.5,2) {$W_2$};
\draw (5.7,2.2) -- (6,2.2);

\node at (5.5,1) {$W_1$};
\draw (5.7,0.7) -- (6,0.7);


\draw (11,-0.5) rectangle (14,3.2);
\node at (12.5,1.5) {$2\nmrec$};

\end{tikzpicture}}
\caption{Tampering experiment for average-case $2$-out-of-$2$ NMSS.}\label{fig:22nmssq}
\end{figure}

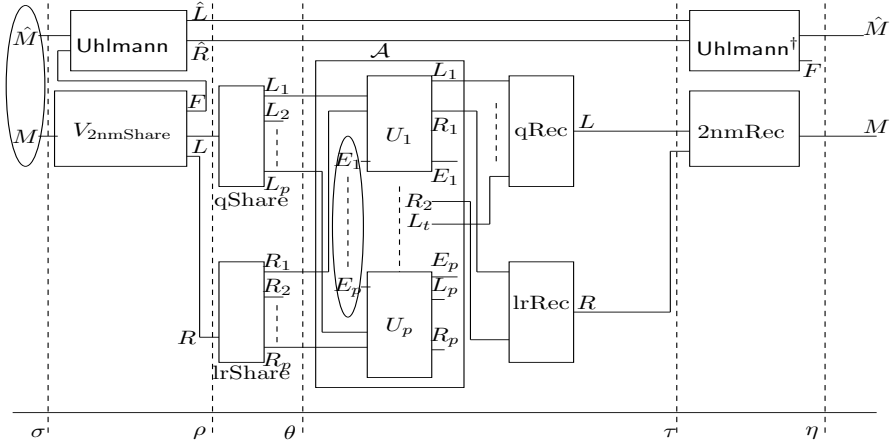
\begin{figure}
\centering
\resizebox{12cm}{6cm}{
\begin{tikzpicture}

\node at (0,6.5) {$\hat{M}$};
\node at (13.2,6.7) {$\hat{M}$};
\draw (0.2,6.5) -- (0.7,6.5);
\draw (2.5,6.8) -- (10.3,6.8);
\draw (2.5,6.4) -- (10.3,6.4);
\draw (12,6.5) -- (13,6.5);
\draw (0.7,5.8) rectangle (2.5,7);
\node at (1.5,6.3) {$\mathsf{Uhlmann}$};
\node at (2.7,7) {$\hat{L}$};
\node at (2.7,6.2) {$\hat{R}$};
\draw (10.3,5.8) rectangle (12,7);
\node at (11.2,6.3) {$\mathsf{Uhlmann}^\dagger$};

\draw (0.45,3.9) rectangle (2.5,5.4);
\draw (10.3,3.9) rectangle (12,5.4);
\node at (0,4.5) {$M$};
\node at (12.2,5.8) {$F$};
\draw (12,6) -- (12.2,6);
\node at (2.65,5.2) {$F$};

\node at (1.5,4.6) {$V_{2\nmshare}$};
\node at (3.5,3.2) {$\qshare$};
\node at (3.5,-0.2) {$\lrenc$};
\node at (8.0,4.6) {$\qrec$};
\node at (8.0,1.2) {$\lrdec$};

\draw (2.5,4.5) -- (3,4.5);
\node at (11.1,4.6) {$2\nmrec$};


\draw (0.2,4.5) -- (0.5,4.5);
\draw (12,4.5) -- (13.2,4.5);
\node at (13.2,4.7) {$M$};
\node at (2.5,0.5) {$R$};

\draw (2.7,0.5) -- (3,0.5);
\draw (3.7,1.8) -- (4.7,1.8);
\node at (3.9,2) {$R_1$};
\draw (3.7,1.3) -- (4,1.3);
\node at (3.9,1.5) {$R_2$};
\draw [dashed] (3.9,0.4) -- (3.9,1.2);
\draw (3.7,0.3) -- (5.3,0.3);
\draw (4.6,0.6) -- (5.3,0.6);
\node at (3.9,0) {$R_p$};
\draw (3,0) rectangle (3.7,2);
\draw (7.5,0) rectangle (8.5,2);
\draw (8.5,1) -- (10,1);
\draw (10,1) -- (10,4.2);
\draw (10,4.2) -- (10.3,4.2);
\node at (8.7,1.2) {$R$};

\draw [dashed] (7.3,4) -- (7.3,5.25);

\draw (5,2.7) ellipse (0.24cm and 1.8cm);

\draw (4.7,5) -- (4.7,1.8);
\draw (4.7,5) -- (5.3,5);
\draw (3.7,5.3) -- (5.3,5.3);
\draw (6.3,5.6) -- (7.5,5.6);
\draw (6.3,5) -- (7,5);
\draw (7,1.8) -- (7,5);
\draw (7,1.8) -- (7.5,1.8);
\draw (6.3,4) -- (6.7,4);

\node at (6.5,2) {$E_p$};
\draw (6.3,1.7) -- (6.7,1.7);
\node at (6.5,1.5) {$L_p$};
\node at (6.1,2.8) {$L_t$};
\draw (6.3,1.25) -- (6.5,1.25);
\draw (7.2,3.7) -- (7.2,2.75);
\draw (6.3,2.75) -- (7.2,2.75);
\draw (7.2,3.7) -- (7.5,3.7);

\draw (6.9,0.45) -- (7.5,0.45);

\draw (6.9,0.45) -- (6.9,3.2);
\draw (6.3,3.2) -- (6.9,3.2);
\node at (6.1,3.2) {$R_2$};
\node at (6.5,0.5) {$R_p$};
\draw (6.3,0.25) -- (6.5,0.25);

\node at (3.9,5.5) {$L_1$};
\draw (3.7,4.8) -- (4,4.8);
\node at (3.9,5) {$L_2$};
\draw [dashed] (3.9,3.9) -- (3.9,4.7);
\draw (3.7,3.8) -- (4.6,3.8);
\draw (4.6,0.6) -- (4.6,3.8);
\node at (3.9,3.5) {$L_p$};
\draw (3,3.5) rectangle (3.7,5.5);
\draw (7.5,3.5) rectangle (8.5,5.7);

\draw (8.5,4.6) -- (10.3,4.6);
\node at (8.7,4.8) {$L$};

\draw [dashed] (5,1.9) -- (5,3.7);
\node at (5,4) {$E_1$};
\node at (6.5,3.7) {$E_1$};
\node at (6.5,4.7) {${R_1}$};
\node at (6.5,5.8) {${L_1}$};

\node at (5,1.5) {$E_p$};


\draw  (2.7,0.5) -- (2.7,4.1);
\draw  (2.5,4.1) -- (2.7,4.1);

\draw [dashed] (0.35,-1.4) -- (0.35,7.2);
\draw [dashed] (2.9,-1.4) -- (2.9,7.2);
\draw [dashed] (10.1,-1.4) -- (10.1,7.2);
\draw [dashed] (4.3,-1.4) -- (4.3,7.2);
\draw [dashed] (12.4,-1.4) -- (12.4,7.2);

\draw (-0.2,-1) -- (13.2,-1);

\node at (0.2,-1.4) {$\sigma$};
\node at (2.7,-1.4) {$\rho$};
\node at (4.1,-1.4) {$\theta$};
\node at (10,-1.4) {$\tau$};
\node at (12.2,-1.4) {$\eta$};

\draw  (0.5,6.2) -- (0.7,6.2);
\draw  (0.5,5.6) -- (0.5,6.2);
\draw  (0.5,5.6) -- (2.8,5.6);
\draw  (2.8,5) -- (2.8,5.6);
\draw  (2.5,5) -- (2.8,5);
\node at (2.7,4.3) {$L$};



\draw (5.3,3.8) rectangle (6.3,5.7);
\node at (5.8,4.5) {$U_1$};

\draw [dashed] (5.8,3.5) -- (5.8,1.8);

\draw (5.3,-0.3) rectangle (6.3,1.8);
\node at (5.8,0.7) {$U_p$};

\node at (5.5,6.2) {$\mathcal{A}$};
\draw (4.5,-0.5) rectangle (6.8,6);




\draw (5.2,1.5) -- (5.35,1.5);
\draw (5.2,4) -- (5.35,4);
\draw (0,5.5) ellipse (0.3cm and 1.6cm);
\end{tikzpicture}}

\caption{Tampering experiment for average-case $t$-out-of-$p$ quantum NMSS.}\label{fig:nmssqq1}
\end{figure}

\subsubsection{Proof of \cref{claim:qss1}}\label{sec:qss1}

For convenience, we restate the relevant claim.
\qssone*
\begin{proof}
    From \cref{eq:qnmssrho}, we have that $\rho_{L\hat{L} R\hat{R} } = \rho_{L\hat{L}} \otimes \rho_{ R\hat{R}} $.  Further we have  $  \theta = (V_{\qshare}  \otimes V_{\lrenc} ) \rho (V^\dagger_{\qshare}  \otimes V^\dagger_{\lrenc} )$. Since $V_{\qshare} $ and $V_{\lrenc}$ act on two different registers $L$ and $R$, respectively, we have that 
    \begin{equation}\label{eq:qssintermediate1}
         \theta_{\hat{L} \tilde{L} L_{[p]} \hat{R}  R_{[p]} \tilde{R}_{[p]} } =\theta_{\hat{L} \tilde{L} L_{[p]} } \otimes \theta_{\hat{R}  R_{[p]} \tilde{R}_{[p]} }.
    \end{equation} 
    Furthermore, using the fact that $V_{\qshare}$ is a $t$-out-of-$p$ secret sharing scheme and invoking \cref{fact:qsshiding}, we have that 
\[ \theta_{\hat{L} L_{[t-1]} }  = \theta_{\hat{L}  }   \otimes  \theta_{ L_{[t-1]} } = \theta_{\hat{L}  }   \otimes  U_{ L_{[t-1]} }  . \]
This means that
\begin{equation*}
    \theta_{\hat{L}L_{[t-1]} \hat{R}  R_{[p]} \tilde{R}_{[p]} }  = \theta_{\hat{L}   }   \otimes  \theta_{ L_{[t-1]} }  \otimes \theta_{\hat{R}  R_{[p]} \tilde{R}_{[p]} },
\end{equation*}
and so by postprocessing (\cref{fact:data}) it follows that
\begin{equation*}
    \theta'_{\hat{L} L_{[t-1]}R_{[t-1]} E_{[t-1]} E_{[t,p]}\hat{R}  R_{[t,p]} \tilde{R}_{[p]} }  = \theta'_{\hat{L}}   \otimes  \theta'_{L_{[t-1]}  R_{[t-1]} E_{[t-1]} E_{[t,p]}\hat{R}  R_{[t,p]} \tilde{R}_{[p]} },
\end{equation*} 
which, again by postprocessing (\cref{fact:data}) and by the definition of $W'_2$, implies the desired statement.
\end{proof}

\subsubsection{Proof of \cref{claim:qss2}}\label{sec:qss2}

For convenience, we restate the relevant claim here.
\qsstwo*
\begin{proof}
We use the $2$-out-of-$p$ $(\ell,\epslk)$-quantum-secure 
augmented leakage-resilient secret sharing property of $\lrenc$ to conclude that  
\begin{equation}\label{eq:eq2Claim2}
    \theta'_{\hat{R} R_t L_{[t-1]} E_{[t,p]} \hat{L} \tilde{L} L_{[t,p]}}  \approx_{\epslk}\theta'_{\hat{R}  }   \otimes  \theta'_{R_t L_{[t-1]} E_{[t,p]} \hat{L} \tilde{L} L_{[t,p]} }.
\end{equation}
Recalling the definition of $W'_2$, this corresponds exactly to the statement of \cref{claim:qss2}.

To see why \cref{eq:eq2Claim2} holds, first recall that
\begin{equation*}
    \theta' = (U_1 \otimes \ldots \otimes U_{t-1} \otimes \id_t \otimes \ldots \id_p ) (\theta \otimes \ketbra{\psi}_{E_1 \ldots E_p}) (U^\dagger_{1}  \otimes  \ldots \otimes U^\dagger_{t-1}\otimes \id^\dagger_t \otimes \ldots \id^\dagger_p ),
\end{equation*} 
with $\theta$ illustrated in \cref{fig:nmssqq1}. 
We now consider the following leakage attack on the $R_{[p]}$ registers: Taking into account the notation from \cref{sec:2lrss}, 
we set $T=[t-1]$ and
for each $i\in[t-1]$ we see $W_i=(L_i,E_i)$ as the entangled register for the local leakage adversary attacking $R_i$.
Recall that the $i$-th tampering map $U_i$ is applied to $(S_i,E_i)=(L_i,R_i,E_i)$, yielding tampered registers $(L_i,R_i,E_i)$.
Given this, we see $L_i$, which satisfies $|L_i|\leq \ell$, as the $\ell$-bounded leakage on the $i$-th share $R_i$ produced by the $i$-th local leakage adversary for $\lrenc$ with access to $R_i$ and the entangled registers $(L_i, E_i)$.
Furthermore, we set the entangled register associated with the share $R_t$ as $W_t=(E_{[t,p]}, \hat{L}, \tilde{L}, L_{[t,p]})$, and for $i>t$ we set the entangled registers $W_i$ to be empty.
Note that from \cref{eq:qssintermediate1} we have that
\begin{equation*}\label{eq:eq1Claim2}
    \theta_{\hat{L} \tilde{L} L_{[p]}\hat{R} R_{[p]} \tilde{R}_{[p]} } =\theta_{\hat{L} \tilde{L} L_{[p]} } \otimes \theta_{\hat{R}  R_{[p]} \tilde{R}_{[p]} },
\end{equation*}
and so the selected shared entangled state stored in registers $W_1,\dots,W_n$ does not depend on the contents of $R$ and $\hat{R}$.
Invoking \cref{thm:2naugmented} and the augmented leakage-resilience property of $(\lrenc,\lrdec)$ under this leakage attack yields \cref{eq:eq2Claim2}, which completes the proof.
\end{proof}

\subsubsection{Proof of \cref{lem:qsstechnical}}\label{sec:qsstechnical}

For convenience, we restate the relevant fact here.
\qsstechnical*
\begin{proof}
We are given a state $\phi_{XABY}$ such that \begin{equation*}
        \phi_{X BY  } \approx_{\eps_1} \phi_{X  } \otimes \phi_{ BY  } \quad \textrm{and} \quad \phi_{YAX   } \approx_{\eps_2} \phi_{Y  } \otimes \phi_{AX} .
    \end{equation*}
    We intend to show that we can generate a state close to $\phi$ using only $\kappa_{X_1 Y_1} \equiv \phi_X \otimes \phi_Y$,  shared entanglement $\gamma$, and isometries $U,V$. 
    We use an independent copy of $\phi$ as shared entanglement, i.e., $\gamma_{XABY} \equiv \phi_{XABY}$. 
    We first correlate registers $\kappa_{X_1 \hat{X}_1}$ and shared entanglement $\gamma_{XABY}$ using the isometry $U$ to obtain state $\phi'$. We next correlate registers $\kappa_{Y_1 \hat{Y}_1}$ and state $\phi'$ using the isometry $V$ to obtain $\phi''$. 
    Then, we show that state $\phi''$ is the final state which will be close to $\phi$ in trace distance, as desired. 
    We use Uhlmann's theorem (\cref{fact:uhlmann}) to argue closeness in trace distance.

Consider the pure state
\begin{equation*}
    \eta = \kappa_{X_1 \hat{X}_1} \otimes \gamma_{XABY}.
\end{equation*}
Using the fact that $\kappa_{X_1}\equiv \phi_X$, we have that $\eta_{X_1BY} \equiv \phi_X \otimes \phi_{BY}$.
Moreover, since $ \phi_X \otimes \phi_{BY} \approx_{\eps_1}  \phi_{XBY}$ by hypothesis, it follows that $\eta_{X_1BY} \approx_{\eps_1}  \phi_{XBY}$.

We invoke Uhlmann's theorem (\cref{fact:uhlmann}) with the following assignment of registers (where the registers on the left are those from the statement of \cref{fact:uhlmann} and the registers on the right are the registers used in this proof):
\begin{equation*}
    (\sigma_A, \rho_A, \sigma_{AC}, \rho_{AB}) \leftarrow ( \eta_{X_1BY} ,  \phi_{XBY}, \eta_{X_1BY \hat{X}_1XA}, \phi_{XBYA}).
\end{equation*}
This guarantees the existence of an isometry $U : \cH_{\hat{X}_1} \otimes \cH_A \otimes \cH_{X} \to \cH_{A}$ such that
\begin{equation}\label{Eq:ulhman1}
   \Delta_B \left({\phi}'_{X_1BYA} , \phi_{XBYA}\right) = \Delta_B( \eta_{X_1BY},  \phi_{XBY}) \leq \sqrt{\eps_1/2}
\end{equation}
with ${\phi}'_{X_1AYB} =  U\left(  \kappa_{X_1 \hat{X}_1} \otimes \gamma_{XABY}\right) U^\dagger$
and where the rightmost inequality follows from \cref{fact:TracevsFidelityvsBures} and the fact that $\eta_{X_1BY} \approx_{\eps_1}  \phi_{XBY}$.

Now, consider the state $\eta' ={\phi}'_{X_1AYB}  \otimes \kappa_{Y_1 \hat{Y}_1}$. 
Observe that 
\begin{equation}\label{Eq:ulhman2}
    \eta'_{X_1AY_1} = \phi'_{X_1A} \otimes \kappa_{Y_1} \equiv \phi'_{X_1A} \otimes \phi_{Y}.
\end{equation} 
Using \cref{fact:TracevsFidelityvsBures} and the fact that $\phi_Y \otimes \phi_{AX}  \approx_{\eps_2}  \phi_{YAX}$ (by hypothesis), we get  
 \begin{equation}\label{Eq:ulhman3}
  \Delta_B( \phi_Y \otimes \phi_{AX} , \phi_{YAX} ) \leq \sqrt{\eps_2/2} .
\end{equation}
From \cref{Eq:ulhman1} and \cref{fact:data}, we get $ \Delta_B( \phi'_{X_1A} , \phi_{XA} ) \leq \sqrt{\eps_1/2}  $. 
This further implies that 
\begin{equation}\label{Eq:ulhman4}
 \Delta_B( \phi'_{X_1A} \otimes \phi_Y , \phi_{XA} \otimes \phi_Y ) \leq \sqrt{\eps_1/2}  .
\end{equation}
Combining \cref{Eq:ulhman2,Eq:ulhman3,Eq:ulhman4} along with the triangle inequality, we have that
\begin{equation*}
    \Delta_B( \eta'_{X_1AY_1},\phi_{XAY}   ) \leq \sqrt{\eps_1/2}  +\sqrt{\eps_2/2}  .
\end{equation*}

 We use \cref{fact:uhlmann}  with the following assignment of registers (below the registers on the left are from ~\cref{fact:uhlmann} and the registers on the right are the registers in this proof),
 $$(\sigma_A, \rho_A, \sigma_{AC}, \rho_{AB}) \leftarrow ( \eta'_{X_1AY_1} ,  \phi_{XAY}, \eta_{X_1AY_1 \hat{Y}_1YB}, \phi_{XAYB}).$$
 From \cref{fact:uhlmann}  we get an isometry  $V : \cH_{\hat{Y}_1} \otimes \cH_B \otimes \cH_{Y} \to \cH_{B}$ such that
\begin{equation}\label{Eq:ulhman5}
  \Delta_B \left(   {\phi}''_{X_1AY_1B} , \phi_{XAYB}\right) = \Delta_B(\phi'_{X_1AY_1},  \phi_{XAY}) \leq \sqrt{\eps_1/2}+\sqrt{\eps_2/2},.
\end{equation}
where,
$$ {\phi}''_{X_1AY_1B} =  V\left( {\phi}'_{X_1AYB}  \otimes \kappa_{Y_1 \hat{Y}_1}\right) V^\dagger = (U \otimes V) ( \kappa \otimes \gamma) (U^\dagger \otimes V^\dagger).$$    
From \cref{Eq:ulhman5} and \cref{fact:TracevsFidelityvsBures}, we obtain the desired statement.
\end{proof}

\subsection{From average-case to worst-case non-malleability}

In \cref{sec:avgnmqss}, we showed that our secret sharing scheme $(\nmshare,\nmrec)$ for quantum messages described in \cref{sec:qnmssscheme} is \emph{average-case} $(\epsnm+2\sqrt{\epslk})$-non-malleable.
To obtain our \cref{thm:mainqnmss}, we would like to upgrade from average-case to worst-case non-malleability.
This can be done in exactly the same manner as was done for our $2$-split-state NMC for quantum messages in \cref{sec:avgtoworst}, and so we avoid repeating the argument here.
More precisely, the following easy extension of \cref{thm:avgtoworst} holds in the secret sharing setting.
\begin{lemma}\label{thm:avgtoworstnmss}
    If $(\share,\rec)$ is an average-case $(\epspriv,\epsnm)$-non-malleable secret sharing scheme for quantum messages of length $b$, then it is also a worst-case $(\epspriv,\epsnm')$-non-malleable secret sharing scheme for quantum messages of length $b$, where $\epsnm'=2^b\cdot \epsnm$.
\end{lemma}

Combining \cref{thm:avgtoworstnmss} with the statistical privacy property shown in \cref{sec:statpriv} and the average-case non-malleability property shown in \cref{sec:avgnmqss} yields the following result.
\begin{theorem}\label{thm:qnmssbeforeparams}
    Given a threshold $t\geq 3$ and a number of parties $p$ such that $t\leq p\leq 2t-1$, it holds that
    the scheme $(\nmshare,\nmrec)$ described in \cref{sec:qnmssscheme} is a $t$-out-of-$p$ $(\epspriv,\epsnm')$-non-malleable secret sharing scheme for quantum messages of length $b$ with $\epsnm'=2^b(\epsnm+2\sqrt{\epslk})$.
\end{theorem}
It remains to set the parameters in \cref{thm:qnmssbeforeparams} (which come from the underlying objects used to define our secret sharing scheme in \cref{sec:qnmssscheme}) appropriately so that we can obtain \cref{thm:mainqnmss}.

\subsection{Setting the parameters}\label{sec:settingparams}

To conclude our argument, we now show how to instantiate the objects used to define our secret sharing scheme in \cref{sec:qnmssscheme} to get \cref{thm:mainqnmss}.
First, due to \cref{thm:22nmss}, the underlying $2$-out-of-$2$ $(\epspriv,\epsnm)$-non-malleable scheme $(2\nmshare,2\nmrec)$ can be set so that $\epspriv=\epsnm=2^{-n^{\Omega(1)}}$ and that the quantum left share and classical right share have lengths $b_1$ and $b_2$, respectively, satisfying $b_1+b_2=n$, provided that the input quantum message has length $b\leq n^c$ for some sufficiently small absolute constant $c>0$.
Since $b_1\leq n$, the scheme $(\qshare,\qrec)$ produces shares of length $\ell=\poly(p,n)$ by \cref{fact:qss}.
Finally, by \cref{thm:2naugmented}, the scheme $(\lrenc,\lrdec)$ can be set so that it is $(\ell=\poly(p,n),\epslk=2^{-n^{\Omega(1)}})$-augmented leakage-resilient with shares of length $\poly(p,n)$.
Additionally, all of these procedures run in time $\poly(p,n)$.
Combining this with \cref{thm:qnmssbeforeparams} yields the desired \cref{thm:mainqnmss}, provided that the input quantum message length $b$ is at most $n^c$ for a sufficiently small absolute constant $c>0$.
We restate the resulting theorem here for convenience.

\qnmss*

\subsection{Quantum-secure classical non-malleable secret sharing schemes}\label{sec:qsqnmss}

In this section, we present and analyze non-malleable threshold secret sharing schemes for classical messages secure against quantum adversaries with shared entanglement. 
Here, both the sharing and reconstruction procedures are classical, and we call these schemes \emph{quantum-secure non-malleable.}
We obtain such schemes via a simpler and entirely analogous realization of the approach laid out earlier in this \cref{sec:nmssquantummsg}, combining a known split-state non-malleable code for classical messages from~\cite{BBJ23} with Shamir's secret sharing scheme and a quantum-secure leakage-resilient secret sharing scheme for classical messages.

We begin by defining threshold quantum-secure non-malleable secret sharing schemes.
This definition follows along the lines of \cref{def:qnmss} specialized to classical messages and with a different definition of the register $\hat{M}$.
Let $\sigma_M$ be a classical message, and set $\hat{M}$ to contain a copy of this message.\footnote{Here, the register $\hat{M}$ is a classical copy of $M$, while in our definition of NMSS schemes for quantum messages $\hat{M}$ is the canonical purification register of $M$.} 
We consider a classical scheme $(\nmshare,\nmrec)$ with a classical sharing procedure $\nmshare : \mathcal{L}(\mathcal{H}_M) \to \mathcal{L}(\mathcal{H}_{S_1} \otimes \mathcal{H}_{S_2} \otimes \cdots \otimes  \mathcal{H}_{S_n})$ and a classical reconstruction procedure $\nmrec : \mathcal{L}( \bigotimes_{i \in T}\mathcal{H}_{S_i} ) \to \mathcal{L}(\mathcal{H}_M)$, where $\cL(\cH)$ is the space of all linear operators in the Hilbert space $\cH$ and the reconstruction procedure $\nmrec$ acts on any authorized set $T$.

Let $\rho = \nmshare(\sigma_{M \hat{M}})$ be the classical sharing of $\sigma_M$, and denote by $S_1, S_2, \ldots, S_p$ the classical registers corresponding to the $p$ resulting shares. 
The most basic property we require of this scheme is correctness, i.e., for any subset of shares $T \subseteq [p]$ such that $ \vert T \vert \geq t$, we must have that
\begin{equation*}
    \nmrec(\rho_{\hat{M}S_{T} }) = \sigma_{M\hat{M}}.
\end{equation*}
A split-state tampering adversary $\cA_T$ for some authorized subset $T$ of size $|T|=t$ is specified as
\begin{equation*}
    \mathcal{A}_T=\left(\bigotimes_{i\in T}U_i, \ket{\psi}_{W_1W_2 \ldots W_p}\right),
\end{equation*}
where $\ket{\psi}_{W_1W_2 \ldots W_p}$ represents the entangled quantum state shared among tampering adversaries and $U_i : \cL(\cH_{S_i}) \otimes \cL(\cH_{W_i})\to  \cL(\cH_{S_i}) \otimes  \otimes \cL(\cH_{W_i})$ are the tampering maps.
 We denote the resulting tampered state by
\begin{equation*}
    \tau^{\cA_T} = \left(\left(\bigotimes_{i\in T}U_i\right)\otimes\left(\bigotimes_{j\not\in T} \id\right)\right) \left(\nmshare( \sigma) \otimes \ketbra{\psi}_{W_1 \cdots W_p} \right) \left(\left(\bigotimes_{i\in T}U^\dagger_i\right)\otimes\left(\bigotimes_{j\not\in T} \id^\dagger\right)\right) .
\end{equation*}

We are now ready to define threshold quantum-secure non-malleable secret sharing schemes.
\begin{definition}[Threshold quantum-secure non-malleable secret sharing scheme\label{def:nmssqcodesclassical}]
The coding scheme $(\nmshare,\nmrec)$ is said to be a \emph{$t$-out-of-$p$ $(\epspriv,\epsnm)$-quantum-secure  
 non-malleable secret sharing scheme}, if for any  classical message ${\sigma_M} \in \cD(\cH_M)$ (with copy $\sigma_{\hat{M}}$ in  $\sigma_{M\hat{M}}$) the following properties are satisfied:
\begin{itemize}
    \item \textbf{Correctness:} For any $T \subseteq [p]$ such that $ \vert T \vert \geq t$  it holds that
    \begin{equation*}
        \nmrec (\nmshare( \sigma_{M \hat{M}})_{S_{T}}) = \sigma_{M\hat{M}},
    \end{equation*}
    where we write $S_T=(S_i)_{i\in T}$.
    
    \item \textbf{Statistical privacy:} For any $T \subseteq [p]$ such that $ \vert T\vert \leq t-1$ it holds that
    \begin{equation*}
        \nmshare( \sigma_{M\hat{M}})_{\hat{M}S_{T}}  \approx_{\epspriv} \sigma_{\hat{M}} \otimes  \zeta_{S_{T}},
    \end{equation*}
    where $\zeta_{S_{T}}$ is a fixed state independent of $\sigma_M$.
    
    \item \textbf{Non-malleability:} For any $T \subseteq [p]$ such that $ \vert T \vert = t$ and for every split-state tampering adversary $\mathcal{A}_{T}$ as above, it holds that
    \begin{equation}\label{eq:qsnmssgoal}
        \nmrec(\tau^{\cA_T}_{\hat{M}S_{T} }) \approx_{\epsnm}  
    p_{\cA_T} \cdot \sigma_{\hat{M}M} + (1-p_{\cA_T}) \sigma_{\hat{M}} \otimes \gamma^{\cA_T}_{M},
    \end{equation}
    where $p_{\cA_T}\in[0,1]$ and $\gamma^{\cA_T}_{M}$ depend only on the adversary $\cA_T$.

    If \cref{eq:qsnmssgoal} is only guaranteed to hold when $\sigma_M=U_M$, we say that $(\nmshare,\nmrec)$ is an \emph{average-case} $t$-out-of-$p$ $(\epspriv, \epsnm)$-quantum-secure non-malleable secret sharing scheme.

\end{itemize}

\end{definition}

We will require the following known quantum-secure $2$-out-of-$2$ non-malleable secret sharing scheme from~\cite{BBJ23}, along with standard Shamir's threshold secret sharing scheme.
\begin{theorem}[Average-case $2$-out-of-$2$ quantum-secure non-malleable secret sharing scheme with constant rate~\cite{BBJ23}]\label{thm:bbj1}
For any fixed constant $c>0$, there exists an average-case $2$-out-of-$2$ $(\epspriv, \epsnm)$-quantum-secure non-malleable secret sharing scheme for classical messages with codeword length $n$, message length at most $\left(\frac{1}{5}-c\right)n$, $\epspriv = 2^{-n^{\Omega_c(1)}}$ and  $\epsnm = 2^{-n^{\Omega_c(1)}}$, where $\Omega_c(\cdot)$ hides constants which depend only on $c$.
    Furthermore, the sharing and reconstruction procedures can be computed in time $\poly(n)$.
\end{theorem}

\suppress{
\begin{theorem}[Worst-case $2$-split-state quantum-secure NMC~\cite{BBJ23}]\label{thm:bbj2}
    There exists a constant $c\in(0,1)$ such that the following holds:
    There exists a quantum-secure $\eps$-non-malleable code with codeword length $n$ (two shares of sizes $n/5$ , $4n/5$ respectively), message length at most $n^c$, and error $\eps = 2^{-n^{\Omega(1)}}$.
    Furthermore, the encoding and decoding procedures of this code can be computed in time $\poly(n)$.
\end{theorem}}

\begin{fact}[Classical Shamir's secret sharing~\cite{Sha79}]\label{fact:css}
    For any number of parties $p$ and threshold $t$ such that $t\leq p$, there exists a $t$-out-of-$p$ secret sharing scheme $(\share,\rec)$ for classical messages of length $b$ with share size at most $\max(p,b)$, where both the sharing and reconstruction procedures run in time $\poly(p,b)$. 
\end{fact}

\subsubsection{Our candidate quantum-secure non-malleable secret sharing scheme}\label{sec:qnmssschemeqs}

We leverage the $2$-out-of-$2$ quantum-secure non-malleable secret sharing scheme from~\cref{thm:bbj1} and our quantum-secure augmented leakage-resilient secret sharing scheme for classical messages from \cref{thm:2naugmented} to construct quantum-secure non-malleable threshold secret sharing schemes.
As mentioned before, we accomplish this via the same techniques we used to construct our NMSS schemes for quantum messages earlier in this section.

Fix a threshold $t\geq 3$ and a number of parties $p$ such that $t\leq p$.
We will require the following objects:
\begin{itemize}
    \item A $2$-out-of-$2$ $(\epsnm=\eps,\epspriv=\eps)$-quantum-secure non-malleable secret sharing scheme, which we denote by $(2\nmenc,2\nmdec)$, for messages of length $b$ with a left share of length $b_1$ and a right share of length $b_2$, guaranteed by \cref{thm:bbj1};

    \item A $t$-out-of-$p$ secret sharing scheme $(\share,\rec)$ for messages of length $b_1$ with shares of length at most $\ell$, guaranteed by \cref{fact:css};

    \item A $2$-out-of-$p$ $(\ell,\epslk)$-quantum-secure 
    augmented leakage-resilient secret sharing scheme, which we denote by $(\lrenc,\lrdec)$, for classical messages of length $b_2$, guaranteed by \cref{thm:2naugmented}.
\end{itemize}

We proceed to describe our candidate $t$-out-of-$p$ scheme $(\nmshare,\nmrec)$.
On input a classical message $\sigma_M$ (with a copy $\hat{M}$ in $\sigma_{M\hat{M}}$), the classical sharing algorithm $\nmshare(\sigma)$ proceeds as follows:
\begin{enumerate}
    \item Compute the encoding $\rho_{LR}=2\nmshare(\sigma_M)$;

     \item Apply $\share$ to the contents of register $L$ to obtain $p$  shares stored in registers $L_1,\dots,L_p$;

    \item Apply $\lrenc$ to the contents of register $R$ to obtain $p$  shares stored in registers $R_1,\dots,R_p$;

    \item Form the $i$-th final share $S_i=(L_i,R_i)$.
\end{enumerate}
The corresponding reconstruction algorithm $\nmrec$ is straightforward.
The differences between our quantum-secure NMSS scheme considered here and the NMSS scheme for quantum messages from \cref{sec:qnmssscheme} are the following:
\begin{itemize}
    \item The NMSS scheme for quantum messages from \cref{sec:qnmssscheme} uses $(2\nmshare, 2\nmrec)$ guaranteed by \cref{thm:22nmss}, while we use $(2\nmshare, 2\nmrec)$ guaranteed by \cref{thm:bbj1}.
    
    \item The NMSS scheme for quantum messages from \cref{sec:qnmssscheme} uses quantum Shamir's secret sharing $(\qshare, \qrec)$ guaranteed by \cref{fact:qss}, while we use classical Shamir's secret sharing $(\share, \rec)$ guaranteed by \cref{fact:css}. 
    This ensures that for classical messages we get $t$-out-of-$p$ NMSS schemes for any $t \leq p$ and $t \geq 3$.
    
    \item In the non-malleability argument for our NMSS scheme for quantum messages in \cref{sec:avgnmqss}, we have used the $\mathsf{Uhlmann}$ isometry to create the coherent copies of quantum registers. 
    For the non-malleability argument of our quantum-secure NMSS scheme we can completely ignore such an argument and just copy the classical registers and proceed with the rest of argument.
\end{itemize}

We have the following result. 
\begin{restatable}[Average-case $t$-out-of-$p$ quantum-secure non-malleable secret sharing scheme]{theorem}{avg-quantum-sec-nmss}\label{thm:mainqnmssqs1}
The coding scheme 
$(\nmshare,\nmrec)$ as described in \cref{sec:qnmssschemeqs} is an average-case $(\epspriv,\epsnm+2\sqrt{\epslk})$-quantum-secure 
 non-malleable secret sharing scheme. 
\end{restatable}
\begin{proof}
The correctness and runtime of this coding scheme are clear. 
Establishing the statistical privacy and average-case non-malleability of our scheme follows analogously to the arguments in~\cref{sec:statpriv} and \cref{sec:avgnmqss}, respectively, and we opt to not repeat them.
\end{proof}

As before, it remains to upgrade average-case non-malleability to worst-case non-malleability.
This is also generically possible in the present setting, as captured by the following lemma whose proof is analogous to those of \cref{thm:avgtoworst,thm:avgtoworstnmss}.
\begin{lemma}\label{thm:avgtoworstnmssqs}
    If $(\share,\rec)$ is an average-case $(\epspriv,\epsnm)$-quantum-secure 
 non-malleable secret sharing scheme for messages of length $b$, then it is also a worst-case  $(\epspriv,\epsnm')$-quantum-secure 
non-malleable secret sharing scheme for messages of length $b$, where $\epsnm'=2^b\cdot \epsnm$.
\end{lemma}

Finally, the choices of parameters used to instantiate our construction in this section are analogous to those of~\cref{sec:settingparams}. 
We obtain the following final result by combining \cref{thm:mainqnmssqs1} with \cref{thm:avgtoworstnmssqs}.
\begin{restatable}[Threshold quantum-secure NMSS schemes]{theorem}{quantum-sec-nmss}\label{thm:mainqnmssqs}
    There exists a constant $c\in(0,1)$ such that the following holds: Given a threshold $t\geq 3$ and a number of parties $p$ such that $t\leq p$, there exists a $(\epspriv=\eps,\epsnm=\eps)$-$t$-out-of-$p$ quantum-secure non-malleable secret sharing scheme for messages with shares of size at most $\poly(p,n)$, any message length at most $n^c$, and $\eps=2^{-n^{\Omega(1)}}$.
    Furthermore, the sharing and reconstruction procedures of this scheme can be computed in time $\poly(p,n)$.
\end{restatable}

\bibliography{References}
\bibliographystyle{alpha}

\end{document}